\documentclass[a4paper,11pt,fleqn]{article} 

\usepackage{amsmath,amssymb,amsthm}

\allowdisplaybreaks

\makeatletter
\long\def\@makecaption#1#2{%
  \vskip\abovecaptionskip\footnotesize
  \sbox\@tempboxa{#1. #2}%
  \ifdim \wd\@tempboxa >\hsize
    #1. #2\par
  \else
    \global \@minipagefalse
    \hb@xt@\hsize{\hfil\box\@tempboxa\hfil}%
  \fi
  \vskip\belowcaptionskip}
\makeatother

\newcommand{\todo}[1][\null]{\ensuremath{\clubsuit}}

\newcommand{\noprint}[1]{}
\newcommand{\checked}[1][\null]{\ensuremath{\boldsymbol{\surd}}}

\newcommand{\p}{\partial}
\newcommand{\const}{\mathop{\rm const}\nolimits}
\newcommand{\sgn}{\mathop{\rm sgn}\nolimits}
\newcommand{\spanindex}{{\mbox{\tiny$\langle\,\rangle$}}}
\newcommand{\EqOrd}{r}

\newtheorem{theorem}{Theorem}
\newtheorem{lemma}[theorem]{Lemma}
\newtheorem{corollary}[theorem]{Corollary}
\newtheorem{proposition}[theorem]{Proposition}
\newtheorem*{problem*}{Problem}
{\theoremstyle{definition}
\newtheorem{definition}[theorem]{Definition}

\newtheorem{remark}[theorem]{Remark}
\newtheorem*{remark*}{Remark}
}

\begin{document}

\par\noindent {\LARGE\bf
Group analysis of general Burgers--Korteweg--de Vries equations
\par}

{\vspace{4mm}\par\noindent {\large Stanislav Opanasenko$^{\dag\ddag}$, Alexander Bihlo$^\dag$ and Roman O.\ Popovych$^{\ddag\S}$,
} \par\vspace{2mm}\par}

{\vspace{2mm}\par\noindent {\it
$^{\dag}$~Department of Mathematics and Statistics, Memorial University of Newfoundland,\\
$\phantom{^{\dag}}$~St.\ John's (NL) A1C 5S7, Canada\par
}}
{\vspace{2mm}\par\noindent {\it
$^\ddag$~Institute of Mathematics of NAS of Ukraine, 3 Tereshchenkivska Str., 01004 Kyiv, Ukraine\par
}}
{\vspace{2mm}\par\noindent {\it
$^{\S}$~Wolfgang Pauli Institut, Oskar-Morgenstern-Platz 1, 1090 Wien, Austria\\
$\phantom{^{\S}}$~Mathematical Institute, Silesian University in Opava, Na Rybn\'\i{}\v{c}ku 1, 746 01 Opava, Czechia
}}

{\vspace{2mm}\par\noindent {\it
\textup{E-mail:} sopanasenko@mun.ca, abihlo@mun.ca, rop@imath.kiev.ua
}\par}

\vspace{8mm}\par\noindent\hspace*{10mm}\parbox{140mm}{\small
The complete group classification problem for the class 
of $(1{+}1)$-dimensional $\EqOrd$th order general variable-coefficient Burgers--Korteweg--de Vries equations
is solved for arbitrary values of $\EqOrd$ greater than or equal to two.
We find the equivalence groupoids of this class and 
its various subclasses obtained by gauging equation coefficients with equivalence transformations. 
Showing that this class and certain gauged subclasses are normalized in the usual sense, 
we reduce the complete group classification problem for the entire class 
to that for the selected maximally gauged subclass, 
and it is the latter problem that is solved efficiently using the algebraic method of group classification. 
Similar studies are carried out for the two subclasses of equations with coefficients 
depending at most on the time or space variable, respectively. 
Applying an original technique, we classify Lie reductions of equations from the class under consideration 
with respect to its equivalence group. 
Studying alternative gauges for equation coefficients with equivalence transformations 
allows us not only to justify the choice of the most appropriate gauge for the group classification 
but also to construct for the first time classes of differential equations 
with nontrivial generalized equivalence group 
such that equivalence-transformation components corresponding to equation variables 
locally depend on nonconstant arbitrary elements of the class. 
For the subclass of equations with coefficients depending at most on the time variable, 
which is normalized in the extended generalized sense, 
we explicitly construct its extended generalized equivalence group in a rigorous way. 
The new notion of effective generalized equivalence group is introduced.
}\par\vspace{4mm}

\noprint{
}


\section{Introduction}

The development of new powerful tools and techniques of group analysis of differential equations
during the last decade has essentially extended the range of effectively solvable problems 
of this branch of mathematics. 
In particular, it became possible to study admissible transformations and Lie symmetries for 
systems of differential equations from complex classes parameterized by a few functions of several arguments.    
 
A number of evolution equations that are important in mathematical physics are of the general form
\begin{equation}\label{eq:GenBurgersKdVEqs}
 u_t+C(t,x)uu_x=\sum_{k=0}^\EqOrd A^k(t,x)u_k+B(t,x).
\end{equation}
Here and in the following the integer parameter~$\EqOrd$ is fixed, and $\EqOrd\geqslant2$.
We require the condition $CA^\EqOrd\ne0$ guaranteeing 
that equations from the class~\eqref{eq:GenBurgersKdVEqs} are nonlinear and of genuine order~$\EqOrd$. 
Throughout the paper we use the standard index derivative notation $u_t=\partial u/\partial t$, $u_k=\partial^k u/\partial x^k$, 
and also $u_0=u$, $u_x=u_1$, $u_{xx}=u_2$ and $u_{xxx}=u_3$. 

The class~\eqref{eq:GenBurgersKdVEqs} contains a number of the prominent classical models of fluid mechanics, 
including
\[
\begin{array}{ll}
  u_t+uu_x=a_2u_{xx}                      &  \mbox{Burgers equation},                  \\[.5ex]
  u_t+uu_x=a_3u_{xxx}                     &  \mbox{Korteweg--de Vries (KdV) equation}, \\[.5ex]
  u_t+uu_x=a_4u_4+a_3u_3+a_2u_2\qquad     &  \mbox{Kuramoto--Sivashinsky equation},    \\[.5ex]
  u_t+uu_x=a_5u_5+a_3u_3                  &  \mbox{Kawahara equation},                 \\[.5ex]
  u_t+uu_x=a_\EqOrd u_\EqOrd              &  \mbox{generalized Burgers--KdV equation}, 
\end{array}
\]
where $a$'s are constants, and the coefficient of the highest-order derivative is nonzero.

Due to the importance of equations from the class~\eqref{eq:GenBurgersKdVEqs}, 
there exist already a lot of papers in which particular equations of the form~\eqref{eq:GenBurgersKdVEqs} 
have been considered in light of their symmetries, integrability, exact solutions, etc. 
Here we review only papers on admissible transformations and group classification 
of classes related to the class~\eqref{eq:GenBurgersKdVEqs}.

Particular subclasses of the class~\eqref{eq:GenBurgersKdVEqs} with small values of~$\EqOrd$
were the subject of a number of papers published over the past twenty-five years. 
Thus, the equivalence groupoids of the class of variable-coefficient Burgers equations of the form $u_t+uu_x+f(t,x)u_{xx}=0$ with $f\ne0$ 
and of its subclass of equations with $f=f(t)$ were constructed by Kingston and Sophocleous in~\cite{king91Ay} 
as sets of point transformations in pairs of equations.
In fact, it was implicitly proved there that these classes are normalized in the usual sense. 
Later such transformations were called allowed~\cite{gaze92a,wint92Ay} or form-preserving~\cite{king98Ay},  
and they had preceded the notion of admissible transformations, 
which is of central importance in group analysis of classes of differential equations. 
Solving of the group classification problem for the above subclass with $f=f(t)$ was initiated in~\cite{doyl90a,wafo04a} 
and completed in~\cite{poch14a,vane15a}. 
Most recently, an extended symmetry analysis of the class with $f=f(t,x)$ considered in~\cite{king91Ay} 
was comprehensively carried out in~\cite{poch16a}. 
The partial preliminary group classification problem for the related class of Burgers equations with sources, 
which are of the form $u_t+uu_x=u_{xx}+f(t,x,u)$, with respect to the maximal Lie symmetry group of the Burgers equation 
was considered in~\cite{long17a}.

In~\cite{gaze92a,wint92Ay} allowed transformations were computed for 
the class of variable-coefficient KdV equations of the form $u_t+f(t,x)uu_x+g(t,x)u_{xxx}=0$ with $fg\ne0$
and were then used in~\cite{gaze92a} to carry out the group classification of this class;  
see~\cite{vane16a} for a modern interpretation of these results. 
An attempt to reproduce these results for the class of variable-coefficient Burgers equations of the form $u_t+f(t,x)uu_x+g(t,x)u_{xx}=0$ with $fg\ne0$
was made in~\cite{qu95a} supposing 
that admissible transformations of this class are similar to admissible transformations of its third-order counterpart
but in fact the structure of the corresponding equivalence groupoid is totally different from that in~\cite{gaze92a}. 
Transformational properties of the class of variable-coefficient KdV equations of the form
\[u_t+f(t)uu_x+g(t)u_{xxx}+(q(t)x+p(t))u_x+h(t)u+k(t)x+l(t)=0, \quad fg\ne0,\]
were studied in~\cite{popo10By}; see also~\cite{vane13a}. 
This class coincides with the class~$\mathcal K_1$  (with the particular value $\EqOrd=3$)
arising in Section~\ref{sec:GenBurgersKdVEqsWithTime-DependentCoeffs} of the present paper. 
The class~$\mathcal K_1^{\EqOrd=3}$ was proved to be normalized in the usual sense 
and mapped by a family of its equivalence transformations, 
e.g., to its subclass of equations of the form $u_t+uu_x+g(t)u_{xxx}=0$  with $g\ne0$, 
which is also normalized in the usual sense, 
and solving the group classification problem in the class~$\mathcal K_1^{\EqOrd=3}$ 
is equivalent to that in the subclass; 
cf.\ Proposition~\ref{pro:GroupClassificationsOfGenBurgersKdVEqsAndGaugedGenBurgersKdVEqs} 
and Section~\ref{sec:GenBurgersKdVEqsWithTime-DependentCoeffs}. 
Variable-coefficient generalizations of the Kawahara equation 
were studied in a similar way in~\cite{kuri14b,kuri14a}. 
The group classification of Galilei-invariant equations of the form $u_t+uu_x=F(u_\EqOrd)$ 
was carried out in~\cite{boyk01a,fush1996b}.

Unfortunately, the results obtained in many papers were incomplete or even faulty; 
for a partial listing of papers with such results on variable-coefficient Korteweg--de Vries equations, we refer to~\cite{popo10By}. 
It is thus appropriate to solve the group classification problem for the general class~\eqref{eq:GenBurgersKdVEqs} systematically for the first time.

More general classes of evolutions equations, which include the class~\eqref{eq:GenBurgersKdVEqs}, 
were also considered in the literature. 
Contact symmetries of $(1{+}1)$-dimensional evolution equations were studied by Magadeev~\cite{maga93Ay}.
More specifically, Magadeev proved that a $(1{+}1)$-dimensional evolution equation admits 
an infinite-dimensional Lie algebra of contact symmetries if and only if 
it is linearizable by a contact transformation. 
He also classified, up to contact equivalence transformations of evolution equations, 
realizations of finite-dimensional Lie algebras by contact vector fields 
with the independent variables $(t,x)$ and the dependent variable~$u$
each of which is the maximal contact invariance algebras of a non-linearizable evolution equation.
At the same time, evolution equations admitting these realizations 
as maximal contact invariance algebra were not constructed. 
Moreover, singling out the classification of contact symmetries of equations 
from a particular subclass of the entire class of $(1{+}1)$-dimensional evolution equations 
is in general a more challenging problem than the direct classification 
of contact symmetries for equations from the subclass.   
A claim similar to the above one on results of~\cite{maga93Ay} 
is also true for classifications of low-order evolution equations (with $\EqOrd=2,3,4$)
considered, e.g., in~\cite{basa01Ay,gung04a,huan09a,zhda99Ay}.

In the present paper we solve the complete group classification problem of the class~\eqref{eq:GenBurgersKdVEqs} 
for any fixed value of~$\EqOrd\geqslant2$. 
An even more important result of the paper is 
the construction, in the course of the study of admissible transformations 
within subclasses of the class~\eqref{eq:GenBurgersKdVEqs}, 
of several examples of classes with specific properties, 
the existence of which was in doubt for a long time. 
These examples give us unexpected insights into the general theory of equivalence transformations 
within classes of differential equations. 
In particular, they justify introducing the notion of effective generalized equivalence group.  

We also point out that the class of linear (in general, variable-coefficient) evolution equations 
is obtained from~\eqref{eq:GenBurgersKdVEqs} by setting $C=0$. 
This class has transformational properties different from those of the class~\eqref{eq:GenBurgersKdVEqs}, 
and equations with $C=0$ and $C\ne0$ are not related to each other by point or contact transformations. 
Moreover, the class of linear equations was completely classified 
in~\cite{lie81Ay} and~\cite[pp.~114--118]{ovsi82Ay} for $\EqOrd=2$ and in~\cite{bihl16a} for $\EqOrd>2$. 
(The last paper enhanced and extended results of~\cite[Section~III]{gung04a} on $\EqOrd=3$ and of~\cite{huan12a} on $\EqOrd=4$.)
This is why it is appropriate to exclude linear equations from the present consideration.

The further structure of this paper is the following. 
In Section~\ref{sec:AlgebraicGroupClassificationMethodGenBurgersKdVEqs} 
we present sufficient background information 
on point transformations in classes of differential equations and 
on the algebraic method of group classification. 
Since the class~\eqref{eq:GenBurgersKdVEqs} is normalized in the usual sense, 
particular emphasis is given to describing the concept of normalization of a class of differential equations. 
The notion of effective generalized equivalence group is introduced there for the first time. 
Using known results on admissible transformations of the entire class of $r$th order evolution equations and its wide subclasses, 
in Section~\ref{sec:EquivalenceGroupoidGenBurgersKdVEqs} we compute the equivalence groupoids for the class~\eqref{eq:GenBurgersKdVEqs} 
and its two subclasses singled out by the arbitrary-element gauges $C=1$ and $(C,A^1)=(1,0)$, 
which are realized by families of equivalence transformations. 
Due to each of the above classes being normalized in the usual sense, its equivalence groupoid 
coincides with the subgroupoid induced by the usual equivalence group of this class. 
The solution of the group classification problem for the class~\eqref{eq:GenBurgersKdVEqs} is shown 
to reduce to that for its gauged subclass associated with the constraint $(C,A^1)=(1,0)$ 
and referred to in the paper as the class~\eqref{eq:GenBurgersKdVEqsGaugedSubclass}.
Moreover, this subclass turns out to be the most convenient for carrying out the group classification 
since it is normalized in the usual sense and maximally gauged. 
The consideration of alternative arbitrary-element gauges in Section~\ref{sec:AlternativeGauges} 
additionally justifies the selection of the subclass with $(C,A^1)=(1,0)$ for group classification. 
This also allows us to construct for the first time examples of classes with nontrivial generalized equivalence groups, 
where transformation components for variables locally depend on nonconstant arbitrary elements. 
The determining equations for Lie symmetries of equations from the class~\eqref{eq:GenBurgersKdVEqs} 
are derived and preliminarily studied in Section~\ref{sec:DetEqsLieSymmetriesGenBurgersKdVEqs}. 
These results are used in Section~\ref{sec:PropertiesOfAppropriateSubalgebrasGenBurgersKdVEqs} 
for analyzing properties of appropriate (for group classification) subalgebras 
of the projection of the equivalence algebra of the class~\eqref{eq:GenBurgersKdVEqsGaugedSubclass}
to the space of equation variables.  
The group classification of the gauged subclass~\eqref{eq:GenBurgersKdVEqsGaugedSubclass} 
and thus the entire class~\eqref{eq:GenBurgersKdVEqs}
is completed in Section~\ref{sec:GroupClassificationGenBurgersKdVEqs} 
via classifying the appropriate subalgebras and finding the corresponding values of the arbitrary-element tuple. 
In Section~\ref{sec:AlternativeClassificationCases}, we discuss the optimality of chosen inequivalent representatives for cases of Lie symmetry extensions.
The subclass~$\mathcal K_3$ of equations from  the class~\eqref{eq:GenBurgersKdVEqs} with coefficients depending at most on~$t$ is 
the object of the study in Sections~\ref{sec:GenBurgersKdVEqsWithTime-DependentCoeffs}. 
We exhaustively describe the equivalence groupoid of this subclass, gauge its arbitrary elements by equivalence transformations 
and single out a complete list of inequivalent Lie symmetry extensions within this subclass from that for the class~\eqref{eq:GenBurgersKdVEqs}. 
In fact, we begin with the wider subclass~$\mathcal K_1$, where $A^1$ and~$B$ may affinely depend on~$x$, 
since this subclass is normalized in the usual sense with respect to its nice usual equivalence group 
whereas the subclass~$\mathcal K_2$, where only~$B$ may affinely depend on~$x$, and the subclass~$\mathcal K_3$ 
give new nontrivial examples of classes normalized only in the generalized sense and only in the extended generalized sense, respectively. 
Similar results for the subclass~$\mathcal F_1$ of equations 
from  the class~\eqref{eq:GenBurgersKdVEqs} with coefficients depending at most on~$x$ 
are obtained in Section~\ref{sec:GenBurgersKdVEqsWithSpace-DependentCoeffs}. 
Since this subclass is not normalized in any appropriate sense, 
in order to describe its equivalence groupoid~$\mathcal G^\sim_{\mathcal F_1}$, 
it is necessary to solve a complicated classification problem
for values of the arbitrary-element tuple that admit nontrivial admissible transformations, 
which are not generated by transformations from the corresponding usual equivalence group. 
It turns out that the equivalence groupoid~$\mathcal G^\sim_{\mathcal F_1}$ 
has an interesting structure related to Lie symmetry extensions in the subclass~$\mathcal F_1$ 
although the subclass~$\mathcal F_1$ is far from even being semi-normalized.
Lie reductions of equations from the class~\eqref{eq:GenBurgersKdVEqs} are classified in Section~\ref{sec:LieReductionsGenBurgersKdVEqs}
with respect to the usual equivalence group of this class using the fact that it is normalized in the usual sense. 
Other possibilities for finding exact solutions of these equations are also discussed. 
The conclusions of the paper are presented in the final Section~\ref{sec:ConclusionGenBurgersKdVEqs}.
\looseness=-1

\section{Algebraic method of group classification}\label{sec:AlgebraicGroupClassificationMethodGenBurgersKdVEqs}

Basic notions and results underlying
the algebraic method of group classification in its modern advanced form as will be used below 
were extensively discussed in~\cite{bihl11Dy,bihl16a,kuru16a,popo06Ay,popo10Cy,popo10Ay}, 
to which we refer for further details. 
Examples of applying various versions of the algebraic method to particular group classification problems 
can be found in~\cite{basa01Ay,gagn93a,gaze92a,gung04a,huan12a,maga93Ay,mkhi15a,zhda99Ay}. 
In this section we not only review 
the needed part of the known theory on symmetry analysis in classes of differential equations 
but also present some new notions and results of this field. 

Let $\mathcal L_\theta$ denote a system of differential equations of the form $L(x,u^{(\EqOrd)},\theta(x,u^{(\EqOrd)}))=0$. 
Here, $x=(x_1,\dots,x_n)$ are the~$n$ independent variables, $u=(u^1,\dots,u^m)$ are the~$m$ dependent variables, and~$L$ is a tuple of differential functions in~$u$. 
We use the standard short-hand notation~$u^{(\EqOrd)}$ to denote the tuple of derivatives of~$u$ with respect to $x$ up to order~$\EqOrd$, 
which also includes the~$u$'s as the derivatives of order zero. 
The system~$\mathcal L_\theta$ is parameterized by the tuple of functions~$\theta=(\theta^1(x,u^{(\EqOrd)}),\dots,\theta^k(x,u^{(\EqOrd)}))$, called the arbitrary elements, 
which runs through the solution set~$\mathcal S$ of an auxiliary system 
of differential equations and inequalities in~$\theta$, 
$S(x,u^{(\EqOrd)},\theta^{(q)}(x,u^{(\EqOrd)}))=0$ and, e.g., $\Sigma(x,u^{(\EqOrd)},\theta^{(q)}(x,u^{(\EqOrd)}))\ne0$. 
Here, the notation $\theta^{(q)}$ encompasses the partial derivatives of the arbitrary elements~$\theta$ up to order $q$ with respect to both~$x$ and~$u^{(\EqOrd)}$. 
Thus, the \textit{class of (systems of) differential equations} $\mathcal L|_{\mathcal S}$ is the parameterized family of systems~$\mathcal L_\theta$'s,
such that~$\theta$ lies in~$\mathcal S$.

For the specific class of general Burgers--KdV equations~\eqref{eq:GenBurgersKdVEqs} considered below, 
we have $n=2$, $m=1$, and, in the more traditional notation, $x_1=t$ and $x_2=x$.
The tuple of arbitrary elements is $\theta=(A^0,\dots,A^\EqOrd,B,C)$, 
which runs through the solution set of the auxiliary system 
\[
 A^k_{u_\alpha}=0,\quad k=0,\dots,\EqOrd,\quad B_{u_\alpha}=0,\quad C_{u_\alpha}=0,\quad |\alpha|\leqslant \EqOrd,\quad CA^\EqOrd\ne0,
\]
where $\alpha=(\alpha_1,\alpha_2)$ is a multi-index, $\alpha_1,\alpha_2\in\mathbb N\cup\{0\}$, $|\alpha|=\alpha_1+\alpha_2$, and
$u_{\alpha}=\partial^{|\alpha|}u/\partial t^{\alpha_1}\partial x^{\alpha_2}$. 
Satisfying the auxiliary differential equations is equivalent to the fact 
that the arbitrary elements do not depend on derivatives of~$u$. 
The inequality $A^\EqOrd C\ne0$ ensures 
that equations from the class~\eqref{eq:GenBurgersKdVEqs} are both nonlinear and of order~$\EqOrd$.

Group classification of differential equations is based on studying how systems from a given class are mapped to each other. 
This study is formalized in the notion of \textit{admissible transformations}, 
which constitute the \textit{equivalence groupoid} of the class~$\mathcal L|_{\mathcal S}$. 
An admissible transformation is a triple $(\theta,\tilde\theta,\varphi)$, 
where $\theta,\tilde\theta\in\mathcal S$ are arbitrary-element tuples 
associated with systems $\mathcal L_\theta$ and~$\mathcal L_{\tilde\theta}$ from the class~$\mathcal L_{\mathcal S}$ 
that are similar to each other, 
and $\varphi$ is a point transformation in the space of~$(x,u)$ that maps~$\mathcal L_\theta$ to~$\mathcal L_{\tilde \theta}$.

A related notion of relevance in the group classification of differential equations is that of \textit{equivalence transformations}. 
Usual equivalence transformations are point transformations 
in the joint space of independent variables, derivatives of~$u$ up to order~$\EqOrd$ and arbitrary elements 
that are projectable to the space of~$(x,u^{(\EqOrd')})$ for each $\EqOrd'=0,\dots,\EqOrd$, 
respect the contact structure of the $\EqOrd$th order jet space coordinatized by the $\EqOrd$-jets $(x,u^{(\EqOrd)})$ 
and map every system from the class~$\mathcal L|_{\mathcal S}$ to a system from the same class. 
The Lie (pseudo)group constituted by the equivalence transformations of~$\mathcal L|_{\mathcal S}$ 
is called the \emph{usual equivalence group} of this class and denoted by~$G^\sim$. 
If the arbitrary elements depend at most on derivatives of~$u$ up to order $\hat\EqOrd<\EqOrd$, 
then one can assume that equivalence transformations act in the space of $(x,u^{(\hat\EqOrd)},\theta)$
instead of the space of~$(x,u^{(\EqOrd)},\theta)$.

The usual equivalence group~$G^\sim$ gives rise to a subgroupoid of the equivalence groupoid~$\mathcal G^\sim$
since each equivalence transformation~$\mathcal T\in G^\sim$ generates 
a family of admissible transformations parameterized by~$\theta$,
\[
G^\sim\ni\mathcal T\rightarrow\big\{(\theta,\mathcal T\theta,\pi_*\mathcal T)\mid \theta\in\mathcal S\big\}\subset\mathcal G^\sim.
\]
Here $\pi$ denotes the projection of the space of $(x,u^{(\EqOrd)},\theta)$ to the space of equation variables only, $\pi(x,u^{(\EqOrd)},\theta)=(x,u)$. 
The pushforward $\pi_*\mathcal T$ of~$\mathcal T$ by $\pi$ is then just the restriction of~$\mathcal T$ to the space of~$(x,u)$.

The \textit{usual equivalence algebra}~$\mathfrak g^\sim$ of the class~$\mathcal L|_{\mathcal S}$ 
is an algebra associated with the usual equivalence group~$G^\sim$ 
and constituted by the infinitesimal generators of one-parameter groups of equivalence transformations. 
These infinitesimal generators are vector fields in the joint space of $(x,u^{(\EqOrd)},\theta)$ 
that are projectable to~$(x,u^{(\EqOrd')})$ for each $\EqOrd'=0,\dots,\EqOrd$. 
Since equivalence transformations respect the contact structure on the $\EqOrd$th order jet space, 
the vector fields from~$\mathfrak g^\sim$ inherit this compatibility, 
meaning their projections to the space of~$(x,u^{(\EqOrd)})$ 
coincide with the $\EqOrd$th order prolongations of the associated projections to the space of~$(x,u)$.

In the case when the arbitrary elements~$\theta$'s are functions of~$(x,u)$ only, 
we can assume that equivalence transformations of the class~$\mathcal L|_{\mathcal S}$ 
are point transformations of $(x,u,\theta)$ 
mapping every system from the class~$\mathcal L|_{\mathcal S}$ to a system from the same class.
The projectability property for equivalence transformations is neglected here. 
Then these equivalence transformations constitute a Lie (pseudo)group~$\bar G^{\sim}$
called the \emph{generalized equivalence group} of the class~$\mathcal L|_{\mathcal S}$. 
See the first discussion of this notion in~\cite{mele94Ay,mele96Ay} with no relevant examples 
and the further development in~\cite{popo06Ay,popo10Ay}. 
Often the generalized equivalence group coincides with the usual one; 
this situation is considered as trivial. 
Similar to usual equivalence transformations, 
each element of~$\bar G^{\sim}$ generates a family of admissible transformations parameterized by~$\theta$,
\[
\bar G^\sim\ni\mathcal T\rightarrow\big\{(\theta',\mathcal T\theta',\pi_*(\mathcal T|_{\theta=\theta'(x,u)}))\mid \theta'\in\mathcal S\big\}\subset\mathcal G^\sim,
\] 
and thus the generalized equivalence group~$\bar G^{\sim}$ also generates 
a subgroupoid~$\bar{\mathcal H}$ of the equivalence groupoid~$\mathcal G^\sim$. 

\begin{definition}
We call any minimal subgroup of~$\bar G^{\sim}$ 
that generates the same subgroupoid of~$\mathcal G^\sim$ 
as the entire group~$\bar G^{\sim}$ does 
an \emph{effective generalized equivalence group} of the class~$\mathcal L|_{\mathcal S}$.
\end{definition}

The uniqueness of an effective generalized equivalence group is obvious 
if the entire group~$\bar G^{\sim}$ is effective itself; 
cf.\ Remark~\ref{rem:GenEquivGroupOfGenBurgersKdVEqsGaugeAr1A10} below. 
At the same time, there exist classes of differential equations, 
where effective generalized equivalence groups are proper subgroups 
of the corresponding generalized equivalence groups that are even not normal. 
Hence each of these effective generalized equivalence groups is not unique 
since it differs from some of subgroups non-identically similar to it, 
and all of these subgroups are also effective generalized equivalence groups of the same class. 
See the discussion of particular examples in Remark~\ref{rem:UniquenessOfEffectiveGenEquivGroups} below. 

Suppose that the class~$\mathcal L|_{\mathcal S}$ possesses 
parameterized non-identity usual equivalence transformations 
and some of its arbitrary elements are constants. 
Then this class necessarily admits purely generalized equivalence transformations. 
Indeed, we can set all parameters of elements from the usual equivalence group~$G^\sim$ 
depending on constant arbitrary elements, which gives generalized equivalence transformations. 
The set~$\bar G^{\sim}_0$ of such transformations is a subgroup of the generalized equivalence group~$\bar G^{\sim}$. 
If $\bar G^{\sim}_0=\bar G^{\sim}$, the usual equivalence group~$G^\sim$ 
is an effective generalized equivalence group of the class~$\mathcal L|_{\mathcal S}$.

The \emph{generalized equivalence algebra}~$\bar{\mathfrak g}^\sim$ 
and an \emph{effective generalized equivalence algebra} of the class~$\mathcal L|_{\mathcal S}$ 
are the algebras associated with the generalized equivalence group~$\bar G^\sim$ 
and with an effective generalized equivalence group of this class 
and are constituted by the infinitesimal generators of one-parameter subgroups of these groups, respectively. 
These infinitesimal generators are vector fields in the joint space of $(x,u,\theta)$. 

The property for equivalence transformations to be point transformations with respect to arbitrary elements 
can also be weakened. 
We formally extend the arbitrary-element tuple~$\theta$ of the class~$\mathcal L|_{\mathcal S}$ with virtual arbitrary elements 
that are related to initial arbitrary elements by differential equations 
and thus expressed via initial arbitrary elements in a nonlocal way. 
Denote the reparameterized class by~$\hat{\mathcal L}|_{\mathcal S}$. 
Suppose that the usual (resp.\ generalized or effective generalized) 
equivalence group~$\hat G^\sim$ of~$\hat{\mathcal L}|_{\mathcal S}$ 
induces the maximal subgroupoid of the equivalence groupoid~$\mathcal G^\sim$ 
among the classes obtained from~$\mathcal L|_{\mathcal S}$ by similar reparameterizations, 
and the extension of the arbitrary-element tuple~$\theta$ for~$\hat{\mathcal L}|_{\mathcal S}$ 
is minimal among the reparameterized classes giving the same subgroupoid of~$\mathcal G^\sim$ as~$\hat{\mathcal L}|_{\mathcal S}$. 
Then we call the group~$\hat G^\sim$ an \emph{extended equivalence group} (resp.\ an \emph{extended generalized equivalence group})
of the class~$\mathcal L|_{\mathcal S}$.

A point symmetry transformation of a system~$\mathcal L_\theta$ 
is a point transformation in the space of $(x,u)$ that preserves the solution set of~$\mathcal L_\theta$.
Each point symmetry transformation~$\varphi$ of~$\mathcal L_\theta$ 
gives rise to the single admissible transformation $(\theta,\theta,\varphi)$ of the class~$\mathcal L|_{\mathcal S}$. 
The point symmetry transformations of the system~$\mathcal L_\theta$ constitute the maximal point symmetry group~$G_\theta$ of this system.
The common part~$G^\cap$ of all~$G_\theta$ is called
the \emph{kernel} of maximal point symmetry groups of systems from the class~$\mathcal L|_{\mathcal S}$,
$G^\cap:=\bigcap_{\theta\in\mathcal S} G_\theta$. 
The infinitesimal counterparts of the maximal point symmetry group~$G_\theta$ and the kernel~$G^\cap$, 
which are called  
the \emph{maximal Lie invariance algebra}~$\mathfrak g_\theta$ of~$\mathcal L_\theta$ and 
the \emph{kernel Lie invariance algebra}~$\mathfrak g^\cap$ of systems from~$\mathcal L|_{\mathcal S}$, 
consist of the vector fields in the space of~$(x,u)$ generating one-parameter subgroups of~$G_\theta$ and~$G^\cap$, 
respectively.

Group analysis of differential equations becomes much simpler 
when working with infinitesimal counterparts of objects consisting of point transformations. 
Thus, the problem on Lie (i.e., continuous point) symmetries of a system~$\mathcal L_\theta$ reduces to
constructing the maximal Lie invariance algebra~$\mathfrak g_\theta$ and, therefore, is linear
in contrast to the similar problem on general point symmetries. 
Under certain quite natural conditions on the system~$\mathcal L_\theta$,
the infinitesimal invariance criterion states 
that a vector field~$Q$ in the space of $(x,u)$ 
belongs to the maximal Lie invariance algebra~$\mathfrak g_\theta$ if and only if 
the condition 
$
 Q^{(\EqOrd)}L(x,u^{(\EqOrd)},\theta^{(q)}(x,u^{(\EqOrd)}))=0
$
is identically satisfied on the manifold~$\mathcal L^\EqOrd_\theta$ 
defined by the system~$\mathcal L_\theta$ and its differential consequences in the jet space~$J^{(\EqOrd)}$.
Here $Q^{(\EqOrd)}$ is the standard $\EqOrd$th order prolongation of the vector field~$Q$,
see~\cite{olve86Ay,ovsi82Ay} and Section~\ref{sec:DetEqsLieSymmetriesGenBurgersKdVEqs}.

The group classification problem for the class~$\mathcal L|_{\mathcal S}$ is 
to list all $G^\sim$-inequivalent values for~$\theta\in\mathcal S$ 
such that the associated systems, $\mathcal L_\theta$, 
admit maximal Lie invariance algebras, $\mathfrak g_\theta$, 
that are wider than the kernel Lie invariance algebra~$\mathfrak g^\cap$. 
Further taking into account additional point equivalences between obtained cases, 
provided such additional equivalences exist, 
one solves the group classification problem up to~$\mathcal G^\sim$-equivalence.
Restricting the consideration to Lie symmetries is important 
for the group classification problem to be well-posed within the framework of classes of differential equations. 

When applied to systems from a class~$\mathcal L|_{\mathcal S}$, 
the infinitesimal invariance criterion yields, after splitting with respect to the parametric derivatives of~$u$,
a system of the determining equations for the components of Lie symmetry generators of these systems, 
which is in general parameterized by the arbitrary-element tuple~$\theta$.
It is quite common that there is a subsystem of the determining equations 
that does not involve the tuple of arbitrary elements~$\theta$ and hence may be integrated in a regular way. 
The remaining part of the determining equations that explicitly involve the arbitrary elements is referred to as the system of \emph{classifying equations}. 
In this setting, the group classification problem reduces to the exhaustive investigation of the classifying equations. 
The direct integration of the classifying equations up to $G^\sim$-equivalence 
is usually only possible for classes of the simplest structure, 
e.g., classes involving only constants or functions of a single argument as arbitrary elements, see, e.g., examples in~\cite{ovsi82Ay}. 
Since most classes of interest in applications are of more complicated structure, 
various methods have to be used, which at least enhance the direct method~\cite{niki01Ay,vane09Ay,vane12Ay}.

The most advanced classification techniques rest on the study of algebras of vector fields 
associated with systems from the class~$\mathcal L|_{\mathcal S}$ under consideration 
and constitute, in total, the \emph{algebraic method} of group classification.
For this method to be really effective, 
the class~$\mathcal L|_{\mathcal S}$ has to possess certain properties, 
which are conveniently formulated in terms of various notions of normalization.
The class of differential equations~$\mathcal L|_{\mathcal S}$ is \emph{normalized} 
in the usual (resp.\ generalized, extended, extended generalized) sense
if the subgroupoid induced by its usual (resp.\ generalized, extended, extended generalized) equivalence group
coincides with the entire equivalence groupoid~$\mathcal G^\sim$ of~$\mathcal L|_{\mathcal S}$. 

The normalization of~$\mathcal L|_{\mathcal S}$ in the usual sense 
is equivalent to the following conditions.
The transformational part~$\varphi$ of each admissible transformation $(\theta',\theta'',\varphi)\in\mathcal G^\sim$
does not depend on the fixed initial value~$\theta'$ of the arbitrary-element tuple~$\theta$
and, therefore, is appropriate for any initial value of~$\theta$.
Moreover, the prolongation of~$\varphi$ to the space of~$(x,u^{(\EqOrd)})$
and the further extension to the arbitrary elements according to the relation between~$\theta'$ and~$\theta''$
gives a point transformation in the joint space of~$(x,u^{(\EqOrd)},\theta)$.
Then $G_\theta\leqslant\pi_*G^\sim$ and $\mathfrak g_\theta\subseteq\pi_*\mathfrak g^\sim$
for any~$\theta\in\mathcal S$, 
and hence the group classification of the class~$\mathcal L|_{\mathcal S}$ reduces 
to the classification of certain $G^\sim$-inequivalent subalgebras of the equivalence algebra~$\mathfrak g^\sim$
or, equivalently, 
to the classification of certain $\pi_*G^\sim$-inequivalent subalgebras of the projection~$\pi_*\mathfrak g^\sim$. 

If the class~$\mathcal L|_{\mathcal S}$ is normalized in the generalized sense, 
the expression for transformational parts of admissible transformations may involve 
arbitrary elements but only in a quite specific way.  
The equivalence groupoid is partitioned into families of admissible transformations 
parameterized by the source arbitrary-element tuple, 
and the transformational parts of admissible transformations from each of these families 
jointly give, after the extension to the arbitrary elements according to the relation between 
the corresponding source and target arbitrary elements, 
a point transformation in the joint space of~$(x,u,\theta)$. 
Then $G_{\theta'}\leqslant\pi_*(G^\sim|_{\theta=\theta'(x,u)})$ 
and $\mathfrak g_{\theta'}\subseteq\pi_*(\mathfrak g^\sim|_{\theta=\theta'(x,u)})$
for any~$\theta'\in\mathcal S$.

The class~$\mathcal L|_{\mathcal S}$ is called \emph{semi-normalized} in the usual sense 
if for any $(\theta',\theta'',\varphi)\in\mathcal G^\sim$ there exist 
$\mathcal T\in G^\sim$, $\varphi'\in G_{\theta'}$ and $\varphi''\in G_{\theta''}$ 
such that $\theta''=\mathcal T\theta'$ and $\varphi=\varphi''(\pi_*\mathcal T)\varphi'$. 
One of the transformations~$\varphi'$ and~$\varphi''$ can always be assumed 
to coincide with the identity transformation. 
Semi-normalization in the generalized sense is defined in a similar way. 
Roughly speaking, a class is semi-normalized in a certain sense 
if its equivalence groupoid is generated by its relevant equivalence group 
jointly with point symmetry groups of systems from this class. 
Each normalized class is semi-normalized in the same sense, 
and the converse is not in general true. 
There are also more sophisticated notions, uniform semi-normalization and weak uniform semi-normalization, 
which mediate the notion of normalization and semi-normalization~\cite{kuru16a,kuru17a}.
\looseness=-1

To establish the normalization properties of the class~$\mathcal L|_{\mathcal S}$ 
one should compute its equivalence groupoid~$\mathcal G^\sim$, 
which is realized using the direct method. 
Here one fixes two arbitrary systems from the class,
$\mathcal L_\theta\colon L(x,u^{(\EqOrd)},\theta(x,u^{(\EqOrd)}))=0$
and $\mathcal L_{\tilde\theta}\colon L(\tilde x,\tilde u^{(\EqOrd)},\tilde \theta(\tilde x,\tilde u^{(\EqOrd)}))=0$,
and aims to find the (nondegenerate) point transformations, 
$\varphi$: $\tilde x_i=X^i(x,u)$, $\tilde u^a=U^a(x,u)$, $i=1,\dots,n$, $a=1,\dots,m$, connecting them.
For this, one changes the variables in the system~$\mathcal L_{\tilde\theta}$ 
by expressing the derivatives $\tilde u^{(\EqOrd)}$ in terms of $u^{(\EqOrd)}$ and derivatives of the functions $X^i$ and $U^a$ 
as well as by substituting $X^i$ and $U^a$ for $\tilde x_i$ and $\tilde u^a$, respectively.
The requirement that the resulting transformed system 
has to be satisfied identically for solutions of~$\mathcal L_\theta$
leads to the system of determining equations for the components of the transformation~$\varphi$.

In the case of a single dependent variable ($m=1$), 
all the above notions involving point transformations can be directly extended to contact transformations.

\section{Equivalence groupoid}\label{sec:EquivalenceGroupoidGenBurgersKdVEqs}

We now compute the equivalence groupoid and equivalence group of the class~\eqref{eq:GenBurgersKdVEqs} using the direct method. 
Equivalence transformations will be used to find an appropriate gauged subclass of~\eqref{eq:GenBurgersKdVEqs} 
that is suitable for carrying out the complete group classification. The presentation closely follows~\cite{bihl16a}. 
In particular, it is convenient to start with the wide superclass 
of general $(1{+}1)$-dimensional $\EqOrd$th order ($\EqOrd\geqslant2$) evolution equations of the form
\begin{equation}\label{eq:GenEvolEqs}
 u_t=H(t,x,u_0,\dots,u_\EqOrd), \quad H_{u_\EqOrd}\ne0,
\end{equation}
and sequentially narrowing it until the class~\eqref{eq:GenBurgersKdVEqs} and its gauged subclasses are reached.
The advantage of this method is that one can infer the normalization properties of the class~\eqref{eq:GenBurgersKdVEqs} 
by keeping track of the normalization properties of the class~\eqref{eq:GenEvolEqs} and its relevant subclasses. 
This not only gives restrictions on the transformational part of admissible transformations in the class~\eqref{eq:GenEvolEqs} and its subclasses, 
but also leads to a more and more constrained relation between the initial and target arbitrary elements until this relation is sufficiently specified.

It was established in~\cite{maga93Ay} that a contact transformation of the independent variables $(t,x)$ and the dependent variable~$u$
connects two fixed equations from the class~\eqref{eq:GenEvolEqs} 
if and only if its components are of the form
$\tilde t=T(t)$, $\tilde x=X(t,x,u,u_x)$ and $\tilde u=U(t,x,u,u_x)$ 
provided that the usual nondegeneracy assumption and contact condition hold, 
\[
T_t\ne0, \quad 
\mathop{\rm rank}\left(\begin{array}{ccc}X_x&X_u&X_{u_x}\\U_x&U_u&U_{u_x}\end{array}\right)=2
\qquad\mbox{and}\qquad
(U_x+U_uu_x)X_{u_x}=(X_x+X_uu_x)U_{u_x}.
\]
The prolongation of the transformation to the first derivatives is given by 
\[
\tilde u_{\tilde x}=V, \quad
\tilde u_{\tilde t}=\frac{U_u-X_uV}{T_t}u_t+\frac{U_t-X_tV}{T_t},
\qquad\mbox{where}\qquad
V=\frac{U_x+U_uu_x}{X_x+X_uu_x}\quad\mbox{or}\quad V=\frac{U_{u_x}}{X_{u_x}}
\]
if $X_x+X_uu_x\ne0$ or $X_{u_x}\ne0$, respectively. 
Such transformations prolonged to the arbitrary element~$H$ according to
\[
\tilde H=\frac{U_u-X_uV}{T_t}H+\frac{U_t-X_tV}{T_t}
\]
constitute the contact usual equivalence group of the class~\eqref{eq:GenEvolEqs}
and thus this class is normalized in the usual sense with respect to contact transformations. 
It is also normalized in the usual sense with respect to point transformations. 
The point equivalence groupoid and the point usual equivalence group are singled out from their contact counterparts 
by the condition \mbox{$X_{u_x}=U_{u_x}=0$}.

Consider the subclass~$\mathcal E$ of the class~\eqref{eq:GenEvolEqs} singled out by the constraints
\[
H_{u_\EqOrd u_k}=0,\quad k=1,\dots,\EqOrd, \qquad
H_{u_{\EqOrd-1}u_l}=0,\quad l=1,\dots,\EqOrd-1,
\]
cf.~\cite{vane14a}.
Due to the constraints $H_{u_\EqOrd u_k}=0$, $k=2,\dots,\EqOrd$,
it follows that contact admissible transformations in the subclass~$\mathcal E$
are induced by point admissible transformations.
In other words, the contact equivalence groupoid of~$\mathcal E$
coincides with the first prolongation of the point equivalence groupoid of~$\mathcal E$.
Then the constraints $H_{u_\EqOrd u_1}=0$ and $H_{u_{\EqOrd-1}u_l}=0$, \mbox{$l=1,\dots,\EqOrd-1$},
successively imply $X_u=0$ and $U_{uu}=0$ for any admissible transformation in~$\mathcal E$,
i.e., its transformational part is of the form
\begin{equation}\label{eq:IntermediatePointTransformationBurgersKdV}
\tilde t=T(t),\quad \tilde x=X(t,x),\quad \tilde u=U^1(t,x)u+U^0(t,x) 
\qquad\mbox{with}\qquad T_tX_xU^1\ne0.
\end{equation}
The prolongations of these transformations to the arbitrary element~$H$
constitute the usual equivalence group of class~$\mathcal E$.
Therefore, the class~$\mathcal E$ is normalized in the usual sense.

To single out the class~\eqref{eq:GenBurgersKdVEqs}, we should set more constraints for~$H$. 
The complete system of these constraints is given by 
\[
H_{u_ku_l}=0,\quad 1\leqslant k\leqslant l\leqslant\EqOrd, \quad (k,l)\ne(0,1),\quad H_{u_\EqOrd}\ne0,\quad H_{u_0u_1}\ne0.
\]
Then we should also reparameterize the obtained subclass, 
assuming $\theta=(A^0,\dots,A^\EqOrd,B,C)$ as the tuple of arbitrary elements instead of~$H$.
Using the direct method for computing the equivalence groupoid of the class~\eqref{eq:GenBurgersKdVEqs}, 
we first fix two arbitrary equations~$\mathcal L_\theta$ and~$\mathcal L_{\tilde \theta}$ 
from the class~\eqref{eq:GenBurgersKdVEqs} 
and require that they are connected through a point transformation~$\varphi$ 
of the form~\eqref{eq:IntermediatePointTransformationBurgersKdV}. 
This particular form can be posed for admissible transformations
since the class~\eqref{eq:GenBurgersKdVEqs} is a subclass of the normalized class~$\mathcal E$.
It is thus necessary to re-express the jet variables~$(\tilde t,\tilde x,\tilde u^{(\EqOrd)})$ in terms of $(t,x,u^{(\EqOrd)})$. 
In view of~\eqref{eq:IntermediatePointTransformationBurgersKdV},
the expressions for the transformed total derivative operators~are
\[
{\rm D}_{\tilde t}=\frac{1}{T_t}\left({\rm D}_t-\frac{X_t}{X_x}{\rm D}_x\right),\quad {\rm D}_{\tilde x}=\frac1{X_x}{\rm D}_x.
\]
Substituting the expressions for the transformed values into~$\mathcal L_{\tilde \theta}$ 
yields an intermediate equation~\smash{$\tilde{\mathcal L}$}. 
Because the equations~$\mathcal L_\theta$ and~$\mathcal L_{\tilde \theta}$ 
are by assumption connected by~$\varphi$, 
the equation~$\tilde{\mathcal L}$ has to be satisfied by all solutions of~$\mathcal L_\theta$.
We assume $u_t$ as the leading derivative in~$\mathcal L_\theta$ 
and substitute the expression for~$u_t$ obtained from~$\mathcal L_\theta$ into~$\tilde{\mathcal L}$. 
This leads to an identity, which can be split with respect to the parametric derivatives~$u_0$, \dots, $u_\EqOrd$.
The condition that the coefficient of $u^2$ in~$\tilde{\mathcal L}$ has to be zero requires $U^1_x=0$.
Collecting the other coefficients of powers of parametric derivatives, 
we derive the formulas pointwise relating $\theta$ and~$\tilde\theta$ 
with no constraints for~$T$, $X$, $U^1$ and~$U^0$. 
These formulas are quite cumbersome (although obtainable using the Fa\`{a} di Bruno's formula). 
In addition, they are not needed at the present stage 
since we can first fix a suitable gauged subclass of the class~\eqref{eq:GenBurgersKdVEqs}. 
To do this, we only need the transformation component 
for the arbitrary element $C=C(t,x)$, which is readily obtained without using Fa\`{a} di Bruno's formula, 
\[
 \tilde C=\frac{X_x}{T_tU^1}C.
\]

\begin{proposition}\label{pro:EquivGroupoidOfGenBurgersKdVEqs}
The class~\eqref{eq:GenBurgersKdVEqs} is normalized in the usual sense.
Its usual equivalence group~$G^\sim_{\mbox{\tiny\eqref{eq:GenBurgersKdVEqs}}}$
consists of the transformations in the joint space of $(t,x,u,\theta)$
whose $(t,x,u)$-components are of the form
\[
\tilde t=T(t),\quad \tilde x=X(t,x),\quad \tilde u=U^1(t)u+U^0(t,x), 
\]
where $T=T(t)$, $X=X(t,x)$, $U^1=U^1(t)$ and $U^0=U^0(t,x)$
are arbitrary smooth functions of their arguments such that $T_tX_xU^1\ne0$.
\end{proposition}

We can now use the family of equivalence transformations parameterized by the arbitrary element~$C$, where
\[
\tilde t=t,\quad\tilde x=\int_{x_0}^x\frac{\mathrm{d}y}{C(t,y)},\quad \tilde u=u,
\]
to map the class~\eqref{eq:GenBurgersKdVEqs} 
onto its subclass singled out by the constraint $C=1$. 
Under this gauging we derive that $X_x=T_tU^1$ and thus $X_{xx}=0$, 
i.e., $X=X^1(t)x+X^0(t)$ and $U^1=X^1/T_t$. 
The gauged subclass is still normalized in the usual sense. 
Moreover, due to the most prominent equations from the class~\eqref{eq:GenBurgersKdVEqs} 
being of the form with $C=1$, this gauge is quite natural.

\begin{remark}
Given a class of differential equations, if a subclass is singled from it 
by constraints with explicit expressions for some arbitrary elements, 
we reparameterize this subclass using the reduced tuple of arbitrary elements, 
which is obtained from the complete tuple by excluding the constrained arbitrary elements. 
For example, under the gauge~$C=1$ we can assume 
that the tuple of arbitrary elements for the corresponding subclass 
is $(A^0,\dots,A^\EqOrd,B)$.
\end{remark}

With the restrictions on $T$, $X$ and $U$ obtained thus far, 
we now complete the procedure for finding the equivalence groupoid 
of the subclass associated with gauge $C=1$, 
which consists of equations of the form
\begin{equation}\label{eq:GenBurgersKdVEqsGaugeSubclassC1}
 u_t+uu_x=\sum_{k=0}^\EqOrd A^k(t,x)u_k+B(t,x).
\end{equation}
Transforming the equations from the gauged subclass, we find
\begin{gather*}
\frac1{T_t}\left(U^1u_t+U^1_tu+U^0_t-\frac{X_t}{X^1}(U^1u_x+U^0_x)\right)+\frac{1}{X^1}(U^1u+U^0)(U^1u_x+U^0_x)\\
\qquad{}=\sum_{k=0}^\EqOrd\frac{\tilde A^k}{(X^1)^k}(U^1u_k+U^0_k)+\tilde B,
\end{gather*}
which can be written, after substituting for~$u_t$ in view of the equation~\eqref{eq:GenBurgersKdVEqsGaugeSubclassC1}, as
\begin{gather*}
\sum_{k=0}^\EqOrd A^ku_k+B+\left(\frac{U^0}{U^1}-\frac{X_t}{X^1}\right)u_x+\left(\frac{U^1_t}{U^1}+\frac{U^0_x}{U^1}\right)u
+\frac{U^0_t}{U^1}-\frac{X_t}{X^1}\frac{U^0_x}{U^1}+\frac{U^0U^0_x}{(U^1)^2}{}\\
\qquad{}=\sum_{k=0}^\EqOrd\left[\frac{T_t}{(X^1)^k}\tilde A^ku_k+\frac{T_t}{U^1}\left(\tilde B+ \frac{\tilde A^k}{(X^1)^k}U^0_k\right)\right].
\end{gather*}
Splitting this equation with respect to~$u_k$ yields
the transformation components for the arbitrary elements. 
We have thus proved the following theorem.

\begin{theorem}\label{thm:EquivalenceGroupGenBurgersKdVEqsGaugeC1}
The class~\eqref{eq:GenBurgersKdVEqsGaugeSubclassC1} 
of reduced $(1{+}1)$-dimensional general $\EqOrd$th order Burgers--KdV equations with $C=1$ is normalized in the usual sense.
Its usual equivalence group~$G^\sim_{\mbox{\tiny\eqref{eq:GenBurgersKdVEqsGaugeSubclassC1}}}$ 
is constituted by the transformations of the form
\begin{gather*}
\tilde t=T(t),\quad \tilde x=X^1(t)x+X^0(t),\quad \tilde u=\frac{X^1}{T_t}u+U^0(t,x),
\\
\tilde A^j=\frac{(X^1)^j}{T_t}A^j,\ \
\tilde A^1=\frac{X^1}{T_t}A^1+U^0-\frac{X^1_tx+X^0_t}{T_t},\ \
\tilde A^0=\frac{1}{T_t}\left(A^0+\frac{X^1_t}{X^1}-\frac{T_{tt}}{T_t}+\frac{T_t}{X^1}U^0_x\right),
\\
\tilde B=\frac{X^1}{(T_t)^2}B+\frac{U^0_t}{T_t}
+\frac{U^0_x}{X^1}\left(U^0-\frac{X^1_tx+X^0_t}{T_t}\right)
-\sum_{k=0}^\EqOrd\frac{U^0_k}{(X^1)^k}\tilde A^k,
\end{gather*}
where $j=2,\dots,\EqOrd$, and $T=T(t)$, $X^1=X^1(t)$, $X^0=X^0(t)$ and $U^0=U^0(t,x)$
are arbitrary smooth functions of their arguments such that $T_tX^1\ne0$.
\end{theorem}

At this stage, it is convenient to introduce one more gauge.
In particular, the family of equivalence transformations parameterized by the arbitrary element~$A^1$
with $T=t$, $X^1=1$, $X^0=0$ and $U^0=-A^1$
maps the associated gauged subclass~\eqref{eq:GenBurgersKdVEqsGaugeSubclassC1} with $C=1$
to the subclass of~\eqref{eq:GenBurgersKdVEqs} with $C=1$ and $A^1=0$.
This gauging implies that $U^0=(X^1_tx+X^0_t)/T_t$.
The corresponding gauged subclass consisting of equations of the form
\begin{equation}\label{eq:GenBurgersKdVEqsGaugedSubclass}
\mathcal L_\kappa\colon\quad u_t+uu_x=\sum_{j=2}^\EqOrd A^j(t,x)u_j+A^0(t,x)u+B(t,x),
\end{equation}
where $A^r\ne0$ and $\kappa=(A^0,A^2,\dots,A^\EqOrd,B)$ is the reduced arbitrary-element tuple, is still normalized in the usual sense.
This is the gauged subclass that is appropriate
for solving the complete group classification problem
for the class~\eqref{eq:GenBurgersKdVEqs}.
This leads to the following theorem.

\begin{theorem}\label{thm:EquivalenceGroupGenBurgersKdVEqsGaugeC1A10}
The class~\eqref{eq:GenBurgersKdVEqsGaugedSubclass} 
of reduced $(1{+}1)$-dimensional general  $\EqOrd$th~order Burgers--KdV equations,  
which is singled out from the class~\eqref{eq:GenBurgersKdVEqs} by the gauge $(C,A^1)=(1,0)$, 
is normalized in the usual sense.
Its usual equivalence group~$G^\sim$ consists of the transformations of the form
\begin{subequations}\label{eq:EquivalenceTransformationsGenBurgersKdVEqsGaugeC1A10}
\begin{gather}
\tilde t=T(t),\quad \tilde x=X^1(t)x+X^0(t),\quad \tilde u=\frac{X^1}{T_t}u+\frac{X^1_t}{T_t}x+\frac{X^0_t}{T_t},
\label{eq:PointTransformationBetweenGenBurgersKdVEqsGaugeC1A10a}\\
\tilde A^j=\frac{(X^1)^j}{T_t}A^j,\quad  \tilde A^0=\frac{1}{T_t}\left(A^0+2\frac{X^1_t}{X^1}-\frac{T_{tt}}{T_t}\right),
\label{eq:PointTransformationBetweenGenBurgersKdVEqsGaugeC1A10b}\\
\tilde B=\frac{X^1}{(T_t)^2}B+\frac1{T_t}\left(\frac{X^1_t}{T_t}\right)_tx+\frac1{T_t}\left(\frac{X^0_t}{T_t}\right)_t
-\left(\frac{X^1_t}{T_t}x+\frac{X^0_t}{T_t}\right)\tilde A^0,
\label{eq:PointTransformationBetweenGenBurgersKdVEqsGaugeC1A10c}
\end{gather}
\end{subequations}
where $j=2,\dots,\EqOrd$, and $T=T(t)$, $X^1=X^1(t)$ and $X^0=X^0(t)$ are arbitrary smooth functions of their arguments with $T_tX^1\ne0$.

\end{theorem}

\begin{corollary}\label{cor:OnEquivalenceAlgebraGenBurgersKdVEqs}
The equivalence algebra of the class~\eqref{eq:GenBurgersKdVEqsGaugedSubclass} 
of $(1{+}1)$-dimensional general $\EqOrd$th~order Burgers--KdV equations is given by
$
 \mathfrak g^\sim=\langle\hat D(\tau),\hat S(\zeta),\hat P(\chi)\rangle,
$
where $\tau=\tau(t)$, $\zeta=\zeta(t)$ and $\chi=\chi(t)$ run through the set of smooth functions of~$t$, with
\begin{gather*}
\hat D(\tau)=\tau\partial_t-\tau_tu\p_u-\tau_t\sum_{j=2}^\EqOrd A^j\p_{A^j}-(\tau_tA^0+\tau_{tt})\p_{A^0}-2\tau_tB\p_B,
\\[-1ex]
\hat S(\zeta)=\zeta x\p_x+(\zeta u+\zeta_tx)\p_u+j\zeta\sum_{j=2}^\EqOrd A^j\p_{A^j}+2\zeta_t\p_{A^0}+(\zeta B+\zeta_{tt}x-\zeta_t x A^0)\p_B,
\\[1ex]
\hat P(\chi)=\chi\partial_x+\chi_t\p_u+(\chi_{tt}-\chi_t A^0)\p_B.
\end{gather*}
\end{corollary}

\begin{proof}
Since we have already computed the usual equivalence group~$G^\sim$, the associated equivalence algebra~$\mathfrak g^\sim$ can be obtained in a straightforward deductive fashion. In particular, $\mathfrak g^\sim$ is spanned by vector fields representing the infinitesimal generators of one-parameter subgroups of the usual equivalence group~$G^\sim$. Thus, successively assuming one of the parameter functions~$T$, $X^1$ and $X^0$ to depend on a continuous group parameter~$\varepsilon$ 
(in such a manner that the identical transformation corresponds to the value~$\varepsilon=0$), 
we can obtain the coefficients of the infinitesimal generators of the form
$\hat Q=\tau\p_t+\xi\p_x+\eta\p_u+\phi^0\p_{A^0}+\sum_{j=2}^\EqOrd\phi^j\p_{A^j}+\psi\p_B$ by determining
\[
 \tau=\frac{\mathrm{d} \tilde t}{\mathrm{d} \varepsilon}\Big|_{\varepsilon=0},\ \ 
 \xi=\frac{\mathrm{d} \tilde x}{\mathrm{d} \varepsilon}\Big|_{\varepsilon=0},\ \ 
 \eta=\frac{\mathrm{d} \tilde u}{\mathrm{d} \varepsilon}\Big|_{\varepsilon=0},\ \ 
 \phi^0=\frac{\mathrm{d} \tilde A^0}{\mathrm{d} \varepsilon}\Big|_{\varepsilon=0},\ \ 
 \phi^j=\frac{\mathrm{d} \tilde A^j}{\mathrm{d} \varepsilon}\Big|_{\varepsilon=0},\ \ 
 \psi=\frac{\mathrm{d} \tilde B}{\mathrm{d} \varepsilon}\Big|_{\varepsilon=0}.
\]
This results in the generating vector fields~$\hat D(\tau)$, $\hat S(\zeta)$ and~$\hat P(\chi)$, 
which are associated to the parameter functions~$T$, $X^1$ and $U^0$, respectively.
\end{proof}

Since the class~\eqref{eq:GenBurgersKdVEqs} is mapped onto
the class~\eqref{eq:GenBurgersKdVEqsGaugedSubclass} by a family of equivalence transformations, 
and both the classes are normalized in the usual sense, 
the following assertion is obvious. 

\begin{proposition}\label{pro:GroupClassificationsOfGenBurgersKdVEqsAndGaugedGenBurgersKdVEqs} 
The group classification of the class~\eqref{eq:GenBurgersKdVEqs} reduces to 
that of its subclass~\eqref{eq:GenBurgersKdVEqsGaugedSubclass}. 
More specifically, any complete list of $G^\sim$-inequivalent Lie symmetry extensions 
in the class~\eqref{eq:GenBurgersKdVEqsGaugedSubclass} is 
a complete list of $G^\sim_{\mbox{\tiny\eqref{eq:GenBurgersKdVEqs}}}$-inequivalent Lie symmetry extensions 
in the class~\eqref{eq:GenBurgersKdVEqs}.
\end{proposition}

\section{Alternative gauges}\label{sec:AlternativeGauges}

We show that the gauge $C=1$ is the best initial gauge for the class~\eqref{eq:GenBurgersKdVEqs}
and the gauge $(C,A^1)=(1,0)$ is the best for singling out a subclass
in order to carry out the group classification.

An obvious choice for an alternative gauge is $A^\EqOrd=1$. 
It was used in~\cite{bihl16a} as the basic gauge in the course of group classification of 
linear equations of the form~\eqref{eq:GenBurgersKdVEqs}, for which $C=0$.
The $A^\EqOrd$-component of equivalence transformations 
in the class~\eqref{eq:GenBurgersKdVEqs} is 
\[
\tilde A^\EqOrd=\frac{(X_x)^\EqOrd}{T_t}A^\EqOrd.
\]
If~$A^\EqOrd=1$ and~$\tilde A^\EqOrd=1$, then the parameters of the corresponding admissible transformations
given in Proposition~\ref{pro:EquivGroupoidOfGenBurgersKdVEqs}
satisfy the constraint $(X_x)^\EqOrd=T_t$, i.e., $X=X^1(t)x+X^0(t)$, where $(X^1)^\EqOrd=T_t$,
which makes the parameterization of the usual equivalence group of the corresponding subclass
more complicated than using the gauge~$C=1$.

\begin{proposition}\label{pro:EquivGroupoidOfGenBurgersKdVEqsGaugeAr1}
The subclass of the class~\eqref{eq:GenBurgersKdVEqs}
singled out by the constraint $A^\EqOrd=1$ is normalized in the usual sense.
Its usual equivalence group is constituted by the transformations of the form
\begin{gather*}
\tilde t=T(t),\quad \tilde x=X^1(t)x+X^0(t),\quad \tilde u=U^1(t)u+U^0(t,x),
\\
\tilde A^l=\frac{(X^1)^l}{T_t}A^l,\quad
\tilde A^1=\frac{X^1}{T_t}\left(A^1+\frac{U^0}{U^1}C-\frac{X^1_tx+X^0_t}{X^1}\right),\quad
\tilde A^0=\frac{1}{T_t}\left(A^0+\frac{U^1_t}{U^1}+\frac{U^0_x}{U^1}C\right),
\\
\tilde B=\frac{U^1}{T_t}B+\frac{U^0_t}{T_t}
+\frac{U^0_x}{T_t}\left(\frac{U^0}{U^1}C-\frac{X^1_tx+X^0_t}{X^1}\right)
-\frac{U^0_\EqOrd}{(X^1)^\EqOrd}-\sum_{k=0}^{\EqOrd-1}\frac{U^0_k}{(X^1)^k}\tilde A^k,\quad
\tilde C=\frac{X^1}{T_tU^1}C,
\end{gather*}
where $l=2,\dots,\EqOrd-1$,
and $T=T(t)$, $X^0=X^0(t)$, $U^1=U^1(t)$ and $U^0=U^0(t,x)$ are arbitrary smooth functions of their arguments
such that $T_tU^1\ne0$,
as well as $X^1=(T_t)^{1/\EqOrd}$ if $\EqOrd$ is odd and
$T_t>0$, $X^1=\varepsilon(T_t)^{1/\EqOrd}$ with $\varepsilon=\pm1$ if $\EqOrd$ is even.
\end{proposition}

In contrast to the class~\eqref{eq:GenBurgersKdVEqsGaugeSubclassC1},
the additional gauge $A^1=0$ slightly worsens the normalization property.
It leads to the appearance of the arbitrary element~$C$ in the $u$-component of admissible transformations
since then we have
\[
U^0=\frac{X^1_tx+X^0_t}{X^1C}U^1.
\]
Denote by~$\theta'$ the arbitrary-element tuple reduced by the double gauge, 
\[\theta'=(A^0,A^2,\dots,A^{\EqOrd-1},B,C).\] 

\begin{proposition}\label{pro:EquivGroupoidOfGenBurgersKdVEqsGaugeAr1A10}
The equivalence groupoid of the subclass~$\mathcal A_1$ of the class~\eqref{eq:GenBurgersKdVEqs}
singled out by the constraints $A^\EqOrd=1$ and $A^1=0$
consists of the triples $(\theta',\tilde\theta',\varphi)$'s, 
where the point transformation~$\varphi$ is of the form 
\begin{subequations}\label{eq:PointTransOfGenBurgersKdVEqsGaugeAr1A10}
\begin{gather}
\tilde t=T(t),\quad \tilde x=X^1(t)x+X^0(t),\quad \tilde u=U^1(t)u+U^0,\quad U^0:=\frac{X^1_tx+X^0_t}{X^1C}U^1,
\label{eq:PointTransOfGenBurgersKdVEqsGaugeAr1A10a}
\end{gather}
the arbitrary-element tuples~$\theta'$ and~$\tilde\theta'$ are related according to
\begin{gather}
\tilde A^l=\frac{(X^1)^l}{T_t}A^l,\quad
\tilde A^0=\frac{1}{T_t}\left(A^0+\frac{U^1_t}{U^1}+\frac{U^0_x}{U^1}C\right),\quad
\tilde C=\frac{X^1}{T_tU^1}C,
\label{eq:PointTransOfGenBurgersKdVEqsGaugeAr1A10b}
\\
\tilde B=\frac{U^1}{T_t}B+\frac{U^0_t}{T_t}
+\frac{U^0_x}{T_t}\left(\frac{U^0}{U^1}C-\frac{X^1_tx+X^0_t}{X^1}\right)
-\frac{U^0_\EqOrd}{(X^1)^\EqOrd}-\sum_{l=2}^{\EqOrd-1}\frac{U^0_l}{(X^1)^l}\tilde A^l-U^0\tilde A^0,
\label{eq:PointTransOfGenBurgersKdVEqsGaugeAr1A10c}
\end{gather}
\end{subequations}
with $l=2,\dots,\EqOrd-1$,
and $T=T(t)$, $X^0=X^0(t)$ and $U^1=U^1(t)$ being arbitrary smooth functions of~$t$
such that $T_tU^1\ne0$,
as well as $X^1=(T_t)^{1/\EqOrd}$ if $\EqOrd$ is odd and
$T_t>0$, $X^1=\varepsilon(T_t)^{1/\EqOrd}$ with $\varepsilon=\pm1$ if $\EqOrd$ is even.
\end{proposition}

It is obvious that the subclass~$\mathcal A_1$ is not normalized in the usual sense. 
Its usual equivalence group is constituted by the point transformations 
of the form~\eqref{eq:PointTransOfGenBurgersKdVEqsGaugeAr1A10}
in the joint space of the variables $(t,x,u)$ and the arbitrary elements~$\theta'$, 
where parameters satisfy more constraints, $T_{tt}=X^0_t=0$, and thus $X^1_t=0$ and~$U^0=0$. 

All the components of~\eqref{eq:PointTransOfGenBurgersKdVEqsGaugeAr1A10} 
locally depend on~$C$, and, moreover, the expressions for~$\tilde A^0$ and~$\tilde B$ involve 
derivatives of~$C$ with respect to~$t$ and~$x$. 
This is why, to interpret~\eqref{eq:PointTransOfGenBurgersKdVEqsGaugeAr1A10} 
as generalized equivalence transformations, 
we need to formally extend the arbitrary-element tuple~$\theta'$ 
with the derivatives of~$C$ as new arbitrary elements,
$Z^0:=C_t$ and~$Z^k:=C_k$, $k=1,\dots,\EqOrd$,
and prolong equivalence transformations to them, 
\begin{gather}\label{eq:PointTransOfGenBurgersKdVEqsGaugeAr1A10d}
\tilde Z^0=\frac{X^1}{T_t^2U^1}Z^0+\left(\frac{X^1}{T_tU^1}\right)_t\frac C{T_t}, \quad
\tilde Z^k=\frac{(X^1)^{1-k}}{T_t^2U^1}Z^k, \quad k=1,\dots,\EqOrd.
\end{gather}
The derivatives of~$U^0$ in the expressions for~$\tilde A^0$ and~$\tilde B$ should 
be expanded and then derivatives of~$C$ should be replaced by the corresponding~$Z$'s.

We denote by~$\bar{\mathcal A}_1$ 
the class of equations of the form~\eqref{eq:GenBurgersKdVEqs} with $(A^\EqOrd,A^1)=(1,0)$ 
and the extended arbitrary-element tuple~$\bar\theta'=(A^0,A^2,\dots,A^{\EqOrd-1},B,C,Z^0,\dots,Z^\EqOrd)$, 
where the relations defining~$Z^0$, \dots, $Z^\EqOrd$ are assumed as 
additional auxiliary equations for arbitrary elements. 

\begin{theorem}\label{thm:GenEquivGroupOfGenBurgersKdVEqsGaugeAr1A10}
The class~$\bar{\mathcal A}_1$ is normalized in the generalized sense.
Its generalized equivalence group~\smash{$\bar G^\sim_{\bar{\mathcal A}_1}$} coincides with its effective generalized equivalence group 
and consists of the point transformations 
in the joint space of the variables $(t,x,u)$ and the arbitrary elements~$\bar\theta'$ with components 
of the form~\eqref{eq:PointTransOfGenBurgersKdVEqsGaugeAr1A10},~\eqref{eq:PointTransOfGenBurgersKdVEqsGaugeAr1A10d} 
and the same constraints for parameters as in Proposition~\ref{pro:EquivGroupoidOfGenBurgersKdVEqsGaugeAr1A10}, 
where partial derivatives of~$U^0$ are replaced by the corresponding restricted total derivatives with 
$\bar{\mathrm D}_t=\p_t+Z^0\p_C$ and $\bar{\mathrm D}_x=\p_x+Z^1\p_C+Z^2\p_{Z^1}+\dots+Z^\EqOrd\p_{Z^{\EqOrd-1}}$.
\end{theorem}

\begin{proof}
The point transformations of the above form constitute a group~$G$, 
which generates the entire equivalence groupoid of the class~$\bar{\mathcal A}_1$ 
and is minimal among point-transformation groups in the joint space of $(t,x,u,\bar\theta')$ 
that have this generation property. 
Therefore, $G$ is an effective generalized equivalence group of the class~$\bar{\mathcal A}_1$. 
We are going to prove that the group~$G$ coincides with~\smash{$\bar G^\sim_{\bar{\mathcal A}_1}$}.
Indeed, substituting every particular value of~$\bar\theta'$ 
to any element of~\smash{$\bar G^\sim_{\bar{\mathcal A}_1}$} gives 
an admissible transformation of the class~$\bar{\mathcal A}_1$. 
This implies that elements of~\smash{$\bar G^\sim_{\bar{\mathcal A}_1}$} 
are of the form~\eqref{eq:PointTransOfGenBurgersKdVEqsGaugeAr1A10},~\eqref{eq:PointTransOfGenBurgersKdVEqsGaugeAr1A10d}, 
where the parameter functions~$T$, $X^0$ and~$X^1$ may depend on arbitrary elements, 
and the partial derivatives of these functions are replaced 
by the corresponding total derivatives prolonged to the arbitrary elements of the class~$\bar{\mathcal A}_1$.
At the same time, these parameters satisfy the condition 
${\rm D}_xT={\rm D}_xX^0={\rm D}_xX^1=0$ 
with the prolonged total derivative operator~${\rm D}_x$.
This condition implies via splitting with respect to unconstrained derivatives of arbitrary elements 
that the parameters~$T$, $X^0$ and~$X^1$ are functions of~$t$ only.
Hence \smash{$\bar G^\sim_{\bar{\mathcal A}_1}=G$}.
\end{proof}

Therefore, the gauge $(C,A^1)=(1,0)$ is better than the gauge $(A^\EqOrd,A^1)=(1,0)$.

\begin{remark}\label{rem:GenEquivGroupOfGenBurgersKdVEqsGaugeAr1A10}
To the best of our knowledge,
Theorem~\ref{thm:GenEquivGroupOfGenBurgersKdVEqsGaugeAr1A10}
provides the first example for a generalized equivalence group
containing transformations
whose components for equation variables depend on a nonconstant arbitrary element. 
This is also an example of a generalized equivalence group being effective itself, 
and thus the corresponding class of differential equations admits 
a unique effective generalized equivalence group.
\end{remark}

A more complicated example of a generalized equivalence group
is given by the subclass~$\mathcal A_0$ of the class~\eqref{eq:GenBurgersKdVEqs}
singled out by the mere constraint $A^1=0$.
The $A^1$-component of equivalence transformations of the class~\eqref{eq:GenBurgersKdVEqs}
takes the form
\[
\tilde A^1=\frac{X_x}{T_t}A^1+\frac{X_x}{T_t}\frac{U^0}{U^1}C-\frac{X_t}{T_t}
-\sum_{j=2}^\EqOrd\tilde A^jX_x\left(\frac1{X_x}\partial_x\right)^{j-1}\frac1{X_x},
\]
where each~$\tilde A^j$, $j=2,\dots,\EqOrd$, is a combination of $A^i$, $i=j,\dots,\EqOrd$
with coefficients expressed via $T_t$ and derivatives of~$X$ with respect to~$x$.
Substituting the expression for~$U^0$ implied by the gauge $A^1=0$,
\[
U^0=\frac{X_tU^1}{X_xC}
+\frac{T_tU^1}{C}\sum_{j=2}^\EqOrd\tilde A^j\left(\frac1{X_x}\partial_x\right)^{j-1}\frac1{X_x},
\]
into the general form of admissible transformations of the class~\eqref{eq:GenBurgersKdVEqs} 
and neglecting the relation between~$A^1$ and~$\tilde A^1$,
we get the elements of the equivalence groupoid of the subclass~$\mathcal A_0$. 
Therefore, this subclass is not normalized in the usual sense, 
and its usual equivalence group is isomorphic to 
the subgroup of the group~$G^\sim_{\mbox{\tiny\eqref{eq:GenBurgersKdVEqs}}}$ 
singled out by constraining group parameters with $X_{xx}=X_t=0$.
Similarly to Theorem~\ref{thm:GenEquivGroupOfGenBurgersKdVEqsGaugeAr1A10},
we can consider the counterpart~$\bar{\mathcal A}_0$ of the subclass~$\mathcal A_0$, 
where the tuple of arbitrary elements $(A^0,A^2,\dots,A^\EqOrd,B,C)$
is formally extended 
with the derivatives~$C_t$, $C_k$, $A^j_t$ and~$A^j_k$, $j=2,\dots,\EqOrd$, $k=1,\dots,\EqOrd$.
Then the expressions for transformational parts of admissible transformations of~$\mathcal A_0$
and their relations between initial and transformed arbitrary elements 
including the prolongation to the above derivatives 
give the components of the transformations 
constituting a group~$G$, 
which is obviously an effective generalized equivalence group~$G$ of the class~$\bar{\mathcal A}_0$.
Thus, the class~$\bar{\mathcal A}_0$ is also normalized in the generalized sense. 

Note that the entire generalized equivalence group~\smash{$\bar G^\sim_{\bar{\mathcal A}_0}$} 
of the class~$\bar{\mathcal A}_0$ coincides with its effective generalized equivalence group~$G$.
Indeed, in view of the description of the equivalence groupoid of the subclass~$\mathcal A_0$, 
elements of~\smash{$\bar G^\sim_{\bar{\mathcal A}_0}$} 
are of the form similar to that of elements of~$G$, 
where the group parameters~$T$, $X$ and~$U^1$ may also depend on arbitrary elements of the subclass~$\mathcal A_0$, 
and their partial derivatives in~$t$ and~$x$ are replaced 
by the corresponding total derivatives prolonged to the arbitrary elements of the class~$\bar{\mathcal A}_0$. 
At the same time, the condition ${\rm D}_xT={\rm D}_xU^1=0$ 
with the prolonged total derivative operator~${\rm D}_x$ 
implies that the parameter functions~$T$ and~$U^1$ still depend at most on~$t$. 
The corresponding expression for~$U^0$ involves the derivative~${\rm D}_x^\EqOrd X$, 
and hence the transformation component for~$B$ necessarily contains 
the  derivative~${\rm D}_x^{2\EqOrd}X$. 
Splitting this component with respect to the $2\EqOrd$th order $x$-derivatives of all arbitrary elements, 
which are not constrained, 
we derive that in fact the parameter function~$X$ also does not depend on arbitrary elements.

\section{Preliminary analysis of Lie symmetries}\label{sec:DetEqsLieSymmetriesGenBurgersKdVEqs}

We compute the maximal Lie invariance group of an equation~$\mathcal L_\kappa$ 
from the class~\eqref{eq:GenBurgersKdVEqsGaugedSubclass} using the infinitesimal method. 
For this, we define the generators of one-parameter point symmetry groups of~$\mathcal L_\kappa$ through $Q=\tau\p_t+\xi\p_x+\eta\p_u$
with the components~$\tau$, $\xi$ and~$\eta$ depending on $(t,x,u)$. 
The infinitesimal invariance criterion reads
\[
 Q^{(\EqOrd)}\bigg(u_t+uu_x-\sum_{j=2}^\EqOrd A^ju_j-A^0u-B\bigg)=0 \quad\text{for all solutions of~$\mathcal L_\kappa$}.
\]
The $\EqOrd$th prolongation $Q^{(\EqOrd)}$ of the vector field~$Q$ is given by
$
 Q^{(\EqOrd)}=Q+\sum_{0<|\alpha|\leqslant \EqOrd}\eta^\alpha\p_{u_\alpha}.
$
Recall that $\alpha=(\alpha_1,\alpha_2)$ is a multi-index, 
$\alpha_1,\alpha_2\in\mathbb N\cup\{0\}$, $|\alpha|=\alpha_1+\alpha_2$, 
and $u_{\alpha}=\partial^{|\alpha|}u/\partial t^{\alpha_1}\partial x^{\alpha_2}$. 
The coefficients~$\eta^\alpha$ in the prolonged vector fields~$Q^{(\EqOrd)}$ 
are obtainable from the general prolongation formula~\cite[Theorem 2.36]{olve86Ay},
\[
\eta^\alpha=\mathrm D^\alpha\left(\eta-\tau u_t-\xi u_x\right)+\tau u_{\alpha+\delta_1}+\xi u_{\alpha+\delta_2},
\]
where 
$\mathrm D^\alpha=\mathrm D_t^{\alpha_1}\mathrm D_x^{\alpha_2}$, 
$\mathrm D_t=\partial_t+\sum_\alpha u_{\alpha+\delta_1}\partial_{u_\alpha}$ and 
$\mathrm D_x=\partial_x+\sum_\alpha u_{\alpha+\delta_2}\partial_{u_\alpha}$ 
are the operators of total differentiation with respect to~$t$ and~$x$, respectively, 
and $\delta_1=(1,0)$ and~$\delta_2=(0,1)$.

The infinitesimal invariance criterion yields
\begin{align}\label{eq:InfinitesimalInvarianceGenBurgersKdVEqs}
\begin{split}
 &\eta^{(1,0)}+u\eta^{(0,1)}+\eta u_x
 =\sum_{j=2}^\EqOrd\left[(\tau A^j_t+\xi A^j_x)u_j+A^j\eta^{(0,j)}\right]
 +(\tau A^0_t+\xi A^0_x)u+A^0\eta\\[-1ex]
 &\qquad{}+\tau B_t+\xi B_x,
\qquad\mbox{wherever}\qquad 
u_t+uu_x=\sum_{j=2}^\EqOrd A^ju_j+A^0u+B.
\end{split}
\end{align}

We have shown in Section~\ref{sec:EquivalenceGroupoidGenBurgersKdVEqs} 
that the class~\eqref{eq:GenBurgersKdVEqsGaugedSubclass} is normalized in the usual sense. 
Therefore, we know that the restrictions derived 
in Corollary~\ref{cor:OnEquivalenceAlgebraGenBurgersKdVEqs}
for the $(t,x,u)$-components of vector fields from the equivalence algebra~$\mathfrak g^\sim$ 
also hold for the components of infinitesimal symmetry generators. 
It is thus true that
$\tau=\tau(t)$, $\xi=\zeta(t)x+\chi(t)$, $\eta=(\zeta(t)-\tau_t(t))u+\zeta_t(t)x+\chi_t(t)$.

Substituting this restricted form of the coefficients of~$Q$ 
into the infinitesimal invariance criterion~\eqref{eq:InfinitesimalInvarianceGenBurgersKdVEqs}, we obtain
\begin{gather*}
\sum_{j=2}^\EqOrd(\tau A^j_t+(\zeta x+\chi)A^j_x+(\tau_t-j\zeta)A^j)u_j
+(\tau A^0_t+(\zeta x+\chi)A^0_x+\tau_tA^0)u\\
\qquad{}+\tau B_t+(\zeta x+\chi)B_x-(\zeta-2\tau_t)B+(\zeta_tx+\chi_t)A^0
=(2\zeta_t-\tau_{tt})u+\zeta_{tt}x+\chi_{tt}.
\end{gather*}
This equation can be split with respect to~$u$ and its spatial derivatives, resulting in the system
\begin{subequations}\label{eq:DeterminingEqsGenBurgersKdVEqs}
\begin{align}
 &\tau A^j_t+(\zeta x+\chi)A^j_x+(\tau_t-j\zeta)A^j=0,\quad j=2,\dots,\EqOrd,\label{eq:DeterminingEqsGenBurgersKdVEqs1}\\
 &\tau A^0_t+(\zeta x+\chi)A^0_x+\tau_tA^0=2\zeta_t-\tau_{tt},\label{eq:DeterminingEqsGenBurgersKdVEqs2}\\
 &\tau B_t+(\zeta x+\chi)B_x-(\zeta-2\tau_t)B+(\zeta_tx+\chi_t)A^0=\zeta_{tt}x+\chi_{tt}.\label{eq:DeterminingEqsGenBurgersKdVEqs3}
\end{align}
\end{subequations}

Since all of the determining equations~\eqref{eq:DeterminingEqsGenBurgersKdVEqs} essentially depend on the arbitrary elements~$\kappa$, 
they constitute the system of \emph{classifying equations} for Lie symmetries of equations from the class~\eqref{eq:GenBurgersKdVEqsGaugedSubclass}. 
Thus, solving the group classification problem for the class~\eqref{eq:GenBurgersKdVEqsGaugedSubclass} reduces to 
solving the classifying equations~\eqref{eq:DeterminingEqsGenBurgersKdVEqs} up to the equivalence induced by~$G^\sim$. 
Due to the structure of the determining equations~\eqref{eq:DeterminingEqsGenBurgersKdVEqs} we have proved the following proposition.

\begin{proposition}
The maximal Lie invariance algebra~$\mathfrak g_\kappa$ of the equation~$\mathcal L_\kappa$ from the class~\eqref{eq:GenBurgersKdVEqsGaugedSubclass} 
is spanned by the vector fields of the form~$Q=D(\tau)+S(\zeta)+P(\chi)$,
where the parameter functions~$\tau$, $\zeta$ and $\chi$ run through the solution set of the determining equations~\eqref{eq:DeterminingEqsGenBurgersKdVEqs}, and
\begin{align*}
D(\tau)=\tau\p_t-\tau_tu\p_u,\quad
S(\zeta)=\zeta x\p_x+(\zeta u+\zeta_t x)\p_u,\quad P(\chi)=\chi\p_x+\chi_t\p_u.
\end{align*}
\end{proposition}

\begin{proposition}
The kernel Lie invariance algebra~$\mathfrak g^\cap:=\bigcap_\kappa\mathfrak g_\kappa$ of equations 
from the class~\eqref{eq:GenBurgersKdVEqsGaugedSubclass} is trivial, 
that is,~$\mathfrak g^\cap=\{0\}$.
\end{proposition}

\begin{proof}
To derive the kernel of maximal Lie invariance algebras one assumes 
the arbitrary elements~$\kappa$ to vary. 
Then it is possible to split the determining equations~\eqref{eq:DeterminingEqsGenBurgersKdVEqs} 
also with respect to the arbitrary elements~$\kappa$ and their derivatives. 
This immediately yields $\tau=0$, $\zeta x+\chi=0$ and, 
since $\zeta$ and $\chi$ are functions of~$t$ only, it follows that $\zeta=\chi=0$.
\end{proof}

At this stage, it is appropriate to introduce the linear span
\[
 \mathfrak g_\spanindex:=\langle D(\tau),S(\zeta),P(\chi)\rangle,
\]
where the parameter functions $\tau$, $\zeta$ and $\chi$ 
run through the set of smooth functions of $t$. 
The nonzero commutation relations between vector fields spanning~$\mathfrak g_\spanindex$ 
are given by
\begin{gather*}
[D(\tau),D(\check\tau)]=D(\tau\check\tau_t-\check\tau\tau_t),\quad
[D(\tau),S(\zeta)]=S(\tau\zeta_t),\quad
[D(\tau),P(\chi)]=P(\tau\chi_t),\\
[S(\zeta),P(\chi)]=-P(\zeta\chi).
\end{gather*}
In view of these commutation relations, 
it is obvious that $\mathfrak g_\spanindex$ is a Lie algebra. 
Moreover, it is true that $\mathfrak g_\spanindex=\cup_\kappa\,\mathfrak g_\kappa$ 
since any of the vector fields $D(\tau)$, $S(\zeta)$ and $P(\chi)$ 
lies in $\mathfrak g_\kappa$ for some particular value of the arbitrary element~$\kappa$. 

Let us here and in the following denote by~$\pi$ the projection map from the joint space of $(t,x,u,\kappa)$ 
onto the space of the variables $(t,x,u)$ alone, 
i.e., $\pi(t,x,u,\kappa)=(t,x,u)$.
This map properly pushforwards vector fields from~$\mathfrak g^\sim$ 
and transformations from~$G^\sim$, and $\pi_*\mathfrak g^\sim=\mathfrak g_\spanindex$. 
This displays that in fact the class~\eqref{eq:GenBurgersKdVEqsGaugedSubclass} 
is strongly normalized in the usual sense~\cite{popo10Ay}.
Note that the normalization of the class~\eqref{eq:GenBurgersKdVEqsGaugedSubclass} in the usual sense
only implies that $\mathfrak g_\spanindex\subseteq\pi_*\mathfrak g^\sim$.
The spanning vector fields of~$\mathfrak g^\sim$, 
$\hat D(\tau)$, $\hat S(\zeta)$ and $\hat P(\chi)$, are pushforwarded by~$\pi$ 
to the spanning vector fields~$D(\tau)$, $S(\zeta)$ and $P(\chi)$ 
of~$\mathfrak g_\spanindex$, respectively. 
By means of the pushforward by~$\pi$, the adjoint action of~$G^\sim$ on~$\mathfrak g^\sim$ 
induces the action of $\pi_*G^\sim$ on~$\mathfrak g_\spanindex$ and, therefore, 
on the set of subalgebras of~$\mathfrak g_\spanindex$.  
Recall that a subalgebra~$\mathfrak s$ of $\mathfrak g_\spanindex$ is called appropriate if there exists a $\kappa$ such that $\mathfrak s=\mathfrak g_\kappa$. 
For any value of the arbitrary element~$\kappa$ and any transformation~$\mathcal T\in G^\sim$, 
the pushforward by $\pi_*\mathcal T$ maps 
the maximal Lie invariance algebra~$\mathfrak g_\kappa$ of the equation~$\mathcal L_\kappa$
onto the maximal Lie invariance algebra~$\mathfrak g_{\mathcal T\kappa}$
of the equation~$\mathcal L_{\mathcal T\kappa}$, 
and both the algebras~$\mathfrak g_\kappa$ and~$\mathfrak g_{\mathcal T\kappa}$ 
are included in~$\mathfrak g_\spanindex$. 
In other words, the action of $\pi_*G^\sim$ on~$\mathfrak g_\spanindex$ preserves 
the set of appropriate subalgebras of~$\mathfrak g_\spanindex$, 
and hence the group~$\pi_*G^\sim$ generates a well-defined equivalence relation on these subalgebras. 
As a result, we have proved the following proposition, which is the basis for the group classification of class~\eqref{eq:GenBurgersKdVEqsGaugedSubclass}.

\begin{proposition}
 The complete group classification of the class~\eqref{eq:GenBurgersKdVEqsGaugedSubclass} of gauged $(1{+}1)$-dimensional general Burgers--KdV equations of order $\EqOrd$ is obtained by classifying all appropriate subalgebras of the Lie algebra~$\mathfrak g_\spanindex$ under the equivalence relation generated by the action of~$\pi_* G^\sim$.
\end{proposition}

To efficiently carry out the group classification using the algebraic method, it is necessary to compute the adjoint actions of the transformations from $\pi_*G^\sim$
on the vector fields $Q$ from~$\mathfrak g_\spanindex$.

The adjoint actions of the transformations, $\varphi$, from $\pi_*G^\sim$ on the vector fields, $Q$, from $\mathfrak g_\spanindex$ are directly computable from the definition of the pushforward~$\varphi_*$ of~$Q$ by~$\varphi$,
\[
 \varphi_* Q=Q(T)\p_{\tilde t}+Q(X)\p_{\tilde x}+Q(U)\p_{\tilde u},
\]
see, e.g., \cite{bihl11Dy,bihl16a,card11Ay,kuru16a}. Here, the coefficients of~$\varphi_* Q$ are given in terms of the transformed variables obtained by substituting~$(t,x,u)=\varphi^{-1}(\tilde t,\tilde x,\tilde u)$
using the inverse transformation~$\varphi^{-1}$ of~$\varphi$.

In practice, this is done by considering the families of elementary transformations~$\mathcal D(T)$, $\mathcal S(X^1)$ and~$\mathcal P(X^0)$ from~$\pi_*G^\sim$, which follow from~\eqref{eq:EquivalenceTransformationsGenBurgersKdVEqsGaugeC1A10} by restricting all except one of the parameter functions~$T$, $X^1$ and $X^0$ to trivial values (which are $t$ for $T$, one for $X^1$ and zero for $X^0$). The nontrivial pushforwards of the spanning vector fields of~$\mathfrak g_\spanindex$ by these elementary transformations from~$\pi_*G^\sim$ are given by
\begin{align}\label{eq:AdjointActionsGenBurgersKdVEqs}
\begin{split}
 &\mathcal D_*(T)D(\tau)=\tilde D(\tau T_t),\quad \mathcal D_*(T)S(\zeta)=\tilde S(\zeta),\quad \mathcal D_*(T)P(\chi)=\tilde P(\chi),\\
 &\mathcal S_*(X^1)D(\tau) =\tilde D(\tau)+\tilde S\left(\tau \frac{X^1_t}{X^1}\right) ,\quad \mathcal S_*(X^1)P(\chi)=\tilde P(\chi X^1),\\
 & \mathcal P_*(X^0)D(\tau) =\tilde D(\tau)+\tilde P\left(\tau X^0_t\right), \quad \mathcal P_*(X^0)S(\zeta)=\tilde S(\zeta)-\tilde P(\zeta X^0).
\end{split}
\end{align}
Here the tildes over the right-hand side operators indicate that the given vector fields are expressed using the transformed variables, which also includes substituting $t=T^{-1}(\tilde t)$ for $t$, where $T^{-1}$ is the inverse function of~$T$.

\section{Properties of appropriate subalgebras}\label{sec:PropertiesOfAppropriateSubalgebrasGenBurgersKdVEqs}

An important step for the group classification of class~\eqref{eq:GenBurgersKdVEqsGaugedSubclass} is
to study properties of appropriate subalgebras of the algebra~$\mathfrak g_\spanindex$ to be classified.
In particular, we determine the maximum dimension of admitted Lie invariance algebras of equations from this class.
This is formulated in the following lemma.

\begin{lemma}\label{lem:DimOfLieInvAlgebraBurgersKdV}
For any tuple of arbitrary elements~$\kappa$, $\dim\mathfrak g_\kappa\leqslant 5$.
\end{lemma}

\begin{proof}
This statement follows directly from analyzing the solution space of the linear system~\eqref{eq:DeterminingEqsGenBurgersKdVEqs} 
with respect to~$\tau$, $\zeta$ and $\chi$ for a fixed tuple~$\kappa$. 
Denote by $\Omega_t\subseteq\mathbb R$ and $\Omega_x\subseteq\mathbb R$ open intervals on the $t$- and $x$-axes, 
such that the equation~$\mathcal L_\kappa$ is defined on the domain $\Omega_t\times \Omega_x$. Since~$A^\EqOrd\ne0$ by definition, 
we can resolve the equation~\eqref{eq:DeterminingEqsGenBurgersKdVEqs1} with $j=\EqOrd$ for~$\tau_t$ and fix $x_1\in\Omega_x$, yielding
\begin{subequations}\label{eq:DimOfLieInvAlgebraBurgersKdV}
\begin{equation}\label{eq:DimOfLieInvAlgebraBurgersKdV1}
 \tau_t=\EqOrd\zeta-\left(\tau \frac{A^\EqOrd_t}{A^\EqOrd}+(\zeta x+\chi)\frac{A^\EqOrd_x}{A^\EqOrd}\right)\bigg|_{x=x_1}=:R^1.
\end{equation}
Evaluating the classifying condition~\eqref{eq:DeterminingEqsGenBurgersKdVEqs3} at the two distinct points $x_2$ and $x_3$ from $\Omega_x$ and varying $t$, we obtain
 \begin{align*}
  &\zeta_{tt}x_2-\chi_{tt}=R^2,\quad \zeta_{tt}x_3-\chi_{tt}=R^3,
 \end{align*}
where $R^2$ and $R^3$ follow from substituting~$x_2$ and~$x_3$
into the left hand side of~\eqref{eq:DeterminingEqsGenBurgersKdVEqs3}, respectively. 
Due to $x_2$ and $x_3$ being different points, the above system can be written as
\begin{equation}\label{eq:DimOfLieInvAlgebraBurgersKdV2}
 \zeta_{tt}=\cdots,\quad \chi_{tt}=\cdots.
\end{equation}
\end{subequations}
If the arbitrary element~$\kappa$ is fixed, 
the system~\eqref{eq:DimOfLieInvAlgebraBurgersKdV} can be considered 
as a canonical system of linear ordinary differential equations 
in~$t$ for $\tau$, $\zeta$ and $\chi$. 
Its solution space is thus five-dimensional. 
Further conditions derived from the classifying equations~\eqref{eq:DeterminingEqsGenBurgersKdVEqs} can only reduce this solution space and hence it follows that $\dim\mathfrak g_\kappa\leqslant 5$.
\end{proof}

We introduce three integers related to the dimensions of certain subspaces
of the maximal Lie invariance algebra~$\mathfrak g_\kappa$ of the equation~$\mathcal L_\kappa$,
\begin{gather*}
k_1:=\dim\big(\mathfrak g_\kappa\cap\langle P(\chi)\rangle\big), \\
k_2:=\dim\big(\mathfrak g_\kappa\cap\langle S(\zeta),P(\chi)\rangle\big)-k_1, \\
k_3:=\dim\mathfrak g_\kappa-\dim\big(\mathfrak g_\kappa\cap\langle S(\zeta),P(\chi)\rangle\big)
    =\dim\mathfrak g_\kappa-k_1-k_2.
\end{gather*}
Although these integers depend on~$\kappa$,
in view of~\eqref{eq:AdjointActionsGenBurgersKdVEqs} it is obvious
that the dimensions of the subalgebras
$\mathfrak g_\kappa\cap\langle P(\chi)\rangle$ and $\mathfrak g_\kappa\cap\langle S(\zeta),P(\chi)\rangle$
as well as of the entire algebra~$\mathfrak g_\kappa$ are $G^\sim$-invariant, 
and thus the integers~$k$'s are also $G^\sim$-invariant.

\begin{lemma}
The integers $k_1$, $k_2$ and $k_3$ are $G^\sim$-invariant values,
i.e., they are the same for all $G^\sim$-equivalent equations
from the class~\eqref{eq:GenBurgersKdVEqsGaugedSubclass}.
\end{lemma}

\begin{proof}
Let $\mathcal T\in G^\sim$ transform $\mathcal L_\kappa$ to $\mathcal L_{\tilde\kappa}$.
The transformation $\pi_*\mathcal T$ pushforwards $\mathfrak g_\kappa$ onto $\mathfrak g_{\tilde\kappa}$.
At the same time, it preserves the spans~$\langle S(\zeta),P(\chi)\rangle$ and $\langle P(\chi)\rangle$.
This is why
$\dim\mathfrak g_\kappa=\dim\mathfrak g_{\tilde\kappa}$,
$\dim\mathfrak g_\kappa\cap\langle S(\zeta),P(\chi)\rangle
=\dim\mathfrak g_{\tilde\kappa}\cap\langle S(\zeta),P(\chi)\rangle$ and
$\dim\mathfrak g_\kappa\cap\langle P(\chi)\rangle
=\dim\mathfrak g_{\tilde\kappa}\cap\langle P(\chi)\rangle$.
\end{proof}

We now proceed to find the upper bounds for values of these integers,
which is proved in a way similar to Lemma~\ref{lem:DimOfLieInvAlgebraBurgersKdV}.

\begin{lemma}\label{lem:DimOfLieInvSubalgebraBurgersKdV1}
For any $\kappa$, we have $(k_1,k_2)\in\{(0,0),(0,1),(2,0)\}$.
\end{lemma}

\begin{proof}
For any vector field $Q=S(\zeta)+P(\chi)$ from the algebra~$\mathfrak g_\kappa$,
the parameter functions~$\zeta$ and~$\chi$ satisfy 
the system of classifying equations~\eqref{eq:DeterminingEqsGenBurgersKdVEqs} 
for the chosen tuple of arbitrary elements $\kappa=(A^0,A^2,\dots,A^\EqOrd,B)$ and $\tau=0$.

In particular, if $A^\EqOrd_x\ne0$, we can solve
the classifying equation~\eqref{eq:DeterminingEqsGenBurgersKdVEqs1} with $j=\EqOrd$
with respect to $\chi$ to obtain $\chi=(\EqOrd A^\EqOrd/A^\EqOrd_x-x)\zeta$.
After fixing a value $x=x_0$, this equation implies
that there exists a function~$f=f(t)$ such that $\chi=f\zeta$.
Then the classifying equation~\eqref{eq:DeterminingEqsGenBurgersKdVEqs2} implies that $2\zeta_t=(x+f)A^0_x\zeta$. 
Again fixing $x=x_0$, we hence derive an equation $\zeta_t=g\zeta$ with some function $g=g(t)$. 
Since the parameter functions~$\zeta$ and~$\chi$ of any vector field $Q=S(\zeta)+P(\chi)$ from~$\mathfrak g_\kappa$  
satisfy the same system $\chi=f\zeta$ and $\zeta_t=g\zeta$, 
we have $k_1=0$ and $k_2\leqslant1$ in this case.

If $A^\EqOrd_x=0$, then the classifying equation \eqref{eq:DeterminingEqsGenBurgersKdVEqs1} with $j=\EqOrd$
directly gives $\zeta=0$, i.e., $k_2=0$.
Suppose that $k_1>0$, i.e., there exists a nonzero~$\chi^1$ such that $P(\chi^1)\in\mathfrak g_\kappa$.
It then follows from the system~\eqref{eq:DeterminingEqsGenBurgersKdVEqs}
that $A^j_x=0$, $A^0_x=0$ and $\chi^1 B_x+\chi^1_tA^0=\chi^1_{tt}$.
Differentiating the last equation with respect to~$x$ results in $B_{xx}=0$.
Thus, we have $P(\chi)\in\mathfrak g_\kappa$ for any~$\chi$
from the two-dimensional solution space of the equation $\chi_{tt}=A^0\chi_t+B_x\chi$,
which means $k_1=2$. 
\end{proof}

The projection~$\varpi$ from the space of~$(t,x,u)$ on the space of~$t$ alone
properly pushforwards elements of $\mathfrak g_\spanindex$ 
according to $D(\tau)+S(\zeta)+P(\chi)\mapsto\tau\p_t$, 
and hence $\varpi_*\mathfrak g_\spanindex=\{\tau\p_t\}$, 
where $\tau$ runs through the set of smooth functions of~$t$.
The pushforward~$\varpi_* G^\sim$ of~$G^\sim$ by~$\varpi$ is also well defined.

\begin{lemma}\label{lem:OnOperatorsInvolvingTau}
The projection~$\varpi_*\mathfrak g_\kappa$ is a Lie algebra
for any tuple of arbitrary elements~$\kappa$, 
and $\dim\varpi_*\mathfrak g_\kappa=k_3\leqslant3$.
It is true that $\varpi_*\mathfrak g_\kappa\in\{0,\langle\p_t\rangle,\langle\p_t,t\p_t\rangle,\langle\p_t,t\p_t,t^2\p_t\rangle\}\bmod\varpi_* G^\sim$.
\end{lemma}

\begin{proof}
We show that $\varpi_*\mathfrak g_\kappa$ is indeed a Lie algebra.
Given $\tau^i\p_t\in \varpi_*\mathfrak g_\kappa$, $i=1,2$,
there exist $Q^i\in \mathfrak g_\kappa$ such that~$\varpi_*Q^i=\tau^i\p_t$.
For any constants $c_1$ and $c_2$, we have that
$c_1Q^1+c_2Q^2\in\mathfrak g_\kappa$
and thus $c_1\tau^1\p_t+c_2\tau^2\p_t=\varpi_*(c_1Q^1+c_2Q^2)\in\varpi_*\mathfrak g_\kappa$ .
This means that~$\varpi_*\mathfrak g_\kappa$ is indeed a linear space.
This space is closed under the Lie bracket of vector fields and, therefore, is a Lie algebra,
because $[\tau^1\p_t,\tau^2\p_t]=(\tau^1\tau^2_t-\tau^2\tau^1_t)\p_t=\varpi_*[Q^1,Q^2]\in\varpi_*\mathfrak g_\kappa$.

Since the pushforward~$\varpi_* G^\sim$ of~$G^\sim$ by the projection~$\varpi$ 
coincides with the (pseudo)group of local diffeomorphisms in the space of~$t$,
Lie's theorem can be invoked. 
It states that the maximum dimension of finite-dimensional Lie algebras of vector fields 
on the complex (resp.\ real) line is three. 
Up to local diffeomorphisms of the line,
these algebras are given by $\{0\}$, $\langle\p_t\rangle$, $\langle\p_t,t\p_t\rangle$ and~$\langle\p_t,t\p_t,t^2\p_t\rangle$.
\end{proof}

It follows that~$|k|:=k_1+k_2+k_3=\dim\mathfrak g_\kappa\leqslant5$.

\section{Group classification}\label{sec:GroupClassificationGenBurgersKdVEqs}

The main result of the paper is given by the following assertion.

\begin{theorem}\label{thm:GroupClassificationGenBurgersKdVEqs}
 A complete list of~$G^\sim$-inequivalent (and, therefore, $\mathcal G^\sim$-inequivalent)
 Lie symmetry extensions in the class~\eqref{eq:GenBurgersKdVEqsGaugedSubclass} is exhausted by the cases given in Table~\ref{tab:CompleteGroupClassificationBurgersKdVEquations}.
\end{theorem}

\begin{table}[!ht]
\begin{center}\newcounter{tbn}\setcounter{tbn}{-1}
\caption{Complete group classification of the class~\eqref{eq:GenBurgersKdVEqs} (resp.~\eqref{eq:GenBurgersKdVEqsGaugedSubclass}).
\label{tab:CompleteGroupClassificationBurgersKdVEquations}}
\def\arraystretch{1.6} 
\begin{tabular}{|c|c|l|}
\hline
no. & $\kappa$ & \hfil Basis of~$\mathfrak g_\kappa$\\
\hline
\refstepcounter{tbn}\thetbn\label{000}  & $A^j=A^j(t,x),\ A^0=A^0(t,x),\ B=B(t,x)$              & ---\\
\refstepcounter{tbn}\thetbn\label{001}  & $A^j=A^j(x),\ A^0=A^0(x),\ B=B(x)$                    & $D(1)$\\
\refstepcounter{tbn}\thetbn\label{002a} & $A^j=a_je^x,\ A^0=a_0e^x,\ B=be^{2x}$                 & $D(1),D(t)-P(1)$\\
\refstepcounter{tbn}\thetbn\label{002b} & $A^j=a_jx^j|x|^\nu,\ A^0=a_0|x|^\nu,\ B=bx|x|^{2\nu}$ & $D(1),D(t)-S(\nu^{-1})$, $\nu\ne0$\\ 
\refstepcounter{tbn}\thetbn\label{003}  & $A^j=a_jx^{j-2},\ A_0=0,\ B=bx^{-3}$                  & $D(1),D(t)+S(\frac12),D(t^2)+S(t)$\\
\refstepcounter{tbn}\thetbn\label{010a} & $A^j=\alpha^j(t)x^j,\ A^0=0,\ B=\beta(t)x$            & $S(1)$\\
\refstepcounter{tbn}\thetbn\label{010b} & $A^j=\alpha^j(t)x^j,\ A^0=1{+}2\ln|x|,\ B=\beta(t)x{-}x\ln^2|x|$                 & $S(e^t)$\\
\refstepcounter{tbn}\thetbn\label{011a} & $A^j=a_jx^j,\ A^0=0,\ B=bx$                           & $S(1),D(1)$\\
\refstepcounter{tbn}\thetbn\label{011b} & $A^j=a_jx^j,\ A^0=1+2\ln|x|,\ B=bx-x\ln^2|x|$         & $S(e^t),D(1)$\\
\refstepcounter{tbn}\thetbn\label{200}  & $A^j=\alpha^j(t),\ A^0=B=0$                           & $P(1),P(t)$\\
\refstepcounter{tbn}\thetbn\label{201}  & $A^j=a_j,\ A^0=a_0,\ B=bx$                            & $P(\chi^1),P(\chi^2),D(1)$\\
\refstepcounter{tbn}\thetbn\label{202}  & $A^\EqOrd=1,\ A^j=0,\,j\ne \EqOrd,\ A^0=B=0$                    & $P(1),P(t),D(1),D(t)+S(\frac1\EqOrd)$\\
                                        &                                                       & and, for~$\EqOrd=2$, $D(t^2)+S(t)$\\
\hline
\end{tabular}
\end{center}
{\footnotesize\looseness=-1
Here $j=2,\dots,\EqOrd$, $C=1$, $A^1=0$,
$D(\tau)=\tau\p_t-\tau_tu\p_u$,
$S(\zeta)=\zeta x\p_x+(\zeta u+\zeta_t x)\p_u$ and
$P(\chi)=\chi\p_x+\chi_t\p_u$.
The parameters~$\alpha^j$, $\alpha^0$ and~$\beta$ are arbitrary smooth functions of~$t$ with $\alpha^\EqOrd\ne0$.
The parameters~$a_j$, $a_0$ and~$b$ are arbitrary constants with $a_\EqOrd\ne0$.
In~Case~\ref{201}, the parameter functions~$\chi^1$ and~$\chi^2$
are linearly independent solutions of the equation $\chi_{tt}-a_0\chi_t-b\chi=0$,
and thus there are three cases for them depending on the sign of $\Delta:=a_0^2-4b$,

\hspace*{1em}(a) $\chi^1=e^{\lambda_1t}$, $\chi^2=e^{\lambda_2t}$ if $\Delta>0$, where $\lambda_{1,2}=(a_0\pm\sqrt{\Delta})/2$;

\hspace*{1em}(b) $\chi^1=e^{\mu t}$, $\chi^2=te^{\mu t}$ if $\Delta=0$, where $\mu=a_0/2$;

\hspace*{1em}(c) $\chi^1=e^{\mu t}\cos{\nu t}$, $\chi^2=e^{\mu t}\sin{\nu t}$ if $\Delta<0$, where $\mu=a_0/2$ and $\nu=\sqrt{-\Delta}/2$.

$a_\EqOrd=1\bmod G^\sim$ in Cases~\ref{002a}, \ref{002b}, \ref{003}, \ref{011a} and \ref{201}.
Moreover, in Case~\ref{002a} one of the constants~$a_j$'s with $j<\EqOrd$, $a_0$ or~$b$, if it is nonzero, can be set to $\pm1$ by shifts of~$x$.
See also Section~\ref{sec:GenBurgersKdVEqsWithSpace-DependentCoeffs} for the justification of gauging of parameters 
and Remark~\ref{rem:ConditionsOfMaximality} for the conditions 
under which the presented vector fields really span the corresponding maximal Lie invariance algebra.\par 
}
\end{table}

\begin{proof}
Lemmas~\ref{lem:DimOfLieInvSubalgebraBurgersKdV1}--\ref{lem:OnOperatorsInvolvingTau} imply that
any appropriate subalgebra~$\mathfrak s$ of~$\mathfrak g_\spanindex$ has a basis consisting of
\begin{enumerate}\itemsep=0ex
 \item $k_1$ vector fields $Q^i=P(\chi^i)$, $i=1,\dots,k_1$, with linearly independent~$\chi$'s,
 \item $k_2$ vector fields $Q^i=S(\zeta^i)+P(\chi^i)$, $i=k_1+1,\dots,k_1+k_2$, with nonzero~$\zeta$'s,
 \item $k_3$ vector fields $Q^i=D(\tau^i)+S(\zeta^i)+P(\chi^i)$, $i=k_1+k_2+1,\dots,|k|$, with linearly independent~$\tau$'s,
\end{enumerate}
where $(k_1,k_2)\in\{(0,0),(0,1),(2,0)\}$ and $k_3\leqslant3$.
The proof proceeds by separately investigating the cases associated with
the possible range of the tuple of invariant integers $(k_1,k_2,k_3)$.
For each possible value of $(k_1,k_2,k_3)$,
we start with the above form of basis vector fields~$Q$'s of~$\mathfrak s$
and simplify them as much as possible
using the adjoint actions of equivalence transformations presented in~\eqref{eq:AdjointActionsGenBurgersKdVEqs}
and linear recombination of~$Q$'s.
At the same time, we take into account the fact that $\mathfrak s$ is a Lie algebra,
i.e., it is closed with respect to the Lie bracket of vector fields,
$[Q^{i'},Q^{i''}]\in\langle Q^i, i=1,\dots,|k|\rangle$, $i',i''=1,\dots,|k|$.
This leads to constraints for the components of~$Q$'s only if
$k_2+k_3>1$ or $k_2+k_3=1$ and $k_1=2$.
Substituting the components of each simplified~$Q^i$
into the system of classifying equations~\eqref{eq:DeterminingEqsGenBurgersKdVEqs}
yields a system of equations in~$\kappa$
for which $\mathfrak g_\kappa\supset\mathfrak s$.
In total, we have $|k|$ such systems.
We unite them and simultaneously solve for the arbitrary elements~$\kappa$.
The compatibility of the joint system with respect to~$\kappa$
may imply additional constraints for the components of~$Q$'s.
The expression obtained for~$\kappa$ can be simplified
by equivalence transformations whose projected adjoint actions preserve~$\mathfrak s$.
Except Cases~\ref{000} and~\ref{001}, 
the condition for~$\kappa$ with $\mathfrak g_\kappa=\mathfrak s$ 
is obtained by the negation of the corresponding condition represented in Remark~\ref{rem:ConditionsOfMaximality}.

We have to consider the following cases:

\medskip

\noindent$\boldsymbol{k_1=k_2=0.}$
Here $\dim\mathfrak s=k_3$.
For $k_3>1$,
in view of Lemma~\ref{lem:OnOperatorsInvolvingTau}
we can use the simplified form of~$Q^i$, $i=1,\dots,k_3-1$,
derived in the subcase with the preceding value of~$k_3$.

\medskip

\noindent$k_3=0$.
We obtain the general Case~\ref{000},
where the algebra~$\mathfrak s$ coincides with the kernel algebra $\mathfrak g^\cap=\{0\}$
and thus there are no constraints for~$\kappa$.

\medskip

\noindent $k_3=1$.
A basis of~$\mathfrak s$ consists of a single vector field~$Q^1$ with $\tau^1\ne0$.
Successively using the adjoint actions $\mathcal D_*(T)$, $\mathcal S_*(X^1)$ and $\mathcal P_*(X^0)$
given in~\eqref{eq:AdjointActionsGenBurgersKdVEqs} for appropriate functions $X^0$, $X^1$ and~$T$,
this vector field can be mapped to $Q^1=D(1)$.
The system of classifying equations~\eqref{eq:DeterminingEqsGenBurgersKdVEqs}
with the components of~$Q^1$ evidently is $A^j_t=A^0_t=B_t=0$,
which results in Case~\ref{001} of Table~\ref{tab:CompleteGroupClassificationBurgersKdVEquations}.

\medskip

\noindent $k_3=2$.
Modulo $\pi_*G^\sim$-equivalence, basis elements of~$\mathfrak s$ take the form $Q^1=D(1)$ and $Q^2=D(t)+S(\zeta^2)+P(\chi^2)$.
The condition $[Q^1,Q^2]=D(1)+S(\zeta^2_t)+P(\chi^2_t)\in\langle Q^1,Q^2\rangle$ requires
that $\zeta^2_t=\chi^2_t=0$ and thus $\zeta^2,\chi^2=\const$.
Substituting the components of~$Q^1$ and then~$Q^2$ into
the classifying equations~\eqref{eq:DeterminingEqsGenBurgersKdVEqs}
and rearranging leads to the system
\begin{gather}\label{eq:SystemForArbitraryElements002}
\begin{split}
&A^j_t=0,\quad (\zeta^2x+\chi^2)A^j_x+(1-j\zeta^2)A^j=0,\\
&A^0_t=0,\quad (\zeta^2x+\chi^2)A^0_x+A^0=0,\\
&B_t=0,\quad   (\zeta^2x+\chi^2)B_x+(2-\zeta^2)B=0.
\end{split}
\end{gather}
If $\zeta^2=0$, in view of the condition $A^\EqOrd\ne0$ the equation $(\zeta^2x+\chi^2)A^\EqOrd_x+(1-r\zeta^2)A^\EqOrd=0$
implies that $\chi^2\ne0$ and we can set $\chi^2=-1\bmod G^\sim$;
otherwise, we can set $\chi^2=0$ using $\mathcal P_*(-\nu\chi^2)$, where $\nu:=-1/\zeta^2$, which preserves~$Q^1$.
Integrating the system~\eqref{eq:SystemForArbitraryElements002} for each of the subcases,
we respectively obtain Cases~\ref{002a} and~\ref{002b}, where $a_\EqOrd\ne0$ and thus $a_\EqOrd=1\bmod G^\sim$. 
In Case~\ref{002a}, one of the constants~$a_j$'s with $j<\EqOrd$, $a_0$ or~$b$, if it is nonzero, can be set to $\pm1$ by shifts of~$x$.

\medskip

\noindent$k_3=3.$
Modulo $\pi_*G^\sim$-equivalence and linearly combining basis elements,
we can assume that $Q^1=D(1)$, $Q^2=D(t)+S(\zeta^2)+P(\chi^2)$ and $Q^3=D(t^2)+S(\zeta^3)+P(\chi^3)$.
The completeness of the algebra~$\mathfrak s$ with respect to the Lie bracket of vector fields
implies that $[Q^1,Q^2]=Q^1$, $[Q^1,Q^3]=2Q^2$ and $[Q^2,Q^3]=Q^3$.
As in the previous case, the first commutation relation holds only if $\zeta^2,\chi^2=\const$.
The second commutation relation expands to
\[
2D(t)+S(\zeta^3_t)+P(\chi^3_t)=2D(t)+2S(\zeta^2)+2P(\chi^2),
\]
and thus $\zeta^3_t=2\zeta^2$, $\chi^3_t=2\chi^2$.
Integrating these two equations yields $\zeta^3=2\zeta^2t+\zeta^{30}$ and $\chi^3=2\chi^2t+\chi^{30}$,
where $\zeta^{30}$ and~$\chi^{30}$ are constants.
Next, commuting $Q^2$ and $Q^3$ yields
\[
 [Q^2,Q^3]=D(t^2)+S(2\zeta^2t)+P(2\chi^2t+\zeta^{30}\chi^2-\zeta^2\chi^{30}).
\]
We thus have $[Q^2,Q^3]=Q^3$ if and only if $\zeta^{30}=0$ and $(1+\zeta^2)\chi^{30}=0$.
Considering the classifying equations~\eqref{eq:DeterminingEqsGenBurgersKdVEqs}
for $Q^1$, $Q^2$ and~$Q^3$, we obtain the system~\eqref{eq:SystemForArbitraryElements002}
supplemented by the equations
$\chi^{30}A^j_x=0$, $\chi^{30}A^0_x=4\zeta^2-2$ and $\chi^{30}B_x+(2\zeta^2x+3\chi^2)A^0=0$.
If $\chi^{30}\ne0$, then $\zeta^2=-1$, $A^j_x=0$,
and the equations $(1+j)A^j=0$ imply that $A^j=0$ for any~$j$,
which contradicts the condition $A^\EqOrd\ne0$.
Therefore, $\chi^{30}=0$, which requires that $\zeta^2=1/2$, as well as $A^0=0$. We can then set $\chi^2=0$ using $\mathcal P_*(2\chi^2)$.
The remaining determining equations are $A^j_x=(j-2)A^j$ and $B_x=-3B$,
which are readily integrated to Case~\ref{003} with $a_\EqOrd\ne0$, and hence $a_\EqOrd=1\bmod G^\sim$.

\medskip

\noindent $\boldsymbol{k_1=0, k_2=1.}$
Hence $Q^1=S(\zeta^1)+P(\chi^1)$, where $\zeta^1\ne0$.

\medskip

\noindent$k_3=0$.
We can use the adjoint action $\mathcal P_*(\chi^1/\zeta^1)$ to set $\chi^1=0$.
Moreover, multiplying $Q^1$ by a nonzero constant and changing~$t$ if $\zeta^1$ is not a constant,
we can gauge $\zeta^1$ to $\zeta^1=e^{\varepsilon t}$, where $\varepsilon\in\{0,1\}\bmod G^\sim$.
The classifying conditions~\eqref{eq:DeterminingEqsGenBurgersKdVEqs} with components of~$Q^1$
then imply the system consisting of the equations
$xA^j_x=jA^j$, $xA^0_x=2\varepsilon$ and $xB_x-B+\varepsilon xA^0=\varepsilon x$.
The general solution of this system is 
\[
A^j=\alpha^j(t)x^j,\quad 
A^0=\alpha^0(t)+2\varepsilon\ln|x|,\quad 
B=\beta(t)x-\varepsilon(\alpha^0(t)-1)x\ln|x|-\varepsilon x\ln^2|x|,
\]
where $\alpha^j$, $\alpha^0$ and~$\beta$ are arbitrary smooth functions of~$t$ with $\alpha^\EqOrd\ne0$.
The subgroups of the equivalence group~$G^\sim$
whose projections to the space $(t,x,u)$ preserve the appropriate subalgebras
$\mathfrak s=\langle S(1)\rangle$ and $\mathfrak s=\langle S(e^t)\rangle$
are singled out from~$G^\sim$ by the constraints $X^0=0$ and $(X^0,T_t)=(0,1)$, respectively.
Elements from these subgroups allow us to set
$(\alpha^0,\alpha^\EqOrd)=(0,1)$ if $\varepsilon=0$  
or $\alpha^0=1$ if $\varepsilon=1$, 
which corresponds to Case~\ref{010a} or~\ref{010b}.
In the latter case we have chosen the gauge~$\alpha^0=1$ instead of~$\alpha^0=0$
in order to make the corresponding values of the tuple of arbitrary elements~$\kappa$ simpler.

\medskip

\noindent$k_3=1$.
First we reduce, modulo $\pi_*G^\sim$-equivalence, the basis element~$Q^2$ with $\tau^2\ne0$
to the form $Q^2=D(1)$.
The Lie bracket $[Q^2,Q^1]=S(\zeta^1_t)+P(\chi^1_t)$ is in~$\mathfrak s$ provided
that the tuples $(\zeta^1_t,\chi^1_t)$ and $(\zeta^1,\chi^1)$ are linearly dependent,
i.e., $\zeta^1=c_1e^{\varepsilon t}$, $\chi^1=c_0e^{\varepsilon t}$
for some constants $c_0$, $c_1$ and~$\varepsilon$, where $c_1\ne0$.
Multiplying~$Q^1$ by $1/c_1$ and, if $\varepsilon\ne0$, using $\mathcal D_*(\varepsilon t)$,
we can set $c_1=1$ and $\varepsilon\in\{0,1\}$.
Since $\mathcal P_*(c_0)D(1)=\tilde D(1)$ and $\mathcal P_*(c_0)Q^1=\tilde S(e^{\varepsilon t})$,
we can finally reduce~$Q^1$ to the form $Q^1=S(e^{\varepsilon t})$.
Therefore, the corresponding complete system for the arbitrary elements includes
the system from the previous case as a subsystem supplemented by the equations $A^j_t=A^0_t=B_t=0$. 
The general solution of the complete system is of the same form as in the case $k_3=0$, 
where the functions~$\alpha^j$, $\alpha^0$ and~$\beta$ should be replaced by 
the arbitrary constants~$a_j$, $a_0$ and~$b$ with $a_\EqOrd\ne0$. 
Similarly to the case $k_3=0$, the appropriate subalgebras
$\mathfrak s=\langle S(1),D(1)\rangle$ and $\mathfrak s=\langle S(e^t),D(1)\rangle$
are preserved by projections of equivalence transformations with
$(T_t,(X^1_t/X^1)_t)=(1,0)$ and $(T_t,(e^{-t}X^1_t/X^1)_t)=(1,0)$,
which makes possible the gauges $(a_0,a_\EqOrd)=(0,1)$ or~$a_0=1$
if $\varepsilon=0$ or $\varepsilon=1$ 
and obviously results in Cases~\ref{011a} and~\ref{011b}, respectively.

\medskip

\noindent$k_3\geqslant 2$. This case cannot be realized.
Indeed, otherwise we would additionally have, modulo $\pi_*G^\sim$-equivalence,
the basis element $Q^3=D(t)+S(\zeta^3)+P(\chi^3)$
and would still be able to repeat the consideration of the case $k_3=1$,
deriving the expression $A^\EqOrd=a_\EqOrd x^\EqOrd$ with $a_\EqOrd\ne0$.
At the same time, the classifying equation~\eqref{eq:DeterminingEqsGenBurgersKdVEqs1}
for $j=\EqOrd$ with the components of~$Q^3$ and with the above~$A^\EqOrd$ reads $\chi^3A^\EqOrd_x+A^\EqOrd=0$,
which would imply $a_\EqOrd=0$, giving a contradiction.

\medskip

\noindent$\boldsymbol{k_1=2, k_2=0}.$
Here we have $Q^i=P(\chi^i)$, $i=1,2$.
The classifying equations~\eqref{eq:DeterminingEqsGenBurgersKdVEqs} in this case imply 
that $A^j_x=0$, $A^0_x=0$ and $\chi^iB_x+\chi^i_tA^0=\chi^i_{tt}$.
Differentiating this last equation with respect to $x$ leads to $B_{xx}=0$,  
i.e., $B=B^1(t)x+B^0(t)$, and therefore $\chi^i_{tt}=A^0\chi^i_t+B^1\chi^i$.

\medskip

\noindent$k_3=0$.
We apply $\mathcal S_*(1/\chi^1)$ to~$\mathfrak s$ and set $\chi^1=1$, which implies that $\chi^2_t\ne0$.
Then using the adjoint action $\mathcal D_*(\chi^2)$ we can set $\chi^2=t$.
The equations $\chi^i_{tt}=A^0\chi^i_t+B^1\chi^i$, $i=1,2$, then imply that $A^0=B^1=0$.
The components of the transformation $\mathcal P(X^0)$ for the arbitrary elements are $\tilde A^j=A^j$, $\tilde A^0=A^0$ and $\tilde B=B+X^0_{tt}$.
Thus, choosing $X^0_{tt}=-B^0$ allows us to gauge $B=0$, leading to Case~\ref{200}.

\medskip

\noindent$k_3=1$.
We start the simplification of basis elements of~$\mathfrak s$ from~$Q^3$ with $\tau^3\ne0$,
setting, modulo $\pi_*G^\sim$-equivalence, $Q^3=D(1)$.
It then follows from the classifying equations~\eqref{eq:DeterminingEqsGenBurgersKdVEqs}
with the components of~$Q^3$ that $A^j_t=A^0_t=B^1_t=B^0_t=0$.
We choose
\[
X^0=\dfrac{B^0}{B^1}\text{ \ if \ } B^1\neq0;\quad
X^0=\dfrac{B^0}{A^0}t\text{ \ if \ }B^1=0,\,A_0\neq0;\quad
X^0=-\dfrac{B^0}{2}t^2\text{ if }B^1=A_0=0.
\]
Since $\mathcal P_*(X^0) D(1)=\tilde D(1)+\tilde P(X^0_t)$,
$\mathcal P_*(X^0) P(\chi^i)=\tilde P(\chi^i)$, $i=1,2$,
and for the chosen value of~$X^0$ we have $X^0_t\in\langle\chi^1,\chi^2\rangle$,
the adjoint action $\mathcal P_*(X^0)$ preserves the algebra~$\mathfrak s$.
At the same time, the component of $\mathcal P(X^0)$ for~$B$ is $\tilde B=B+X^0_{tt}-A^0 X^0_t$
and therefore $\tilde B^0=B^0+X^0_{tt}-A^0X^0_t-B^1X^0=0$,
which yields Case~\ref{201}, where $a_\EqOrd=1\bmod G^\sim$.

\medskip

\noindent$k_3\geqslant2$.
We have one more basis element whose simplified form is $Q^4=D(t)+S(\zeta^4)+P(\chi^4)$.
Commuting $Q^3$ and $Q^4$ yields $[Q^3,Q^4]=D(1)+S(\zeta^4_t)+P(\chi^4_t)$,
which is in the algebra~$\mathfrak s$ provided that $\zeta^4=\const$.
The commutator of $Q^4$ with $Q^i$, $i=1,2$, gives
$
 [Q^4,Q^i]=P(t\chi^i_t-\zeta^4\chi^i),
$
which is in the algebra~$\mathfrak s$ only if $\chi^1=1$ and $\chi^2=t$
up to linearly combining $P(\chi^1)$ and $P(\chi^2)$.
It then follows that $A^0=B^1=0$.
We can let $B^0=0\bmod G^\sim$, and the classifying equation~\eqref{eq:DeterminingEqsGenBurgersKdVEqs3}
then implies that $\chi^4_{tt}=0$.
Upon linearly combining with $P(1)$ and $P(t)$ we can thus set $\chi^4=0$.
Moreover, the classifying equation~\eqref{eq:DeterminingEqsGenBurgersKdVEqs1} for $j=\EqOrd$
with the components of~$Q^4$ implies that $\zeta^2=1/\EqOrd$ and thus,
in view of the same equation for other $j$'s, $A^j=0$ if $j\ne \EqOrd$.
We can gauge $A^\EqOrd=1\bmod G^\sim$.
Now suppose that we also have a $Q^5=D(t^2)+S(\zeta^5)+P(\chi^5)$.
Then, the classifying equations~\eqref{eq:DeterminingEqsGenBurgersKdVEqs} for $A^\EqOrd$ and $A^0$ require
that $2t-\EqOrd\zeta^5=0$ and $2\zeta^5_t-2=0$.
This system is consistent only for $\EqOrd=2$, where $\zeta^5=t$.
This is why Case~\ref{202} splits depending on the value of~$\EqOrd$.
The equation~\eqref{eq:DeterminingEqsGenBurgersKdVEqs3} with $A^0=B=0$ implies that $\chi^5_{tt}=0$
and thus we can set $\chi^5=0$ upon linearly combining $Q^5$ with $Q^1$ and $Q^2$.
\end{proof}

\begin{corollary}
An $\EqOrd$th order evolution equation of the form~\eqref{eq:GenBurgersKdVEqs}
is reduced to the simplest form $u_t+uu_x=u_\EqOrd$ by a point transformation
if and only if the dimension of its maximal Lie invariance algebra is greater than three.
\end{corollary}

\begin{remark}
Case~\ref{011a} can be merged with Case~\ref{002b} as the particular subcase with $\nu=0$
if we choose another second basis element, $\tilde Q^2=-\nu Q^2=S(1)-D(\nu t)$.
\end{remark}

\begin{remark}\label{rem:ConditionsOfMaximality}
Each of the subalgebras of~$\mathfrak g$
whose bases are presented in the third column of Table~\ref{tab:CompleteGroupClassificationBurgersKdVEquations}
is really the maximal Lie invariance algebra for the general case of values of the arbitrary elements~$\kappa$
given in the same row.%
\noprint{
In some cases, it is not difficult to explicitly indicate the necessary and sufficient conditions for~$\kappa$
under which there is no additional Lie symmetry extension.
They can be found by substituting the corresponding expressions for $\kappa$'s components
into the system of classifying conditions~\eqref{eq:DeterminingEqsGenBurgersKdVEqs}.
Thus,
there are no constraints in Cases~\ref{002a}, \ref{202}, \ref{011a} and~\ref{011b};
$a_0\neq0$ if $\nu=-2$ and $(a_2,\dots,a_{\EqOrd-1},b)\neq(0,\dots,0,0)$ if $\nu=-\EqOrd$ with $\EqOrd>2$ in Case~\ref{002b};
$\EqOrd>2$ or $b\ne0$ in Case~\ref{003};
$(a_2,\dots,a_{\EqOrd-1})\neq(0,\dots,0)$ or $(\EqOrd-2)^2b\ne(\EqOrd-1)a_0^2$ in Case~\ref{201};
$(\alpha^2_t,\dots,\alpha^\EqOrd_t,\beta_t)\neq(0,\dots,0,0)$ in Cases~\ref{010a} and~\ref{010b}.
}
For most of the classification cases, it is not difficult to explicitly indicate the necessary and sufficient conditions for~$\kappa$
under which there is an additional Lie symmetry extension.
They can be found by substituting the corresponding expressions for $\kappa$'s components
into the system of classifying conditions~\eqref{eq:DeterminingEqsGenBurgersKdVEqs}.
Thus, these conditions are
$a_0=0$ if $\nu=-2$ or $a_j=b=0$ with $j\ne \EqOrd$ if $\nu=-\EqOrd$ with $\EqOrd>2$ in Case~\ref{002b};
$\EqOrd=2$ and $b=0$ in Case~\ref{003};
$\alpha^j_t=\beta_t=0$ in Cases~\ref{010a} and~\ref{010b};
$a_j=0$ for $j\ne \EqOrd$ and $(\EqOrd-2)^2b=(\EqOrd-1)a_0^2$ in Case~\ref{201}.
There are no additional extensions in Cases~\ref{002a}, \ref{011a}, \ref{011b} and~\ref{202}.
An additional extension exists for Case~\ref{200} if and only if
the parameter functions~$\alpha^j$'s satisfy the system
$(\zeta^1t^2+\tau^1t+\tau^0)\alpha^j_t+(2\zeta^1t+\tau^1-j(\zeta^1t+\zeta^0))\alpha^j=0$
for some constants~$\zeta^0$, $\zeta^1$, $\tau^0$ and~$\tau^1$.
Another form of this condition, which is convenient for checking, is that
\[
\left(\frac{\alpha^j}{\alpha^\EqOrd}\right)_t=(\EqOrd-j)f\frac{\alpha^j}{\alpha^\EqOrd},\quad
\EqOrd\alpha^\EqOrd\alpha^j_t-j\alpha^j\alpha^\EqOrd_t=(\EqOrd-j)g\alpha^j\alpha^\EqOrd,
\]
for some functions~$f=f(t)$ and~$g=g(t)$ constrained by $g_{tt}+3gg_t+g^3=0$, $2(f_t+gf)=g_t+g^2$.
See also the claim on the most complicated Case~\ref{001} in Remark~\ref{rem:OnGenCaseofGenBurgersKdVEqsWithSpace-DependentCoeffs}.
\end{remark}

\begin{remark}\label{rem:OnCompletenessOfRestrictionsForAppropriateSubalgebras}
The proof of Theorem~\ref{thm:GroupClassificationGenBurgersKdVEqs} shows 
that the system of restrictions for appropriate subalgebras of~$\mathfrak g_\spanindex$ 
presented in Lemmas~\ref{lem:DimOfLieInvAlgebraBurgersKdV}, \ref{lem:DimOfLieInvSubalgebraBurgersKdV1} 
and~\ref{lem:OnOperatorsInvolvingTau} is not exhaustive. 
It can be completed by the conditions restricting the value set of~$k_3$ 
depending on values of~$k_1$ and~$k_2$,
\begin{gather*}
k_3\in\{0,1\}\quad \mbox{if}\quad k_2=1;\\  
\mbox{if}\quad k_1=2,\quad \mbox{then}\quad 
k_3\in\{0,1,2\}\quad \mbox{for}\quad \EqOrd>2
\quad \mbox{and}\quad 
k_3\in\{0,1,3\}\quad \mbox{for}\quad \EqOrd=2. 
\end{gather*}
For each value of the tuple $k=(k_1,k_2,k_3)$ satisfying the extended set of restrictions, 
there exists an appropriate subalgebra of~$\mathfrak g_\spanindex$ admitting this value of~$k$.
\end{remark}


\section{Alternative classification cases}\label{sec:AlternativeClassificationCases}

There are various possibilities for choosing representatives
in equivalence classes of pairs $(\mathfrak s, \{\mathcal L_\kappa\mid\mathfrak g_\kappa=\mathfrak s\})$,
where $\mathfrak s$ is an appropriate subalgebra of~$\mathfrak g$.
We have tried to simplify the representation of a pair by paying more attention to the pair's second entry.
In most cases the optimal choice is obvious 
and coincides with the selection carried out for Table~\ref{tab:CompleteGroupClassificationBurgersKdVEquations}.
At the same time, there are other options for the proof and representation of results 
in the cases $(k_1,k_2,k_3)=(2,0,1)$ and $(k_1,k_2)=(0,1)$.    

\medskip

\noindent$\boldsymbol{(k_1,k_2)=(0,1).}$
We follow the proof of Theorem~\ref{thm:GroupClassificationGenBurgersKdVEqs}
but reduce the basis vector fields to another form.
If $k_3=0$, we set, modulo $G^\sim$-equivalence, $Q^1=S(t^\varepsilon)$
instead of $Q^1=S(e^{\varepsilon t})$, where still $\varepsilon\in\{0,1\}$.
In the case $k_3=1$, for~$Q^2$ we choose the form $Q^2=D(t)$
in order to be able to set $Q^1=S(t^\varepsilon)$ again.
For $\varepsilon=1$, this gives the following alternative cases,
which are related to the corresponding cases of Table~\ref{tab:CompleteGroupClassificationBurgersKdVEquations}
by the equivalence transformation with $T=e^{\varepsilon t}$, $X^1=1$ and~$X^0=0$:

\begin{center}%
\def\arraystretch{1.5}\tabcolsep=1ex
\begin{tabular}{|c|c|l|}
\hline
\,$\tilde{\mbox{\ref{010b}}}$\, & $A^j=\alpha^j(t)x^j,\ A^0=\alpha^0(t)+2t^{-1}\ln|x|,$  & \\
                                & $B=\beta(t)x-\alpha^0(t)t^{-1}x\ln|x|-xt^{-2}\ln^2|x|$ & $S(t)$\\
\,$\tilde{\mbox{\ref{011b}}}$\, & $A^j=a_jt^{-1}x^j,\ A^0=a_0t^{-1}+2t^{-1}\ln|x|,$      & \\
                                & $\hspace*{2em} B=bt^{-2}x-a_0t^{-2}x\ln|x|-xt^{-2}\ln^2|x|$\hspace*{2em}                & $S(t),D(t)$\hspace*{12em}\\
\hline
\end{tabular}
\end{center}

\noindent$\boldsymbol{(k_1,k_2,k_3)=(2,0,1)}.$
Instead of setting $Q^3=D(1)$ in the proof for this case,
we can simplify the basis elements~$Q^1$ and~$Q^2$ to~$P(1)$ and~$P(t)$ as in the case $(k_1,k_2,k_3)=(2,0,0)$.
Then readily $A^j_x=0$, $A^0=0$ and, modulo $G^\sim$-equivalence, $B=0$.
Consider the subclass~$\mathcal K$ of equations from the class~\eqref{eq:GenBurgersKdVEqsGaugedSubclass}
with values of~$\kappa$ satisfying the above constraints,
\[\mathcal K=\{\mathcal L_\kappa\mid A^j_x=0,A^0=0,B=0,A^\EqOrd\ne0\}.\]
The subclass~$\mathcal K$ turns out to be normalized
with respect to its usual equivalence group~$G^\sim_{\mathcal K}$, which is finite-dimensional and
consists of the transformations in the space of $(t,x,u,A^2,\dots,A^\EqOrd)$
whose components are of the form~\eqref{eq:PointTransformationBetweenGenBurgersKdVEqsGaugeC1A10a}--\eqref{eq:PointTransformationBetweenGenBurgersKdVEqsGaugeC1A10b}
with
\[
T=\frac{c_1t+c_2}{c_3t+c_4},\quad
X^1=\frac{c_5}{c_3t+c_4},\quad
X^0=\frac{c_6t+c_7}{c_3t+c_4},
\]
where $c_1$, \dots, $c_7$ are arbitrary constants with $(c_1c_4-c_2c_3)c_5\ne0$ that are defined up to a nonzero multiplier.
Hence the equivalence algebra~$\mathfrak g^\sim_{\mathcal K}$ of~$\mathcal K$ is spanned by
$\hat D(1)$, $\hat D(t)$, $\hat D(t^2)+\hat S(t)$, $\hat S(1)$, $\hat P(1)$ and $\hat P(t)$.
The kernel Lie invariance algebra of~$\mathcal K$ is $\mathfrak g^\cap_{\mathcal K}=\langle P(1),P(t)\rangle$.
The restriction of $G^\sim$-equivalence to the subclass~$\mathcal K$ coincides with $G^\sim_{\mathcal K}$-equivalence.
This is why, up to $G^\sim$-equivalence, we can assume that in this case
the appropriate subalgebra~$\mathfrak s$ is spanned by $Q^1=P(1)$, $Q^2=P(t)$ and $Q^3=D(\tau^3)+S(\zeta^3)$,
and there are three cases for $(\tau^3,\zeta^3)$ and~$(A^2,\dots,A^\EqOrd)$:

(a)  $\tau^3=1$, $\zeta^3=\sigma\in\{0,1\}$, $A^j=a_je^{j\sigma t}$;

(b)  $\tau^3=t$, $\zeta^3=\sigma=\const\geqslant0$, $A^j=a_jt^{-1}|t|^{j\sigma}$;

(c)  $\tau^3=t^2+1$, $\zeta^3=\sigma=\const\geqslant0$, $A^j=a_j(t^2+1)^{j/2-1}e^{j\sigma\arctan t}$.

\noindent
Therefore, instead of the single Case~\ref{201} we have the following three cases:

\begin{center}%
\def\arraystretch{1.5}\tabcolsep=1ex
\begin{tabular}{|c|c|l|}
\hline
\,$\tilde{\mbox{\ref{201}a}}$\, &       $A^j=a_je^{j\sigma t},\ A^0=B=0$                                   & $P(1),P(t),D(1)+S(\sigma)$\\
  $\tilde{\mbox{\ref{201}b}}$   &       $A^j=a_jt^{-1}|t|^{j\sigma},\ A^0=B=0$                             & $P(1),P(t),D(t)+S(\sigma)$\\
  $\tilde{\mbox{\ref{201}c}}$   & \quad $A^j=a_j(t^2+1)^{j/2-1}e^{j\sigma\arctan t},\ A^0=B=0$\hspace*{1em}& $P(1),P(t),D(t^2+1)+S(t+\sigma)$\hspace*{2em}\\
\hline
\end{tabular}
\\[1ex]
{\footnotesize 
Modulo $G^\sim$-equivalence, $\sigma\in\{0,1\}$, $\sigma\geqslant0$ and $\sigma\geqslant0$ 
in Cases~$\tilde{\mbox{\ref{201}a}}$, $\tilde{\mbox{\ref{201}b}}$ and~$\tilde{\mbox{\ref{201}c}}$, respectively. \hfill\null
}
\end{center}

The advantage of this form for Case~\ref{201} is that the vector fields~$Q^1$ and~$Q^2$ then have the evident interpretation as
generators of translations with respect to the space variable~$x$ and Galilean boosts, respectively.
Meanwhile the forms of $Q^3$ and, especially, $A^j$ become more complicated. 
This is why we chose the previous form of Case~\ref{201} for Table~\ref{tab:CompleteGroupClassificationBurgersKdVEquations}.
Subcases of this case are related to subcases of the alternative Case~$\tilde{\mbox{\ref{201}}}$ by equivalence transformations
of the form~\eqref{eq:EquivalenceTransformationsGenBurgersKdVEqsGaugeC1A10}, where $X^0=0$ and
\[\arraycolsep=0ex
\begin{array}{llll}
\mbox{\ref{201}a}\to\tilde{\mbox{\ref{201}b}}\colon\quad&T=e^{(\lambda_2-\lambda_1)t},\quad& X^1=e^{-\lambda_1t},     \quad&\sigma=-\lambda_1/(\lambda_2-\lambda_1);\\[1ex]
\mbox{\ref{201}b}\to\tilde{\mbox{\ref{201}a}}\colon\quad&T=t,                         \quad& X^1=e^{-\mu t},          \quad&\sigma=-\mu;\\[1ex]
\mbox{\ref{201}c}\to\tilde{\mbox{\ref{201}c}}\colon\quad&T=\tan\nu t,                 \quad& X^1=e^{-\mu t}/\cos\nu t,\quad&\sigma=-\mu/\nu.
\end{array}
\]

\section{Equations with time-dependent coefficients}\label{sec:GenBurgersKdVEqsWithTime-DependentCoeffs}

To study Lie symmetries of equations from the class~\eqref{eq:GenBurgersKdVEqs}
with coefficients depending at most on~$t$, it is again convenient to start with a wider class,
which is the subclass~$\mathcal K_0$ of the class~\eqref{eq:GenBurgersKdVEqs}
singled out by the constraint $C_x=0$ (resp.\ $A^\EqOrd_x=0$)
implying $X_{xx}=0$ for admissible transformations.

\begin{proposition}\label{pro:EquivGroupoidOfGenBurgersKdVEqsGaugeCx0OrArx0}
The class~$\mathcal K_0$ is normalized in the usual sense.
Its usual equivalence group is constituted by the transformations of the form
\begin{gather*}
\tilde t=T(t),\quad \tilde x=X^1(t)x+X^0(t),\quad \tilde u=U^1(t)u+U^0(t,x),
\\
\tilde A^j=\frac{(X^1)^j}{T_t}A^j,\quad
\tilde A^1=\frac{X^1}{T_t}\left(A^1+\frac{U^0}{U^1}C-\frac{X^1_tx+X^0_t}{X^1}\right),\quad
\tilde A^0=\frac{1}{T_t}\left(A^0+\frac{U^1_t}{U^1}+\frac{U^0_x}{U^1}C\right),
\\
\tilde B=\frac{U^1}{T_t}B+\frac{U^0_t}{T_t}
+\frac{U^0_x}{T_t}\left(\frac{U^0}{U^1}C-\frac{X^1_tx+X^0_t}{X^1}\right)
-\sum_{k=0}^\EqOrd\frac{U^0_k}{(X^1)^k}\tilde A^k,\quad
\tilde C=\frac{X^1}{T_tU^1}C,
\end{gather*}
where $j=2,\dots,\EqOrd$,
and $T=T(t)$, $X^1=X^1(t)$, $X^0=X^0(t)$, $U^1=U^1(t)$ and $U^0=U^0(t,x)$ are arbitrary smooth functions of their arguments
with $T_tX^1U^1\ne0$.
\end{proposition}

Consider the subclass~$\mathcal K_1$ obtained by attaching the constraints
$A^0_x=0$, $A^1_{xx}=0$, $A^j_x=0$, $j=2,\dots,\EqOrd$, $C_x=0$ and $B_{xx}=0$ 
to the auxiliary system for arbitrary elements.
It is also normalized in the usual sense and its usual equivalence group is 
the subgroup of the usual equivalence group~$G^\sim_{\mathcal K_0}$ of the class~$\mathcal K_0$
that is associated with the constraint $U^0_{xx}=0$, i.e., $U^0=U^{01}(t)x+U^{00}(t)$.
Note that we can reparameterize the class~$\mathcal K_1$ by
representing $B=B^1(t)x+B^0(t)$, $A^1=A^{11}(t)x+A^{10}(t)$ 
and assuming the coefficients~$B^1$, $B^0$, $A^{11}$ and~$A^{10}$ as arbitrary elements instead of~$B$ and~$A^1$. 
The transformation component for~$B$ simplifies to 
\[
\tilde B=\frac{U^1}{T_t}B+\frac{U^0_t}{T_t}-\frac{U^0_x}{T_t}A^1
-\frac{U^0}{T_t}\left(A^0+\frac{U^1_t}{U^1}+\frac{U^0_x}{U^1}C\right).
\]

The next intermediate subclass~$\mathcal K_2$ is singled out 
by strengthening the constraint for~$A^1$ to $A^1_x=0$. 
In fact, this can be realized by gauging~$A^1$ in the class~$\mathcal K_0$ 
up to $G^\sim_{\mathcal K_0}$-equivalence. The transformational properties of the subclass~$\mathcal K_2$ 
are similar to those of the subclasses studied in Section~\ref{sec:AlternativeGauges}, 
which is related to bringing the constraint for~$A^1$ into the foreground 
among the constraints or gauges for other arbitrary elements; 
see further discussions below. 
Since the arbitrary element~$C$ is still not gauged to one, 
it parameterizes the $u$-component of admissible transformations in~$\mathcal K_2$, 
$
U^{01}=X^1_tU^1/(X^1C),
$
and this fact can again be interpreted in terms of generalized equivalence group. 

\begin{subequations}\label{eq:AdmTransGenBurgersKdVEqsSubclassK2}

\begin{theorem}\label{thm:EquivGroupoidOfGenBurgersKdVEqsSubclassK2}
The equivalence groupoid of the subclass~$\mathcal K_2$ of the class~\eqref{eq:GenBurgersKdVEqs}
singled out by the constraints $A^k_x=0$, $k=0,\dots,\EqOrd$, $C_x=0$ and $B_{xx}=0$ 
consists of the triples $(\theta,\tilde\theta,\varphi)$'s, 
where the point transformation~$\varphi$ is of the form
\begin{gather}\label{eq:AdmTransGenBurgersKdVEqsSubclassK2a}
\tilde t=T,\quad 
\tilde x=X^1x+X^0,\quad 
\tilde u=U^1u+\frac{X^1_tU^1}{X^1C}x+U^{00}, 
\end{gather}
the arbitrary-element tuples~$\theta$ and~$\tilde\theta$ are related according to
\begin{gather}\label{eq:AdmTransGenBurgersKdVEqsSubclassK2b}
\tilde A^j=\frac{(X^1)^j}{T_t}A^j,\quad
\tilde A^1=\frac{X^1}{T_t}\left(A^1+\frac{U^{00}}{U^1}C-\frac{X^0_t}{X^1}\right),\quad
\tilde A^0=\frac{1}{T_t}\left(A^0+\frac{U^1_t}{U^1}+\frac{X^1_t}{X^1}\right),
\\\label{eq:AdmTransGenBurgersKdVEqsSubclassK2c}
\tilde B=
\frac{U^1}{T_t}B+\left(\frac{X^1_tU^1}{X^1C}\right)_t\frac x{T_t}+\frac{U^{00}_t}{T_t}-\frac{X^1_tU^1}{X^1C}\frac{A^1}{T_t}
-\left(\frac{X^1_tU^1}{X^1C}x+U^{00}\right)\tilde A^0,
\noprint{
\tilde B=\frac{U^1}{T_t}B+\frac{U^0_t}{T_t}-\frac{U^0_x}{T_t}A^1
-\frac{U^0}{T_t}\left(A^0+\frac{U^1_t}{U^1}+\frac{X^1_t}{X^1}\right),\quad
}
\\\label{eq:AdmTransGenBurgersKdVEqsSubclassK2d}
\tilde C=\frac{X^1}{T_tU^1}C,
\end{gather}
with $j=2,\dots,\EqOrd$,
and $T=T(t)$, $X^1=X^1(t)$, $X^0=X^0(t)$, $U^1=U^1(t)$ and $U^{00}=U^{00}(t)$ 
are arbitrary smooth functions of~$t$ with $T_tX^1U^1\ne0$.
\end{theorem}

The usual equivalence group~$G^\sim_{\mathcal K_2}$ of the subclass~$\mathcal K_2$ 
is constituted by the transformations~\eqref{eq:AdmTransGenBurgersKdVEqsSubclassK2a}--\eqref{eq:AdmTransGenBurgersKdVEqsSubclassK2d} 
additionally satisfying the constraint~$X^1_t=0$. 
Hence it is clear that the subclass~$\mathcal K_2$ is not normalized in the usual sense. 

The equation~\eqref{eq:AdmTransGenBurgersKdVEqsSubclassK2c} hints that
the proper treatment of the related generalized equivalence group 
within the framework of point transformations needs
considering the derivative~$C_t$ as an additional arbitrary element~$Z^0$
and prolonging the relation~\eqref{eq:AdmTransGenBurgersKdVEqsSubclassK2d} to~$Z^0$ as a derivative of~$C$,  
\begin{gather}\label{eq:AdmTransGenBurgersKdVEqsSubclassK2e}
\tilde Z^0=\frac{X^1}{T_t^2U^1}Z^0+\left(\frac{X^1}{T_tU^1}\right)_t\frac C{T_t}.
\end{gather}
We denote by~$\bar{\mathcal K}_2$ the class~$\mathcal K_2$ 
in which the tuple of arbitrary elements~$\theta$ is formally extended to
$\bar\theta=(A^0,\dots,A^\EqOrd,B,C,Z^0)$ with~$Z^0:=C_t$.

\end{subequations}

\begin{corollary}\label{cor:EffectGenEquivGroupOfGenBurgersKdVEqsSubclassK2}
The class~$\bar{\mathcal K}_2$ is normalized in the generalized sense. 
The group~\smash{$\breve G^\sim_{\bar{\mathcal K}_2}$} 
constituted by the transformations 
of the form~\eqref{eq:AdmTransGenBurgersKdVEqsSubclassK2}
is an effective generalized equivalence group of this class.
\end{corollary}

\begin{proof}
The set of the transformations 
of the form~\eqref{eq:AdmTransGenBurgersKdVEqsSubclassK2}, 
which is temporarily denoted by~$M$, 
is closed with respect to the transformation composition 
and contains the identity transformation. 
Each transformation from~$M$ is invertible by definition. 
So, $M$ is a group.  
The components of transformations from~$M$ are of the same form 
as the components of admissible transformations 
and the formulas relating the initial and target arbitrary elements. 
This is why the group~$M$ generates the equivalence groupoid~$\bar{\mathcal K}_2$ 
and, moreover, it is minimal among subgroups with such property. 
Therefore, $M$ is an effective generalized equivalence group of the class~$\bar{\mathcal K}_2$.
\end{proof}

The entire generalized equivalence group~\smash{$\bar G^\sim_{\bar{\mathcal K}_2}$} 
of the class~$\bar{\mathcal K}_2$
is much wider than its effective part~\smash{$\breve G^\sim_{\bar{\mathcal K}_2}$}. 

\begin{corollary}\label{cor:GenEquivGroupOfGenBurgersKdVEqsSubclassK2}
The generalized equivalence group~\smash{$\bar G^\sim_{\bar{\mathcal K}_2}$} 
of the class~$\bar{\mathcal K}_2$
consists of the transformations 
of the modified form~\eqref{eq:AdmTransGenBurgersKdVEqsSubclassK2}, 
where $T=T(t)$, $X^1=X^1(t)$, $X^0=X^0(t,C)$, $U^1=U^1(t,C)$ and $U^{00}=U^{00}(t,C)$ 
are arbitrary smooth functions of their arguments with $T_tX^1(CU^1_C-U^1)\ne0$,
and the partial derivatives of~$X^0$, $U^1$ and~$U^{00}$ in~$t$ should be replaced by 
the corresponding restricted total derivatives in~$t$ with $\bar{\mathrm D}_t=\p_t+Z^0\p_C$.
\end{corollary}

\begin{proof}
Theorem~\ref{thm:EquivGroupoidOfGenBurgersKdVEqsSubclassK2} implies that 
elements of~\smash{$\bar G^\sim_{\bar{\mathcal K}_2}$} 
are of the modified form~\eqref{eq:AdmTransGenBurgersKdVEqsSubclassK2},
where the group parameters~$T$, $X^1$, $X^0$, $U^1$ and~$U^{00}$ 
may depend on~$t$ and the arbitrary elements~$\bar\theta$. 
Hence partial derivatives of these parameter functions 
should be replaced by the corresponding total derivatives in~$t$ with 
\[
\mathrm D_t=\partial_t+\sum_\alpha u_{\alpha+\delta_1}\partial_{u_\alpha}+\sum_{k=0}^\EqOrd A^k_t\p_{A^k}+B_t\p_B+C_t\p_C+Z^0_t\p_{Z^0}+\cdots\,.
\] 
After substituting $Z^0$ for the derivative~$C_t$, the transformation components can be split 
with respect to the other derivatives of arbitrary elements in~$t$. 
The splitting implies that in fact the group parameters do not depend on 
$A$'s, $B$ and~$Z^0$, and, moreover, the parameters~$T$ and~$X^1$ do not depend on~$C$. 
The nondegeneracy condition for elements of~\smash{$\bar G^\sim_{\bar{\mathcal K}_2}$} 
is modified in comparison with that for elements 
of the effective part~\smash{$\breve G^\sim_{\bar{\mathcal K}_2}$} 
in view of the parameter function~$U^1$ becoming dependent on~$C$. 
This condition takes the form $T_tX^1U^1(C/U^1)_C\ne0$ 
and reduces to the condition given in the statement of the theorem. 
\end{proof}

\begin{remark}\label{rem:UniquenessOfEffectiveGenEquivGroups}
Given a class of differential equations with nontrivial effective generalized equivalence group, 
this group is in general not defined in a unique way. 
Indeed, consider the class~$\bar{\mathcal K}_2$. 
The effective generalized equivalence group~\smash{$\breve G^\sim_{\bar{\mathcal K}_2}$} 
defined in Corollary~\ref{cor:EffectGenEquivGroupOfGenBurgersKdVEqsSubclassK2} 
is not a normal subgroup 
of the entire generalized equivalence group~\smash{$\bar G^\sim_{\bar{\mathcal K}_2}$} 
of the class~$\bar{\mathcal K}_2$. 
Each subgroup of~\smash{$\bar G^\sim_{\bar{\mathcal K}_2}$} 
that is conjugate to~\smash{$\bar G^\sim_{\bar{\mathcal K}_2}$} 
is an effective generalized equivalence group of the class~$\bar{\mathcal K}_2$.
In other words, the class~$\bar{\mathcal K}_2$ possesses 
a wide family of conjugate effective generalized equivalence groups.
The similar fact is even more obvious for the class~$\bar{\mathcal K}_3$ studied below. 
\end{remark}

To have the required subclass~$\mathcal K_3$ of equations from the class~\eqref{eq:GenBurgersKdVEqs}
whose coefficients depend at most on~$t$, 
we now only need to impose a more restrictive constraint on~$B$,
replacing the additional auxiliary equation $B_{xx}=0$ by $B_x=0$, 
which can be implemented by gauging~$B$ within the class~$\mathcal K_2$
using its equivalence transformations. 
Unfortunately, this deteriorates the normalization property 
since then the function~$X^1$ parameterizing elements 
of the equivalence groupoid~$\mathcal G^\sim_{\mathcal K_3}$ of the class~$\mathcal K_3$ 
depends on the initial arbitrary elements~$C$ and~$A^0$ in a nonlocal way via the equation
\begin{equation}\label{eq:AdmTransGenBurgersKdVEqsSubclassK3Constraint}
\left(\frac{X^1_t}{C(X^1)^2}\right)_t=A^0\frac{X^1_t}{C(X^1)^2}.
\end{equation}
At the same time, the usual equivalence group~$G^\sim_{\mathcal K_3}$ of the subclass~$\mathcal K_3$ 
coincides with the group~$G^\sim_{\mathcal K_2}$. 
The computation of the generalized equivalence group of the subclass~$\mathcal K_3$ gives 
the same group, which is a trivial situation from the point of view of generalized equivalence.
As a result, the class~$\mathcal K_3$ is definitely not normalized in both the usual and the generalized senses. 
This is why we construct the extended generalized equivalence group of the subclass~$\mathcal K_3$ 
in a rigorous way. 
In fact, this is the first construction of such kind in the literature. 

We extend the arbitrary-element tuple~$\theta$ to $\bar\theta=(A^0,\dots,A^\EqOrd,B,C,Y^1,Y^2)$
with two more arbitrary elements, $Y^1$ and~$Y^2$, 
which are functions of~$t$ only and satisfy the auxiliary equations 
\begin{equation}\label{eq:GenBurgersKdVEqsSubclassK3AuxiliaryEqsForY}
Y^1_t=A^0,\quad Y^2_t=Ce^{Y^1}.
\end{equation}
Thus, we also implicitly impose the auxiliary equations $Y^i_{u_\alpha}=Y^i_x=0$, $|\alpha|\leqslant \EqOrd$, $i=1,2$.
Each value of~$\bar\theta$ satisfying all auxiliary equations of the class~$\mathcal K_3$
as well as the above equations for~$Y^1$ and~$Y^2$
is associated with an equation of the form~\eqref{eq:GenBurgersKdVEqs} 
with the corresponding value of~$\theta$. 
We formally denote this equation by $\bar{\mathcal L}_{\bar\theta}$ and the class of such equations by~$\bar{\mathcal K}_3$. 
It is obvious that the equations $\bar{\mathcal L}_{\bar\theta^1}$ and $\bar{\mathcal L}_{\bar\theta^2}$ coincide 
if $\theta^1=\theta^2$. 
This defines a gauge equivalence relation on the value set of arbitrary-element tuple~$\bar\theta$.
We show below that this gauge equivalence gives rise to a nontrivial gauge equivalence group of the class~$\bar{\mathcal K}_3$. 
(See Sections~2.1 and~2.5 of~\cite{popo10Ay} for notions related to gauge equivalence, 
which is called trivial equivalence in~\cite{lisl92Ay}.)
Since the set of point transformations from $\bar{\mathcal L}_{\bar\theta^1}$ to $\bar{\mathcal L}_{\bar\theta^2}$ 
coincides with that from $\mathcal L_{\theta^1}$ to $\mathcal L_{\theta^2}$, 
the equivalence groupoid of~$\mathcal K_3$ is isomorphic to 
the equivalence groupoid of~$\bar{\mathcal K}_3$ factorized with respect to the gauge equivalence. 
In the class~$\bar{\mathcal K}_3$, the constraint~\eqref{eq:AdmTransGenBurgersKdVEqsSubclassK3Constraint} 
can be solved with respect to~$X^1$ in terms of~$Y^2$, 
\looseness=-1
\begin{equation}\label{eq:GenBurgersKdVEqsSubclassK3ExpressionForX1}
X^1=\frac1{\varepsilon_1Y^2+\varepsilon_0},
\end{equation}
where $\varepsilon_1$ and~$\varepsilon_0$ are arbitrary constants with $(\varepsilon_1,\varepsilon_0)\ne(0,0)$.
Using this solution and the auxiliary equations~\eqref{eq:GenBurgersKdVEqsSubclassK3AuxiliaryEqsForY}, 
we prolong the relation~\eqref{eq:AdmTransGenBurgersKdVEqsSubclassK2b}--\eqref{eq:AdmTransGenBurgersKdVEqsSubclassK2d} 
between initial and transformed arbitrary elements to~$Y^1$ and~$Y^2$.
Thus, the equality chain
\[
\tilde Y^1_t=\tilde Y^1_{\tilde t}T_t=\tilde A^0T_t=A^0+\frac{U^1_t}{U^1}+\frac{X^1_t}{X^1}=Y^1_t+\frac{U^1_t}{U^1}+\frac{X^1_t}{X^1}
\]  
implies $\tilde Y^1=Y^1+\ln|U^1X^1|+\delta'$ for some constant~$\delta'$.
Considering the equality chain
\[
\tilde Y^2_t=\tilde Y^2_{\tilde t}T_t=\tilde Ce^{\tilde Y^1}T_t=\frac{X^1}{T_tU^1}Ce^{Y^1}U^1X^1\delta T_t
=\frac{\delta Y^2_t}{(\varepsilon_1Y^2+\varepsilon_0)^2},
\]  
where $\delta=e^{\delta'}\sgn(U^1X^1)\ne0$, we derive for some constants $\varepsilon_1'$ and~$\varepsilon_0'$ 
with $\varepsilon_0\varepsilon_1'-\varepsilon_0'\varepsilon_1=\delta$ that 
\begin{equation}\label{eq:GenBurgersKdVEqsSubclassK3Y2Component}
\tilde Y^2=\frac{\varepsilon_1'Y^2+\varepsilon_0'}{\varepsilon_1Y^2+\varepsilon_0},  
\quad\mbox{and hence}\quad 
\tilde Y^1=Y^1+\ln(\delta U^1X^1).
\end{equation}
We use parentheses instead of vertical bars in the logarithm since $\delta U^1X^1>0$. 
This completes the description of the equivalence groupoid~\smash{$\mathcal G^\sim_{\bar{\mathcal K}_3}$}.
Note that here 
\[
U^{01}=\frac{X^1_tU^1}{X^1C}=-\varepsilon_1U^1X^1e^{Y^1}, \quad 
U^{01}_t=U^{01}\left(A^0+\frac{U^1_t}{U^1}-\varepsilon_1CX^1e^{Y^1}\right)=T_tU^{01}\tilde A^0.
\]

\begin{theorem}\label{thm:GenBurgersKdVEqsGenEquivGroupOfSubclassExtK3}
Let $\mathcal K_3$ be the subclass of equations from the class~\eqref{eq:GenBurgersKdVEqs} 
with coefficients depending at most on~$t$, 
which is singled out from the class~\eqref{eq:GenBurgersKdVEqs} 
by the constraints $A^k_x=C_x=B_x=0$, $k=0,\dots,\EqOrd$. 
The class~$\bar{\mathcal K}_3$ of the same equations, 
where the arbitrary-element tuple is formally extended 
with the virtual arbitrary elements~$Y^1$ and~$Y^2$ 
defined by~\eqref{eq:GenBurgersKdVEqsSubclassK3AuxiliaryEqsForY}, 
is normalized in the generalized sense. 
Its generalized equivalence group~\smash{$\bar G^\sim_{\bar{\mathcal K}_3}$} 
consists of the transformations of the~form
\looseness=-1
\begin{gather*}
\tilde t=\bar T(t,Y^1,Y^2),\quad 
\tilde x=\bar X^1x+\bar X^0(t,Y^1,Y^2),\quad 
\bar X^1:=\frac1{\varepsilon_1Y^2+\varepsilon_0},
\\[.5ex]
\tilde u=\bar U^1(t,Y^1,Y^2)\big(u-\varepsilon_1\bar X^1e^{Y^1}x\big)+\bar U^{00}(t,Y^1,Y^2),
\\[.5ex] 
\tilde A^j=\frac{(\bar X^1)^j}{\bar{\mathrm D}_t \bar T}A^j,\quad
\tilde A^1=\frac{\bar X^1}{\bar{\mathrm D}_t \bar T}\left(A^1+\frac{\bar U^{00}}{\bar U^1}C-\frac{\bar{\mathrm D}_t \bar X^0}{\bar X^1}\right),
\\
\tilde A^0=\frac{1}{\bar{\mathrm D}_t \bar T}\left(A^0+\frac{\bar{\mathrm D}_t \bar U^1}{\bar U^1}-\varepsilon_1C\bar X^1e^{Y^1}\right),
\\
\tilde B=\frac{\bar U^1}{\bar{\mathrm D}_t \bar T}B+\frac{\bar{\mathrm D}_t \bar U^{00}}{\bar{\mathrm D}_t \bar T}
+\varepsilon_1\bar U^1\bar X^1e^{Y^1}\frac{A^1}{\bar{\mathrm D}_t \bar T}-\bar U^{00}\tilde A^0,\quad
\tilde C=\frac{\bar X^1}{\bar U^1\bar{\mathrm D}_t \bar T}C,
\\
\tilde Y^1=Y^1+\ln(\delta \bar U^1\bar X^1),\quad
\tilde Y^2=\frac{\varepsilon_1'Y^2+\varepsilon_0'}{\varepsilon_1Y^2+\varepsilon_0},  
\end{gather*}
where $j=2,\dots,\EqOrd$;
$\bar T$, $\bar X^0$, $\bar U^1$ and $\bar U^{00}$ are arbitrary smooth functions of~$t$, $Y^1$ and~$Y^2$ 
with $\bar T_t\bar U^1\ne0$;
$\varepsilon_0$, $\varepsilon_1$, $\varepsilon_0'$ and $\varepsilon_1'$ 
are arbitrary constants with $\delta:=\varepsilon_0\varepsilon_1'-\varepsilon_0'\varepsilon_1\ne0$ 
and, moreover, $\delta \bar U^1\bar X^1>0$;
$\bar{\mathrm D}_t=\p_t+A^0\p_{Y^1}+Ce^{Y^1}\p_{Y^2}$ is the restricted total derivative operator with respect to~$t$.
\end{theorem}

\begin{proof}
In view of the above description 
of the equivalence groupoid~\smash{$\mathcal G^\sim_{\bar{\mathcal K}_3}$} 
of the class~$\bar{\mathcal K}_3$, 
elements of~\smash{$\bar G^\sim_{\bar{\mathcal K}_3}$} have the general form 
\begin{gather*}
\tilde t=\bar T(t,\bar\theta),\quad
\tilde x=\bar X^1(t,\bar\theta)x+\bar X^0(t,\bar\theta),\quad
\tilde u=\bar U^1(t,\bar\theta)u+\bar U^{01}(t,\bar\theta)x+\bar U^{00}(t,\bar\theta),\\
\tilde{\bar\theta}=\bar\Theta(t,x,u,\bar\theta).
\end{gather*}
The computation of~\smash{$\bar G^\sim_{\bar{\mathcal K}_3}$} by the direct method 
is quite similar to the computation of~\smash{$\mathcal G^\sim_{\bar{\mathcal K}_3}$}
and, after splitting with respect to~$x$ and parametric derivatives of~$u$, 
gives similar expressions for transformation components for the variables $(t,x,u)$ and 
similar constraints for parameter functions. 
The relations between the initial and target arbitrary elements in the equivalence groupoid just convert 
to the transformation components for arbitrary elements in the equivalence group.
But there are several differences, which we are going to discuss. 

In particular, the total derivative operators should be prolonged to the arbitrary elements. 
Since the arbitrary elements of the class~$\bar{\mathcal K}_3$ depend at most on~$t$,
the prolongation is essential only for~$\mathrm D_t$, 
\[
\mathrm D_t=\partial_t+\sum_\alpha u_{\alpha+\delta_1}\partial_{u_\alpha}+\sum_{k=0}^\EqOrd A^k_t\p_{A^k}+B_t\p_B+C_t\p_C+Y^1_t\p_{Y^1}+Y^2_t\p_{Y^2}+\cdots\,.
\] 
The expression for~$\mathrm D_x$ is formally preserved, $\mathrm D_x=\partial_x+\sum_\alpha u_{\alpha+\delta_2}\partial_{u_\alpha}$. 
As a result, all partial derivatives with respect to~$t$ in the expressions derived 
after splitting with respect to~$x$ and parametric derivatives of~$u$ 
are converted to the total derivatives with respect to~$t$. 

The second difference is the possibility of splitting 
with respect to arbitrary elements and their derivatives.  
After substituting for the constrained derivatives~$Y^1_t$ and~$Y^2_t$ 
in view of~\eqref{eq:GenBurgersKdVEqsSubclassK3AuxiliaryEqsForY} 
into the constraint for~$\bar X^1$, 
\[
\mathrm D_t^{\,2}\frac1{\bar X^1}=\left(\frac{C_t}C+A^0\right){\mathrm D}_t\frac1{\bar X^1},
\]
we can split the resulting equation with respect to 
$A^0_{tt}$, \dots, $A^\EqOrd_{tt}$, $B_{tt}$, $C_{tt}$, $A^0_t$ and~$C_t$.  
This leads to the system 
$\bar X^1_{A^0}=\dots=\bar X^1_{A^\EqOrd}=0$, $\bar X^1_B=0$, $\bar X^1_C=0$, $\bar X^1_{Y^1}=0$, $\bar X^1_t=0$ and~$(1/\bar X^1)_{Y^2Y^2}=0$, 
whose general solution is of the form~\eqref{eq:GenBurgersKdVEqsSubclassK3ExpressionForX1}.
The expressions for the transformed arbitrary elements $\tilde A^0$,~\dots, $\tilde A^\EqOrd$, $\tilde B$ and~$\tilde C$ 
can also be split with respect to unconstrained derivatives of arbitrary elements in~$t$, 
implying that the derivatives of $\bar T$, $\bar X^0$, $\bar U^1$ and $\bar U^{00}$ 
with respect to $A^0$, \dots, $A^\EqOrd$, $B$ and~$C$ are zero. 
Hence the operator~$\mathrm D_t$ can be replaced 
by the restricted total derivative operator~$\bar{\mathrm D}_t$.
In particular, the parameter function~$\bar U^{01}$ is defined by 
$\bar U^{01}=(\bar U^1\bar{\mathrm D}_t\bar X^1)/(\bar X^1C)$. 

The additional auxiliary equations~\eqref{eq:GenBurgersKdVEqsSubclassK3AuxiliaryEqsForY} 
are also treated in a different way. 
We substitute the expressions for~$Y^1_t$ and~$Y^2_t$ given by these equations 
into their expanded version for transformed arbitrary elements. 
Splitting the resulting equations with respect to the other derivatives of arbitrary elements
leads to the system of determining equations for the $(Y^1,Y^2)$-components 
of equivalence transformations
\begin{gather*}
\tilde Y^2_t=\tilde Y^2_{Y^1}=\tilde Y^i_{A^k}=\tilde Y^i_B=\tilde Y^i_C=0,\quad i=1,2,\quad k=0,\dots,\EqOrd,\\
\tilde Y^1_t=\frac{U^1_t}{U^1},\quad 
\tilde Y^1_{Y^1}=\frac{U^1_{Y^1}}{U^1}+1,\quad 
\tilde Y^1_{Y^2}=\frac{U^1_{Y^2}}{U^1}-\frac{\varepsilon_1}{\varepsilon_1Y^2+\varepsilon_0},\quad 
\tilde Y^2_{Y^2}=\frac{\bar X^1}{U^1}e^{\tilde Y^1-Y^1}, 
\end{gather*}
whose general solution is of the form presented in the statement of the theorem.
\end{proof}

\begin{remark}
Each element of the generalized equivalence group~\smash{$\bar G^\sim_{\bar{\mathcal K}_3}$} 
generates a family of admissible transformations of the class~$\bar{\mathcal K}_3$ with sources 
at those values of~$\bar\theta$ where the evaluation of~$\bar{\mathrm D}_t \bar T$ does not vanish,
\[
\smash{\bar G^\sim_{\bar{\mathcal K}_3}}\ni\mathcal T \mapsto 
\big\{
(\bar\theta^1,\bar\theta^2,\varphi)\mid
\bar\theta^1\in\bar{\mathcal S}_3,\, 
(\bar{\mathrm D}_t \bar T)|_{\bar\theta=\bar\theta^1}\ne0,\,
\bar\theta^2=\mathcal T\bar\theta^1,\, 
\varphi=(\mathcal T|_{\bar\theta=\bar\theta^1})\big|_{(t,x,u)}
\big\}
\subset\mathcal G^\sim_{\bar{\mathcal K}_3}.
\]
Here $\bar{\mathcal S}_3$ is the value set of the arbitrary-element tuple~$\bar\theta$ of the class~$\bar{\mathcal K}_3$.
\end{remark}

The gauge equivalence group of the class~$\bar{\mathcal K}_3$ 
is the subgroup of~\smash{$\bar G^\sim_{\bar{\mathcal K}_3}$} 
that is singled out by the constraints 
$\varepsilon_0=1$, $\varepsilon_1=0$, $\bar T=t$, $\bar X^0=0$, $\bar U^1=1$, $\bar U^{00}=0$. 
In other words, all the components of gauge equivalence transformations 
are identities, except the components for~$Y^1$ and~$Y^2$, for which we get
$\tilde Y^1=Y^1\!+\ln\varepsilon_1'$, 
$\tilde Y^2=\varepsilon_1'Y^2+\varepsilon_0'$ with $\varepsilon_1'>0$.
The usual equivalence group of the class~$\bar{\mathcal K}_3$ 
is singled out from~\smash{$\bar G^\sim_{\bar{\mathcal K}_3}$} 
by the constraints 
\[
\varepsilon_1=0,\quad
\bar T_{Y^i}=\bar X^0_{Y^i}=\bar U^1_{Y^i}=\bar U^{00}_{Y^i}=0, \quad i=1,2,
\]
and its quotient group with respect the gauge equivalence group of the class~$\bar{\mathcal K}_3$ 
is isomorphic to the usual equivalence group of the class~$\mathcal K_3$. 

It is obvious that the generalized equivalence group~\smash{$\bar G^\sim_{\bar{\mathcal K}_3}$} 
of the class~$\bar{\mathcal K}_3$ generates the whole equivalence groupoid of this class. 
At the same time, functions parameterizing the group depend on two more arguments, $Y^1$ and~$Y^2$, 
than functions parameterizing the groupoid do.
If we omit the arguments~$Y^1$ and~$Y^2$ in the parameter functions, 
the corresponding set of transformations still generates the equivalence groupoid 
but it is not a group with respect to the transformation composition. 
This shows that the class~$\bar{\mathcal K}_3$ may possess an effective generalized equivalence group 
being a proper subgroup of~\smash{$\bar G^\sim_{\bar{\mathcal K}_3}$},  
and its construction needs a more delicate consideration than, e.g., for the class~$\mathcal K_2$.

\begin{corollary}\label{cor:GenBurgersKdVEqsExtGenEquivGroupOfSubclassK3}
The class~$\mathcal K_3$ is normalized in the extended generalized~sense.
Its extended generalized equivalence group~$\hat G^\sim_{\mathcal K_3}$  
can be identified with the effective generalized equivalence group of the class~$\bar{\mathcal K}_3$ 
that consists of the transformations of the form
\begin{gather*}
\tilde t=T(t),\quad 
\tilde x=X^1\big(x+X^{01}(t)Y^2+X^{00}(t)\big),\quad 
X^1:=\frac1{\varepsilon_1Y^2+\varepsilon_0},
\\
\tilde u=V(t)\biggl(\frac u{X^1}-e^{Y^1}(\varepsilon_1x-\varepsilon_0X^{01}+\varepsilon_1X^{00})\biggr),
\\
\tilde A^j=\frac{(X^1)^j}{T_t}A^j,\quad
\tilde A^1=\frac{X^1}{T_t}\left(A^1-X^{01}_tY^2-X^{00}_t\right),\quad
\tilde A^0=\frac{1}{T_t}\left(A^0+\frac{V_t}{V}\right),
\\
\tilde B=\frac{V}{T_t}\left(\frac B{X^1}-e^{Y^1}(\varepsilon_1A^1-\varepsilon_0X^{01}_t+\varepsilon_1X^{00}_t)\right),\quad
\tilde C=\frac{(X^1)^2}{T_tV}C,
\\
\tilde Y^1=Y^1+\ln(\delta V),\quad
\tilde Y^2=\frac{\varepsilon_1'Y^2+\varepsilon_0'}{\varepsilon_1Y^2+\varepsilon_0},  
\end{gather*}
where $j=2,\dots,\EqOrd$;
and $T$, $X^{00}$, $X^{01}$ and~$V$ are arbitrary smooth functions of~$t$ with $T_tV\ne0$;
$\varepsilon_0$, $\varepsilon_1$, $\varepsilon_0'$ and $\varepsilon_1'$ 
are arbitrary constants with $\delta:=\varepsilon_0\varepsilon_1'-\varepsilon_0'\varepsilon_1\ne0$ 
and, moreover, $\delta V>0$.
\end{corollary}

\begin{proof}
We temporarily denote by~$M$ the set of the transformations of the above form. 
This set is a subset of the group~\smash{$G^\sim_{\bar{\mathcal K}_3}$}. 
It is singled out from~\smash{$G^\sim_{\bar{\mathcal K}_3}$} 
by setting the following values for group parameters: 
\begin{gather*}
\bar T=T(t),\quad
\bar X^0=X^1\big(X^{01}(t)Y^2+X^{00}(t)\big),\quad 
\bar U^1=\frac{V(t)}{X^1},\\ 
\bar U^{01}=-\varepsilon_1V(t)e^{Y^1},\quad
\bar U^{00}=V(t)e^{Y^1}\big(\varepsilon_0X^{01}(t)-\varepsilon_1X^{00}(t)\big).
\end{gather*}

The set~$M$ is closed with respect to the transformation composition, 
i.e., $M$ is a subgroup of the group~\smash{$G^\sim_{\bar{\mathcal K}_3}$}. 

The subgroup~$M$ generates the entire equivalence groupoid~\smash{$\mathcal G^\sim_{\bar{\mathcal K}_3}$} 
of the class~$\bar{\mathcal K}_3$ 
and thus the entire equivalence groupoid of the class~$\mathcal K_3$. 
Indeed, let us fix any equation~$\bar{\mathcal L}_{\bar\theta}$ from the class~$\bar{\mathcal K}_3$. 
The set~${\rm T}_{\bar\theta}$ of all admissible transformations with source at~$\bar\theta$ 
is parameterized by the arbitrary smooth functions~$T$, $X^0$, $U^1$ and $U^{00}$ of~$t$
and the arbitrary constants~$\varepsilon_0$, $\varepsilon_1$, $\varepsilon_0'$ and $\varepsilon_1'$ 
with $T_tU^1\ne0$, $\delta:=\varepsilon_0\varepsilon_1'-\varepsilon_0'\varepsilon_1\ne0$ 
and $\delta U^1X^1>0$, 
where $X^1$ is defined by~\eqref{eq:GenBurgersKdVEqsSubclassK3ExpressionForX1} 
for the fixed value of the arbitrary element~$Y^2$, $Y^2=Y^2(t)$.
Each admissible transformation from~${\rm T}_{\bar\theta}$ 
is generated by the equivalence transformation from~$M$ with the same values 
of~$T$, $\varepsilon_0$, $\varepsilon_1$, $\varepsilon_0'$ and~$\varepsilon_1'$, 
and the values of~$X^{00}$, $X^{01}$ and~$V$ defined by
\[
X^{00}=\varepsilon_0X^0(t)-Y^2(t)\frac{U^{00}(t)}{U^1(t)}e^{Y^1(t)},\quad\!
X^{01}=\varepsilon_1X^0(t)+\frac{U^{00}(t)}{U^1(t)}e^{Y^1(t)},\quad\!
V=\frac{U^1(t)}{\varepsilon_1Y^2(t)+\varepsilon_0}.
\]
This establishes a one-to-one correspondence between~$M$ and~${\rm T}_{\bar\theta}$,
and thus the subgroup~$M$ is minimal among the subgroups of~\smash{$G^\sim_{\bar{\mathcal K}_3}$} 
that generate the groupoid~\smash{$\mathcal G^\sim_{\bar{\mathcal K}_3}$}. 

Therefore, $M$ is the effective generalized equivalence group of the class~$\bar{\mathcal K}_3$.
\end{proof}

Nevertheless, we can further gauge arbitrary elements of the class~$\mathcal K_3$
in such a way that the corresponding subclasses of~$\mathcal K_3$ have better normalization properties.
Up to $G^\sim_{\mathcal K_3}$-equivalence we can set, as above, $C=1$ and $A^1=0$.
Moreover, since the class~$\mathcal K_3$ is parameterized by functions depending at most on~$t$,
$G^\sim_{\mathcal K_3}$-equivalence allows us to further gauge arbitrary elements by setting $A^0=0$ and~$B=0$,
which results in the class~$\mathcal K$ normalized in the usual sense 
and whose usual equivalence group~$G^\sim_{\mathcal K}$ is finite-dimensional,
see Section~\ref{sec:AlternativeClassificationCases}.
At the same time, gauging arbitrary elements by usual equivalence transformations in a class 
that is not normalized in the usual sense does not allow one to easily control 
the changing of the corresponding equivalence groupoid. 
This is why we revise the order of gauging 
in order to construct a subclass normalized in the usual sense for each step of gauging.

\begin{remark}
In general, if a multiple gauge is realized step-by-step,
then the order of steps is also important and
intermediate subclasses may have inferior normalization properties 
than the original class and the final gauged subclass.
\end{remark}

We use the following order of gauging for singling out 
the subclass~$\mathcal K$ from the class~$\mathcal K_1$: 
$C=1$, $A^1=0$, $A^0=0$, $B^1=0$ and~$B^0=0$,
which successively constrains the group(oid) parameters~by
\[
U^1=\frac{X^1}{T_t},\quad
U^0=\frac{X^1_tx+X^0_t}{T_t},\quad
\frac{T_{tt}}{T_t}=2\frac{X^1_t}{X^1},\quad
\left(\frac{X^1_t}{(X^1)^2}\right)_t=0,\quad
\left(\frac{X^0_t}{T_t}\right)_t=0.
\]
Note that, e.g., the ordering of the last two gauges is not essential
but if we try to carry out at least one of them before the gauge $A^0=0$,
then we obtain a subclass normalized in the generalized extended sense;
cf.\ the derivation of Corollary~\ref{cor:GenBurgersKdVEqsExtGenEquivGroupOfSubclassK3}.
Thus, the parameter function~$X^1$ of the equivalence groupoid 
for the subclass gauged with $C=1$, $A^1=0$ and $B^1=0$
is related to the initial value of the arbitrary element~$A^0$ 
via the equation~\eqref{eq:AdmTransGenBurgersKdVEqsSubclassK3Constraint} with~$C=1$.

Finally, the solution of the group classification problem for the class~$\mathcal K_3$
up to the generalized extended $\hat G^\sim_{\mathcal K_3}$-equivalence reduces to
that for the class~$\mathcal K$ up to the usual $G^\sim_{\mathcal K}$-equivalence.
Combining results of Sections~\ref{sec:GroupClassificationGenBurgersKdVEqs}
and~\ref{sec:AlternativeClassificationCases},
we readily obtain the following assertion.

\begin{theorem}\label{thm:GroupClassificationGenBurgersKdVEqsWithTime-DependentCoeffs}
A complete list
of $\hat G^\sim_{\mathcal K_3}$-inequivalent Lie symmetry extensions in the class~$\mathcal K_3$
(resp.\ $G^\sim_{\mathcal K}$-inequivalent Lie symmetry extensions in the class~$\mathcal K$)
is exhausted by Cases~\ref{200},
\smash{$\tilde{\mbox{\ref{201}a}}$}, \smash{$\tilde{\mbox{\ref{201}b}}$}, \smash{$\tilde{\mbox{\ref{201}c}}$} and~\ref{202}
given in Table~\ref{tab:CompleteGroupClassificationBurgersKdVEquations} and
Section~\ref{sec:AlternativeClassificationCases}.
\end{theorem}

Whereas each equation from the class~$\mathcal K_3$ admits the two-dimensional algebra $\langle P(\chi^1),P(\chi^2)\rangle$
with $\chi^1$ and $\chi^2$ constituting a fundamental set of solutions of the linear ordinary differential equation $\chi^i_{tt}=A^0\chi^i_t+B^1\chi^i$,
the kernel Lie invariance algebra of this class is zero
since the functions~$\chi^1$ and $\chi^2$ depend on the arbitrary elements~$A^0$ and~$B^1$.
The situation with the class~$\mathcal K$ is different:
its kernel Lie invariance algebra is $\langle P(1),P(t)\rangle$.

\section{Equations with space-dependent coefficients}\label{sec:GenBurgersKdVEqsWithSpace-DependentCoeffs}

The subclass~$\mathcal F_1$ of equations with space-dependent coefficients 
is singled out from the class~\eqref{eq:GenBurgersKdVEqs}, 
which is normalized in the usual sense, 
by the constraints $A^k_t=0$, $k=0,\dots,\EqOrd$, $B_t=0$ and $C_t=0$.  
This is why the usual equivalence group~$G^\sim_{\mathcal F_1}$ of the subclass~$\mathcal F_1$ 
is a subgroup of the equivalence group~$G^\sim_{\mbox{\tiny\eqref{eq:GenBurgersKdVEqs}}}$ 
of the class~\eqref{eq:GenBurgersKdVEqs} 
that consists of transformations preserving the above constraints. 
In~view of Proposition~\ref{pro:EquivGroupoidOfGenBurgersKdVEqs}, 
the $(t,x,u,A^\EqOrd,C,A^1)$-components of transformations from~$G^\sim_{\mathcal F_1}$ are of the form
\begin{gather*}
\tilde t=T(t),\quad 
\tilde x=X(t,x),\quad 
\tilde u=U^1(t)u+U^0(t,x),\quad
A^\EqOrd=\frac{T_t}{(X_x)^\EqOrd}\tilde A^\EqOrd, \quad
C=\frac{T_tU^1}{X_x}\tilde C,\\
A^1=\frac{T_t}{X_x}\tilde A^1+T_t\sum_{j=2}^\EqOrd\tilde A^j\left(\frac1{X_x}{\rm D}_x\right)^{j-1}\!\!\!\frac1{X_x}
    -\frac{T_tU^0}{X_x}\tilde C+\frac{X_t}{X_x},
\end{gather*}
where $T=T(t)$, $X=X(t,x)$, $U^1=U^1(t)$ and $U^0=U^0(t,x)$
are smooth functions of their arguments such that $T_tX_xU^1\ne0$.
(It is convenient here to conversely express source arbitrary elements via target ones.)
Differentiating the $A^\EqOrd$- and $C$-components with respect to~$t$ 
and taking into account the additional constraints for arbitrary elements of the subclass~$\mathcal F_1$, 
we derive classifying equations for admissible transformations in~$\mathcal F_1$, 
\[
\frac{T_t}{(X_x)^\EqOrd}X_t\tilde A^\EqOrd_{\tilde x}+\left(\frac{T_t}{(X_x)^\EqOrd}\right)_t\tilde A^\EqOrd=0,\quad
\frac{T_tU^1}{X_x}X_t\tilde C_{\tilde x}+\left(\frac{T_tU^1}{X_x}\right)_t\tilde C=0,
\]
which should be split with respect to derivatives of arbitrary elements 
when constructing usual equivalence transformations of the subclass~$\mathcal F_1$.
Hence $X_t=T_{tt}=U^1_t=0$ for these transformations. 
Then treating the $A^1$-component in the same way 
gives one more condition for group parameters, $U^0_t=0$. 
In fact, the obtained conditions for group parameters constitute 
the complete system of such conditions that singles out 
the group~$G^\sim_{\mathcal F_1}$ as a subgroup of the group~$G^\sim_{\mbox{\tiny\eqref{eq:GenBurgersKdVEqs}}}$.

\begin{proposition}\label{pro:EquivGroupOfGenBurgersKdVEqsWithX-dependentCoeffs}
The usual equivalence group~$G^\sim_{\mathcal F_1}$ of 
the class~$\mathcal F_1$ of general Burgers--Korteweg--de Vries equations 
with space-dependent coefficients 
consists of the transformations in the joint space of $(t,x,u,\theta)$
whose $(t,x,u)$-components are of the form
\[
\tilde t=c_1t+c_2,\quad \tilde x=X(x),\quad \tilde u=c_3'u+U^0(x), 
\]
where $c_1$, $c_2$ and~$c_3'$ are arbitrary constants 
and $X=X(x)$ and~$U^0=U^0(x)$ are arbitrary smooth functions of~$x$ such that $c_1X_xc_3'\ne0$.
\end{proposition}

The existence of classifying conditions for admissible transformations means 
that the class~$\mathcal F_1$ is definitely not normalized in any sense. 
This is why a proper starting point for computing the equivalence group of a subclass~$\mathcal F'$ of~$\mathcal F_1$
is not the equivalence group of~$\mathcal F_1$ 
but the equivalence group of the superclass~\eqref{eq:GenBurgersKdVEqs} 
or of a normalized subclass of~\eqref{eq:GenBurgersKdVEqs} that contains~$\mathcal F'$. 

Two parameters of the group of the class~$\mathcal F_1$ 
are of the same arbitrariness as the arbitrary elements of this class. 
Therefore, we can set two gauges for the arbitrary elements of~$\mathcal F_1$ 
using\- parameterized families of transformations from~$G^\sim_{\mathcal F_1}$. 
Following the consideration of the superclass~\eqref{eq:GenBurgersKdVEqs} 
in Section~\ref{sec:EquivalenceGroupoidGenBurgersKdVEqs}, 
we successively gauge~$C$ to~1 and~$A^1$ to~0. 

In view of the expression for the $C$-component of transformations from~$G^\sim_{\mathcal F_1}$, 
the gauge $C=1$ can really be implemented. 
We denote the corresponding subclass of~$\mathcal F_1$ by~$\mathcal F_2$.
The usual equivalence group~$G^\sim_{\mathcal F_2}$ of the subclass~$\mathcal F_2$ 
turns out to be the subgroup of~$G^\sim_{\mathcal F_1}$ 
singled out by the constraint $T_tU^1/X_x=1$, i.e., $X_{xx}=0$ and $U^1=X_x/T_t$.
For transformations from~$G^\sim_{\mathcal F_2}$, 
the $C$-component can be neglected and
the expression for the $A^1$-component simplifies to $\tilde A^1=X_xA^1/T_t+U^0$.
It becomes clear that the further gauge $A^1=0$ is also realizable, 
which gives a subclass of~$\mathcal F_2$, and we denote it by~$\mathcal F_3$. 
The usual equivalence group~$G^\sim_{\mathcal F_3}$ of this subclass
turns out to be the subgroup of the group~$G^\sim_{\mathcal F_2}$ 
singled out by the constraint $U^0=0$. 
In other words, the group~$G^\sim_{\mathcal F_3}$ is just four-dimensional 
and consists of the transformations of the simple form 
\begin{gather*}
\tilde t=c_1t+c_2,\quad 
\tilde x=c_3x+c_4,\quad 
\tilde u=\frac{c_3}{c_1}u,\quad 
\tilde A^j=\frac{(c_3)^j}{c_1}A^j,\quad
\tilde A^0=\frac1{c_1}A^0,\quad
\tilde B=\frac{c_3}{(c_1)^2}B,
\end{gather*}
where $j=2,\dots,\EqOrd$, 
and $c_1$, \dots, $c_4$ are arbitrary constants with $c_1c_3\ne0$. 
The $A^1$-component can be neglected.

Since the subclasses~$\mathcal F_2$ and~$\mathcal F_3$ 
are constructed via gauging the arbitrary elements of the class~$\mathcal F_1$ 
by usual equivalence transformations, and $G^\sim_{\mathcal F_1}\supset G^\sim_{\mathcal F_2}\supset G^\sim_{\mathcal F_3}$, 
then any classification problem for the class~$\mathcal F_1$ up to~$G^\sim_{\mathcal F_1}$-equivalence 
(like the group classification problem or the description of the corresponding equivalence groupoid) 
reduces to the respective classification problem for the class~$\mathcal F_3$ up to~$G^\sim_{\mathcal F_3}$-equivalence, 
cf.\ Proposition~\ref{pro:GroupClassificationsOfGenBurgersKdVEqsAndGaugedGenBurgersKdVEqs}.  

We begin with the description of the equivalence groupoid~$\mathcal G^\sim_{\mathcal F_3}$ of the class~$\mathcal F_3$. 
Since the class~$\mathcal F_3$ is not normalized and neither are its superclasses~$\mathcal F_1$ and~$\mathcal F_2$, 
this problem should be considered as the classification of admissible transformations 
up to~$G^\sim_{\mathcal F_3}$-equivalence.
See Section~2.6 and~3.4 in~\cite{popo10Ay}.
Since the class~$\mathcal F_3$ is a subclass of the class~\eqref{eq:GenBurgersKdVEqsGaugedSubclass}, 
the equivalence groupoid~$\mathcal G^\sim_{\mathcal F_3}$ 
is a subgroupoid of the equivalence groupoid of the class~\eqref{eq:GenBurgersKdVEqsGaugedSubclass}. 
Therefore, all the restrictions from Theorem~\ref{thm:EquivalenceGroupGenBurgersKdVEqsGaugeC1A10} 
on the form of admissible transformations are relevant here. 
We solve the relations 
\eqref{eq:PointTransformationBetweenGenBurgersKdVEqsGaugeC1A10b}--\eqref{eq:PointTransformationBetweenGenBurgersKdVEqsGaugeC1A10c}
with respect to the source arbitrary elements, 
and differentiate the result with respect to~$t$. 
This gives the classifying conditions for admissible transformations 
in terms of target arbitrary elements,
\begin{subequations}\label{eq:GenBurgersKdVEqsWithX-dependentCoeffsClassifyingEqsForAdmTrans}
\begin{gather}
(X^1_tx+X^0_t)\tilde A^j_{\tilde x}+\left(\frac{T_{tt}}{T_t}-j\frac{X^1_t}{X^1}\right)\tilde A^j=0,
\label{eq:GenBurgersKdVEqsWithX-dependentCoeffsClassifyingEqsForAdmTransA}\\ 
(X^1_tx+X^0_t)\tilde A^0_{\tilde x}+\frac{T_{tt}}{T_t}\tilde A^0=
\frac1{T_t}\left(2\frac{X^1_t}{X^1}-\frac{T_{tt}}{T_t}\right)_t,\ 
\label{eq:GenBurgersKdVEqsWithX-dependentCoeffsClassifyingEqsForAdmTransB}\\ 
\begin{split}\label{eq:GenBurgersKdVEqsWithX-dependentCoeffsClassifyingEqsForAdmTransC} 
&(X^1_tx+X^0_t)\tilde B_{\tilde x}+\left(2\frac{T_{tt}}{T_t}-\frac{X^1_t}{X^1}\right)\tilde B=
-\frac{T_t}{X^1}\left(X^1_tx+X^0_t\right)^2\tilde A^0_{\tilde x}
\\
&\qquad 
-\frac{X^1}{T_t^2}\left(T_t\frac{X^1_tx+X^0_t}{X^1}\right)_t\tilde A^0
+\frac{X^1}{T_t^2}\left(\frac{T_t}{X^1}\left(\frac{X^1_tx+X^0_t}{T_t}\right)_t\right)_t,
\end{split}
\end{gather}
\end{subequations}
where the initial space variable~$x$ should be substituted, after expanding all derivatives, 
by its expression via~$\tilde x$, $x=(\tilde x-X^0)/X^1$. 
Since admissible transformations with $T_{tt}=X^0_t=X^1_t=0$ are generated 
by the usual equivalence group~$G^\sim_{\mathcal F_3}$, 
the problem on classifying of admissible transformations 
can be stated in the following way: 
\emph{To find $G^\sim_{\mathcal F_3}$-inequivalent values of the arbitrary-element tuple 
$\tilde\kappa=(\tilde A^0,\tilde A^2,\dots,\tilde A^\EqOrd,\tilde B)$ 
for which the system~\eqref{eq:GenBurgersKdVEqsWithX-dependentCoeffsClassifyingEqsForAdmTrans} 
with respect to $(T,X^0,X^1)$ has solutions with $(T_{tt},X^0_t,X^1_t)\ne(0,0,0)$.}

For each fixed~$t$, the system~\eqref{eq:GenBurgersKdVEqsWithX-dependentCoeffsClassifyingEqsForAdmTrans} 
implies a system of ordinary differential equations with respect to the arbitrary-element tuple~$\tilde\kappa$ of the form
\begin{gather}\label{eq:GenBurgersKdVEqsWithX-dependentCoeffsConsequenceOfClassifyingEqsForAdmTrans}
\begin{split}
&(\tilde\nu_1\tilde x+\tilde\nu_2)\tilde A^j_{\tilde x}+(\tilde\nu_3-j\tilde\nu_1)\tilde A^j=0,\quad j=2,\dots,n,\\
&(\tilde\nu_1\tilde x+\tilde\nu_2)\tilde A^0_{\tilde x}+\tilde\nu_3\tilde A^0=\tilde\nu_4,\\
&(\tilde\nu_1\tilde x+\tilde\nu_2)\tilde B_{\tilde x}+(2\tilde\nu_3-\tilde\nu_1)\tilde B=
\tilde\nu_9(\tilde\nu_1\tilde x+\tilde\nu_2)^2\tilde A^0_{\tilde x}-(\tilde\nu_7\tilde x+\tilde\nu_8)\tilde A^0+\tilde\nu_5\tilde x+\tilde\nu_6,
\end{split}
\end{gather}
where $\tilde\nu_1$, \dots, $\tilde\nu_9$ are constants. 
For each value of the arbitrary-element tuple~$\tilde\kappa$,
we denote by~$\mathcal N_{\tilde\kappa}$ the span of the set through which $\tilde\nu=(\tilde\nu_1,\dots,\tilde\nu_9)$ runs 
when the tuple $(T,X^0,X^1)$ takes values possible for elements of~$\mathcal G^\sim_{\mathcal F_3}$ with the target~$\tilde\kappa$
and the variable~$t$ is varied. 
It is obvious that $k_{\tilde\kappa}:=\dim\{(\tilde\nu_1,\tilde\nu_2)\mid\tilde\nu\in\mathcal N_{\tilde\kappa}\}\in\{0,1,2\}$. 
The consistency of system~\eqref{eq:GenBurgersKdVEqsWithX-dependentCoeffsClassifyingEqsForAdmTrans} with the condition $\tilde A^\EqOrd\ne0$ 
implies that $\dim\{(\tilde\nu_1,\dots,\tilde\nu_4)\mid\tilde\nu\in\mathcal N_{\tilde\kappa}\}=k_{\tilde\kappa}$.
Therefore, the equality $k_{\tilde\kappa}=0$ means that $T_{tt}=X^0_t=X^1_t=0$
for all admissible transformations in~$\mathcal F_3$ with target at~$\tilde\kappa$.
The further consideration partitions into five principal cases, depending on the value~$k_{\tilde\kappa}$ 
and certain constraints for components of nonzero~$\tilde\nu\in\mathcal N_{\tilde\kappa}$, 
  I)~$\tilde\nu_1\tilde\nu_3\ne0$,
 II)~$\tilde\nu_1\ne0$, $\tilde\nu_3=0$,
III)~\mbox{$\tilde\nu_1=0$}, $\tilde\nu_2\ne0$, $\tilde\nu_3\ne0$,
 IV)~$\tilde\nu_1=0$, $\tilde\nu_2\ne0$, $\tilde\nu_3=0$
with $k_{\tilde\kappa}=1$ and
  V)~\mbox{$k_{\tilde\kappa}=2$}.
For each of these cases, we integrate the corresponding set of independent copies 
of the system~\eqref{eq:GenBurgersKdVEqsWithX-dependentCoeffsConsequenceOfClassifyingEqsForAdmTrans}
and thus derive the respective form of~$\tilde\kappa$ parameterized by arbitrary constants. 
The relations~\mbox{\eqref{eq:PointTransformationBetweenGenBurgersKdVEqsGaugeC1A10b}--\eqref{eq:PointTransformationBetweenGenBurgersKdVEqsGaugeC1A10c}}
between~$\kappa$ and~$\tilde\kappa$ imply 
that the arbitrary-element tuple~$\kappa$ is of the same form as $\tilde\kappa$, 
but maybe with other values of parameter constants.
In other words, subclasses associated with the above classification cases 
are invariant under the action of the equivalence groupoid~$\mathcal G^\sim_{\mathcal F_3}$, 
which additionally justifies the chosen partition into cases. 
We substitute the derived expressions for~$\kappa$ and~$\tilde\kappa$ 
into the relations~\eqref{eq:PointTransformationBetweenGenBurgersKdVEqsGaugeC1A10b}--\eqref{eq:PointTransformationBetweenGenBurgersKdVEqsGaugeC1A10c}
and, assuming $\tilde x=X^1x+X^0$, split the resulting equations with respect to~$x$.
When integrating the obtained system of determining equations for the parameter functions~$T$, $X^0$ and~$X^1$ of
point transformations mapping the equation~$\mathcal L_\kappa$ to the equation~$\mathcal L_{\tilde\kappa}$,
we check for the existence of solutions with $(T_{tt},X^0_t,X^1_t)\ne(0,0,0)$. 
This implies constraints for constants parameterizing the expressions for~$\kappa$, 
which is in fact identical to the similar constraints for~$\tilde\kappa$. 
Finally, we gauge constant parameters of arbitrary elements with transformations from~$G^\sim_{\mathcal F_3}$, 
where this is possible, 
and show that the lists of $G^\sim$-inequivalent cases 
of admissible-transformation extensions and of Lie symmetry extensions 
in the class~$\mathcal F_3$ coincide.
\looseness=1

Note that a complete list of 
$\mathcal G^\sim_{\mathcal F_3}$-inequivalent Lie symmetry extensions in the class~$\mathcal F_3$ 
(resp.\ $\mathcal G^\sim_{\mathcal F_1}$-inequivalent Lie symmetry extensions in the class~$\mathcal F_1$)
is exhausted by Cases~\ref{001}, \ref{002a}, \ref{002b}, \ref{003}, \ref{011a}, \ref{011b}, \ref{201} and \ref{202} 
of Table~\ref{tab:CompleteGroupClassificationBurgersKdVEquations}. 
This follows from the facts that the integer~$k_3$, 
which is defined in Section~\ref{sec:PropertiesOfAppropriateSubalgebrasGenBurgersKdVEqs}, 
is $\mathcal G^\sim_{\mathcal F_3}$-invariant when considered for equations from the class~$\mathcal F_3$
and is greater than or equal to 1 for these equations. 
Case~\ref{001} of Table~\ref{tab:CompleteGroupClassificationBurgersKdVEquations} 
is the general case for the class~$\mathcal F_3$. 
In other words, the kernel Lie invariance algebra of the class~$\mathcal F_3$ is 
$\mathfrak g^\cap_{\mathcal F_3}=\langle D(1)\rangle$.

Now we list the subclasses of~$\mathcal F_3$ with admissible-transformation extensions. 
For each of the subclasses, we present the corresponding form of arbitrary elements,  
the description of admissible transformations 
and the equivalent case of Table~\ref{tab:CompleteGroupClassificationBurgersKdVEquations}. 
Below $j=2,\dots,\EqOrd$; 
$a$'s and~$b$'s are constants parameterizing the arbitrary-element tuple 
within the subclass under consideration, 
and $c$'s are constants parameterizing admissible transformations.

\medskip 

\noindent
I. $A^j=a_jx^j|x|^{-\nu_3}$, $A^0=a_{00}+a_{01}|x|^{-\nu_3}$, $B=x(b_0+b_1|x|^{-\nu_3}+b_2|x|^{-2\nu_3})$
modulo shifts of~$x$, where $\nu_3a_\EqOrd\ne0$;
$a_{00}=0$ if $\nu_3=2$; $b_0=0$ if $a_{01}=0$; 
$a_{00}a_{01}=(\nu_3-2)b_1$, $(\nu_3-2)^2b_0=(\nu_3-1)a_{00}^2$, and $a_{01}^2b_0=(\nu_3-1)b_1^2$.
Then $X^0=0$ (due to the above shifts of~$x$), $|X^1|^{\nu_3}/T_t=:c_0=\const$, $c_0\ne0$,
$\tilde\nu_3=\nu_3$, $\tilde a_j=c_0a_j$, $\tilde a_{01}=c_0a_{01}$ and $\tilde b_2=c_0^2b_2$. 
$a_{\EqOrd}=1\bmod G^\sim_{\mathcal F_3}$.

\vspace{1ex}\par\noindent
a) If $(\nu_3,a_{01})\ne(2,0)$, then the expression for~$T$ 
is given by the general solution of the equation $T_{tt}/T_t=\lambda-\tilde\lambda T_t$ 
with the constants~$\lambda$ and~$\tilde\lambda$ defined by 
\[
\begin{pmatrix}a_{00}\\ \tilde a_{00}\end{pmatrix}=\dfrac{\nu_3-2}{\nu_3}\begin{pmatrix}\lambda\\ \tilde\lambda\end{pmatrix},
\quad
\begin{pmatrix}b_1\\ \tilde b_1\end{pmatrix}=\nu_3\begin{pmatrix}\lambda a_{01}\\ \tilde\lambda\tilde a_{01}\end{pmatrix}.
\]
(See~\cite{vane07Ay} for a proper representation of this solution.)
Since $a_{00},b_0,b_1=0\bmod\mathcal G^\sim_{\mathcal F_3}$, 
this case reduces to Case~\ref{002b} of Table~\ref{tab:CompleteGroupClassificationBurgersKdVEquations}. 

\vspace{1ex}\par\noindent
b) If $(\nu_3,a_{01})=(2,0)$, then $a_{00}=b_1=0$, and the expression for~$T$ 
is given by the general solution of the equation 
\[
\left(\frac{T_{tt}}{T_t}\right)_t-\frac{\nu_3-1}{\nu_3}\left(\frac{T_{tt}}{T_t}\right)^2=\nu_3\tilde b_0 T_t^2-\nu_3b_0.
\]
Since $b_0=0\bmod\mathcal G^\sim_{\mathcal F_3}$, 
this case reduces to Case~\ref{003} of Table~\ref{tab:CompleteGroupClassificationBurgersKdVEquations}. 

\medskip 

\noindent
II. $A^j=a_jx^j$, $A^0=a_{01}\ln|x|+a_{00}$, $B=x(b_0+b_1\ln|x|+b_2\ln^2|x|)$
modulo shifts of~$x$ with $a_\EqOrd\ne0$, 
$b_2=-\frac14a_{01}^2$ and $b_1=\frac14a_{01}^2-\frac12a_{00}a_{01}$. 
Then $X^0=0$ (due to the above shifts of~$x$), $T=c_1t+c_2$ with $c_1\ne0$, 
$\tilde a_j=c_1^{-1}a_j$, $\tilde a_{01}=c_1^{-1}a_{01}$, 
$\tilde b_0=b_0+\frac14c_1^{-2}(a_{01}-c_1\tilde a_{00}-a_{00})(c_1\tilde a_{00}-a_{00})$, 
and $X^1$ runs through the set of (nonzero) solutions of the equation
\[2\frac{X^1_t}{X^1}=a_{01}\ln|X^1|+c_1\tilde a_{00}-a_{00}.\]
If $a_{01}=0$, then $a_\EqOrd=1\bmod G^\sim_{\mathcal F_3}$ and $a_{00}=1\bmod\mathcal G^\sim_{\mathcal F_3}$, 
giving Case~\ref{011a} of Table~\ref{tab:CompleteGroupClassificationBurgersKdVEquations}. 
Otherwise, $a_{01}=1\bmod G^\sim_{\mathcal F_3}$ and $a_{00}=1\bmod\mathcal G^\sim_{\mathcal F_3}$, 
i.e., we obtain Case~\ref{011b} of Table~\ref{tab:CompleteGroupClassificationBurgersKdVEquations}. 
 
\medskip 

\noindent
III. $A^j=a_je^{-\nu_3 x}$, $A^0=a_{01}e^{-\nu_3 x}+a_{00}$, $B=b_0+b_1e^{-\nu_3 x}+b_2e^{-2\nu_3 x}$, 
where $a_\EqOrd\nu_3\ne0$, $b_0=\nu_3^{-1}a_{00}^2$ and $b_1=\nu_3^{-1}a_{00}a_{01}$.
Then  
$X^1=c_1\ne0$, $X^0=\nu_3^{-1}\ln(c_0T_t)$, $c_0T_t>0$, 
$\tilde\nu_3=c_1^{-1}\nu_3$,
$\tilde a_j=c_0c_1^ja_j$, $\tilde a_{01}=c_0a_{01}$, $\tilde b_2=c_0^2c_1b_2$,
the expression for~$T$ is given by the general solution of the equation $T_{tt}/T_t=\tilde a_{00}T_t-a_{00}$.
We can gauge $(\nu_3,a_\EqOrd)=(-1,1)\bmod G^\sim_{\mathcal F_3}$ and $a_{00}=0\bmod\mathcal G^\sim_{\mathcal F_3}$, 
which leads to Case~\ref{002a} of Table~\ref{tab:CompleteGroupClassificationBurgersKdVEquations}. 
One of the constants~$a_j$'s with $j<\EqOrd$, $a_0$ or~$b_2$, if it is nonzero, can be set to $\pm1$ 
by shifts of~$x$.
 
\medskip 

\noindent
IV. $A^j=a_j$, $A^0=a_0$, $B=b_1x+b_0$ with $a_\EqOrd\ne0$,  
where also $(\EqOrd-2)^2b_1\ne(\EqOrd-1)a_0^2$ or $a_j\ne0$ for some $j<\EqOrd$.
Then $T=c_1t+c_2$ with $c_1\ne0$, $X^1=c_3\ne0$, 
$\tilde a_j=c_1^{-1}c_3^ja_j$, $\tilde a_0=c_1^{-1}a_0$, $\tilde b_1=c_1^{-2}b_1$,
and $X^0$ runs through the solution set of the equation
$X^0_{tt}-a_0X^0_t-b_1X^0=c_1^2\tilde b_0-c_3b_0$.
Gauging $b_0=0\bmod\mathcal G^\sim_{\mathcal F_3}$ and $a_\EqOrd=1\bmod G^\sim_{\mathcal F_3}$, 
we reduce this case to Case~\ref{201} of Table~\ref{tab:CompleteGroupClassificationBurgersKdVEquations}. 
 
\medskip 

\noindent
V. $A^\EqOrd=a_\EqOrd\ne0$, $A^j=0$, $j=2,\dots,\EqOrd-1$, $A^0=a_0$, $B=b_1x+b_0$ 
with $(\EqOrd-2)^2b_1=(\EqOrd-1)a_0^2$. 
We have $(X^1)^\EqOrd/T_t=c_1\ne0$, $\tilde a_\EqOrd=c_1a_\EqOrd$, 
and the parameter function~$T$ runs through the solution set of either the equation 
\[
\frac{\EqOrd-2}\EqOrd\frac{T_{tt}}{T_t}=a_0-\tilde a_0T_t
\quad\mbox{or}\quad
\left(\frac{T_{tt}}{T_t}\right)_t-\frac{\EqOrd-1}{\EqOrd}\left(\frac{T_{tt}}{T_t}\right)^2-\tilde a_0T_{tt}=\EqOrd\tilde b_0 T_t^2-\EqOrd b_0
\]
if $\EqOrd>2$ or $\EqOrd=2$, respectively.  
The expression for~$X^0$ is found from the equation 
\[
\frac1{T_t}\left(\frac{X^0_t}{T_t}\right)_t-\tilde a_0\frac{X^0_t}{T_t}-\tilde b_1X^0=\tilde b_0-b_0\frac{X^1}{T_t{}^2}.
\]
For any allowed value of~$\EqOrd$, we have $a_0,b_0,b_1=0\bmod\mathcal G^\sim_{\mathcal F_3}$ 
and thus get Case~\ref{202} of Table~\ref{tab:CompleteGroupClassificationBurgersKdVEquations}.

\medskip 

There are no additional $\mathcal G^\sim_{\mathcal F_3}$-gauges of constants parameterizing the arbitrary-element tuple~$\kappa$.

By~$\mathcal F_{3.i}$ we denote the subclass of equations from the class~$\mathcal F_3$ 
whose maximal Lie invariance algebras are similar to the algebra presented in Case~$i$ 
of Table~\ref{tab:CompleteGroupClassificationBurgersKdVEquations}, 
$i\in\{\ref{001},\ref{002a},\ref{002b},\ref{003},\ref{011a},\ref{011b},\ref{201},\ref{202}\}$.
In this way, we partition the class~$\mathcal F_1$ into eight subclasses, 
that are associated with cases of Lie symmetry extensions in this class. 
The seven subclasses~$\mathcal F_{3.i}$, $i=\ref{002a},\ref{002b},\ref{003},\ref{011a},\ref{011b},\ref{201},\ref{202}$, 
can be shown to be normalized in the generalized sense, 
and the complement to their union, which is the subclasses~$\mathcal F_{3.\ref{001}}$, 
is normalized in the usual sense with respect to the usual equivalence group~$G^\sim_{\mathcal F_3}$ 
of the entire class~$\mathcal F_3$. 
There are no point transformations between equations 
from the subclasses~$\mathcal F_{3.i}$ and~$\mathcal F_{3.i'}$ if $i\ne i'$. 
In other words, the equivalence groupoid~$\mathcal G^\sim_{\mathcal F_3}$ of the class~$\mathcal F_3$ 
is the disjoint union of the equivalence groupoids of the normalized subclasses~$\mathcal F_{3.i}$'s. 

\begin{remark}\label{rem:OnGenCaseofGenBurgersKdVEqsWithSpace-DependentCoeffs}
The above consideration implies that 
the Lie symmetry extension given in Case~\ref{001} 
of Table~\ref{tab:CompleteGroupClassificationBurgersKdVEquations} 
is maximal if and only if the arbitrary-element tuple 
$\kappa=(A^0,A^2,\dots,A^\EqOrd,B)$ is not of one of the forms listed in Cases~I--V. 
\end{remark}

\section{Lie reductions and exact solutions}\label{sec:LieReductionsGenBurgersKdVEqs}

Since the class~\eqref{eq:GenBurgersKdVEqs}
is mapped onto its subclass~\eqref{eq:GenBurgersKdVEqsGaugedSubclass}
by a family of equivalence transformations and both these classes are normalized,
for studying Lie reductions of equations from the class~\eqref{eq:GenBurgersKdVEqs}
we can apply the technique suggested in~\cite{poch16a}
to the subclass~\eqref{eq:GenBurgersKdVEqsGaugedSubclass}.
Within the framework of the standard approach to Lie reduction,
one considers each case of Lie symmetry extensions separately,
lists subalgebras of the corresponding maximal Lie invariance algebra
that are inequivalent with respect to internal automorphisms of this algebra,
constructs ansatzes with them
and then uses these ansatzes to reduce relevant original equations
to equations with fewer independent variables.
For a normalized class of differential equations, 
one can classify Lie reductions of equations from the class
with respect to its equivalence group of respective kind, 
taking subalgebras in the projection of its equivalence algebra~\cite{poch16a}. 
Restrictions for appropriate subalgebras are relevant here 
except those restrictions related to the property to be maximal for some equations from the~class. 

A complete list of $G^\sim$-inequivalent one-dimensional subalgebras 
of the algebra~$\mathfrak g$ is exhausted by four subalgebras, 
$\langle D(1)\rangle$, $\langle S(1)\rangle$, $\langle S(e^t)\rangle$ 
and $\langle P(1)\rangle$, cf.\ Table~\ref{tab:CompleteGroupClassificationBurgersKdVEquations}. 
The subalgebra $\langle P(1)\rangle$ does not appear 
in Table~\ref{tab:CompleteGroupClassificationBurgersKdVEquations} 
since, in view of Lemma~\ref{lem:DimOfLieInvSubalgebraBurgersKdV1}, 
there is no equation from the class~\eqref{eq:GenBurgersKdVEqsGaugedSubclass} 
for which this algebra is the maximal Lie invariance algebra.
This subalgebra is associated with Case~\ref{200} of Table~\ref{tab:CompleteGroupClassificationBurgersKdVEquations}. 
We construct an ansatz for~$u$ with each of the listed one-dimensional subalgebras, 
select $G^\sim$-inequivalent equations from the subclass~\eqref{eq:GenBurgersKdVEqsGaugedSubclass} 
that are invariant with respect to this subalgebra 
and reduce them to ordinary differential equations.
In the course of selecting invariant equations, we use Table~\ref{tab:CompleteGroupClassificationBurgersKdVEquations}
but neglect the conditions of maximality of invariance algebras, which are collected in Remark~\ref{rem:ConditionsOfMaximality}.
As a result, we obtain the following $G^\sim$-inequivalent reductions to ordinary differential equations:
\begin{gather*}
\begin{array}{ll}
\mbox{Case~\ref{001}}, \ \langle D(1)  \rangle \colon& u=\varphi(\omega),          \quad \omega=x, \\ 
   &\varphi\varphi_\omega=\sum_{j=2}^\EqOrd A^j(\omega)\varphi^{(j)}+A^0(\omega)\varphi+B(\omega); \\[1ex]
\mbox{Case~\ref{010a}},\ \langle S(1)  \rangle \colon& u=\varphi(\omega)x,         \quad \omega=t, \\ 
   & \varphi_\omega+\varphi^2=\noprint{\alpha^0(\omega)\varphi+}\beta(\omega); \\[1ex]
\mbox{Case~\ref{010b}},\ \langle S(e^t)\rangle \colon& u=\varphi(\omega)x+x\ln|x|, \quad \omega=t, \\
   & \varphi_\omega+\varphi^2=\noprint{(\alpha^0(\omega)-1)\varphi+}\sum_{j=2}^\EqOrd(-1)^j(j-2)!\alpha^j(\omega)+\beta(\omega); \\[1ex]
\mbox{Case~\ref{200}}, \ \langle P(1)  \rangle \colon& u=\varphi(\omega),          \quad \omega=t, \\ 
   & \varphi_\omega=0\noprint{\alpha^0(\omega)\varphi+\beta(\omega)}.
\end{array}
\end{gather*}

Now we classify two-dimensional subalgebras of~$\mathfrak g$ 
that are subalgebras of appropriate subalgebras of~$\mathfrak g$ 
and carry out corresponding Lie reductions of equations 
from the class~\eqref{eq:GenBurgersKdVEqsGaugedSubclass} to algebraic equations. 
To obtain a complete list of such subalgebras up to $G^\sim$-equivalence, 
we realize the following algorithm: take two operators of the most general form, 
verify that they satisfy necessary conditions 
and then simplify successively their form until all cases are listed. 
Computations from the proof of Theorem~\ref{thm:GroupClassificationGenBurgersKdVEqs} can be used here.
At first, let two vector fields~$Q^1$ and~$Q^2$ be of the form 
$Q^i=D(\tau^i)+S(\zeta^i)+P(\chi^i)$, $i=1,2$, 
where the parameter functions $\tau^1$ and~$\tau^2$ are linearly independent. 
As has already been shown in the case $(k_1,k_2,k_3)=(0,0,2)$ of the above proof, 
this case is exhausted by the subalgebras 
$\langle D(1),D(t)-P(1)\rangle$ and $\langle D(1),D(t)-S(\nu^{-1})\rangle$, where $\nu\neq0$. 
Suppose that~$\tau^1$ and~$\tau^2$ are linearly dependent but are not simultaneously zero. 
Then, up to linearly re-combining~$Q^1$ and~$Q^2$, we can assume that $\tau^1\ne0$, $\tau^2=0$ 
and, moreover, the first basis vector field can be transformed to $Q^1=D(1)$. 
If $\zeta^2\neq0$ in the new~$Q^2$, the consideration 
is entirely reduced to the case $(k_1,k_2,k_3)=(0,1,1)$ of the proof of Theorem~\ref{thm:GroupClassificationGenBurgersKdVEqs}, 
which gives the subalgebras
$\langle D(1),S(1)\rangle$ and $\langle D(1),S(e^t)\rangle$. 
Otherwise we have two other inequivalent possibilities for~$Q^2$, $Q^2=P(1)$ and $Q^2=P(e^t)$.
The basis elements cannot be of the form $Q^i=S(\zeta^i)+P(\chi^i)$, $i=1,2$, with $(\zeta^1,\zeta^2)\ne(0,0)$
in view of Lemma~\ref{lem:DimOfLieInvSubalgebraBurgersKdV1}.  
And finally, the subalgebra spanned by $Q^i=P(\chi^i)$, $i=1,2$, 
with linearly independent parameter functions $\chi^1$ and~$\chi^2$ 
does not provide a Lie ansatz for~$u$. 
Therefore, $G^\sim$-inequivalent reductions of equations 
from the class~\eqref{eq:GenBurgersKdVEqsGaugedSubclass} 
to algebraic equations are exhausted by the following reductions:
\begin{gather*}
\begin{array}{llll}
\mbox{Case~\ref{002a}},           &\langle D(1),D(t){-}P(1)        \rangle\colon  &u=\varphi e^x,      &\varphi^2=(a_0+\sum_{j=2}^\EqOrd a_j)\varphi+b; \\[1ex]
\mbox{Case~\ref{002b}},           &\langle D(1),D(t){-}S(\frac1\nu)\rangle\colon  &u=\varphi x|x|^\nu, &\\&&\lefteqn{\textstyle(\nu+1)\varphi^2=(a_0+\sum_{j=2}^\EqOrd(\nu+1)^{\underline j}\,a_j)\varphi+b;} \\[1ex]
\mbox{Case~\ref{011a}},           &\langle D(1),S(1)               \rangle\colon  &u=\varphi x,        &\varphi^2=b\noprint{+a_0\varphi}; \\[1ex]
\mbox{Case~\ref{011b}},           &\langle D(1),S(e^t)             \rangle\colon  &u=\varphi x+x\ln|x|,&\varphi^2=\noprint{(a_0-1)\varphi+}\sum_{j=2}^\EqOrd(-1)^j(j{-}2)!a_j+b;\!\!\\[1ex]
\mbox{Case~\ref{201}}_{b=0},      &\langle D(1),P(1)               \rangle\colon  &u=\varphi,          &a_0\varphi\noprint{+b}=0;\\[1ex]
\mbox{Case~\ref{201}}_{a_0+b=1},\!&\langle D(1),P(e^t)             \rangle\colon  &u=\varphi+x,        &b\varphi=0. 
\end{array}
\end{gather*}
Here $(\nu+1)^{\underline j}=(\nu+1)\nu\cdots(\nu-j+3)(\nu-j+2)$ is the falling factorial. 

It is obvious that the Burgers, Korteweg--de Vries, Kuramoto--Sivashinsky and Kawahara equations 
admit solutions affine with respect to~$x$, 
which are of the form $u=(x+c_1)/(t+c_0)$ or $u=c_0$ with arbitrary constants~$c_0$ and~$c_1$.
This is also true for generalized Burgers equations of the form $u_t+uu_x=f(t,x)u_{xx}$, 
which is related to the fact that each of these equations is conditionally invariant 
with respect to the vector field $\p_t+u\p_x$ \cite{poch12a,poch16a} or, equivalently, 
with respect to the generalized vector field $u_{xx}\p_u$. 
The above Lie reductions with respect to the algebras $\langle S(1)\rangle$ and $\langle P(1)\rangle$ 
also solely give solutions affine in~$x$.
This is why we can look for such solutions of equations from a wide subclass of the class~\eqref{eq:GenBurgersKdVEqs}. 
More specifically, each equation~$\mathcal L_\theta$ of the form~\eqref{eq:GenBurgersKdVEqs} 
with $A^0_x=B_{xx}=C_x=0$, i.e., $A^0=\alpha^0(t)$, $B=\beta^1(t)x+\beta^0(t)$ and $C=\gamma(t)$,
is reduced by the ansatz $u=\varphi^1(t)x+\varphi^0(t)$
to a system of two ordinary differential equations for~$\varphi^1$ and~$\varphi^0$,
\begin{gather*}
\varphi^1_t+\gamma\varphi^1\varphi^1=\alpha^0\varphi^1+\beta^1,\\
\varphi^0_t+\gamma\varphi^1\varphi^0=\alpha^0\varphi^0+\beta^0.
\end{gather*}
This means that the equation~$\mathcal L_\theta$ possesses 
the generalized conditional symmetry $u_{xx}\p_u$, 
and the above ansatz just represents the general solution 
of the corresponding invariant surface condition $u_{xx}=0$; 
cf.~\cite{zhda95a} as well as~\cite{kunz11a}. 
If $A^0=B=0$ and $C=1$, then solutions constructed with the ansatz affine in~$x$ 
are of the same form as for the Burgers and Korteweg--de Vries equations.

All the solutions constructed here are quite simple. 
Perhaps the most direct way to construct nontrivial solutions 
for variable-coefficient equations from the class~\eqref{eq:GenBurgersKdVEqs} 
is to generate them by equivalence transformations 
from wide families of solutions known for famous (constant-coefficient) equations 
from the same class, like the Korteweg--de Vries equation. 
This way is often ignored in the literature, 
which is comprehensively discussed in~\cite{popo10By}.

\section{Conclusion}\label{sec:ConclusionGenBurgersKdVEqs}

The class~\eqref{eq:GenBurgersKdVEqs} is much wider than 
the classes of variable-coefficient Korteweg--de Vries and Burgers equations, 
which were considered in~\cite{gaze92a} and~\cite{qu95a}
and are the subclasses of~\eqref{eq:GenBurgersKdVEqs} for~$\EqOrd=3$ and~$\EqOrd=2$ 
with the constraints $A^0=A^1=A^2=B=0$ and $A^0=A^1=B=0$, respectively. 
At the same time, we have carried out 
the complete group classification of the class~\eqref{eq:GenBurgersKdVEqs} 
with less efforts than the ones made in~\cite{gaze92a,qu95a} 
for the group classification of the above particular subclasses, 
and the classification list is much more compact 
for the class~\eqref{eq:GenBurgersKdVEqs}, 
although the equation order~$r$ is not fixed here. 
This became possible due to using several ingredients. 

The first ingredient was the proper choice of the class to be studied. 
In general, the choice should be realized in such a way 
for the selected class to be convenient for group classification, 
with the main criterion for this being the property of normalization. 
In contrast to the classes considered in~\cite{gaze92a} and~\cite{qu95a}, 
the class~\eqref{eq:GenBurgersKdVEqs} is normalized in the usual sense. 

Since there are two arbitrary smooth functions of~$(t,x)$ among the parameters 
of the usual equivalence group of the class~\eqref{eq:GenBurgersKdVEqs} 
and arbitrary elements of this class also depend on~$(t,x)$ only, 
two of the arbitrary elements can be gauged by equivalence transformations. 
It seems at first sight that there are several equally appropriate gauges 
for arbitrary elements of the class~\eqref{eq:GenBurgersKdVEqs} 
but this is not the case. 
It has been shown in Section~\ref{sec:AlternativeGauges} 
that although the gauge $A^\EqOrd=1$ is quite natural 
and gives a subclass normalized in the usual sense, 
the structure of the corresponding usual equivalence group is more complicated 
than for the gauge $C=1$, 
and the further possible gauge $A^1=0$ leads to a subclass 
that is normalized solely in the generalized sense. 
The equivalence groupoid of the subclass 
associated with the other secondary gauge $A^0=0$, 
which can also be realized after gauging $A^\EqOrd=1$, 
has even much more involved normalization properties. 
This is why it is essential to select the gauge of arbitrary elements
that is optimal for group classification of the class~\eqref{eq:GenBurgersKdVEqs}. 
The careful analysis has resulted in the selection of the initial gauge $C=1$ and the secondary gauge $A^1=0$, 
and this selection of, especially, the secondary gauge is quite unforeseen 
without having the complete description of the equivalence groupoid of the class~\eqref{eq:GenBurgersKdVEqs}.
Both the subclasses~\eqref{eq:GenBurgersKdVEqsGaugeSubclassC1} and~\eqref{eq:GenBurgersKdVEqsGaugedSubclass}, 
which are associated with the gauges $C=1$ and $(C,A^1)=(1,0)$, respectively, 
are normalized in the usual sense. 
Due to the normalization, we easily control the changes of equivalence groups in the course of gauging. 
The group classifications of the class~\eqref{eq:GenBurgersKdVEqs} 
and of its subclass~\eqref{eq:GenBurgersKdVEqsGaugedSubclass} are equivalent, 
cf.\ Proposition~\ref{pro:GroupClassificationsOfGenBurgersKdVEqsAndGaugedGenBurgersKdVEqs}. 
Moreover, the usual equivalence group of the maximally gauged subclass~\eqref{eq:GenBurgersKdVEqsGaugedSubclass} 
is parameterized by just three functions of the only argument~$t$. 
This fact jointly with the normalization property makes it convenient 
to solve the group classification problem for the subclass~\eqref{eq:GenBurgersKdVEqsGaugedSubclass} 
using the algebraic method.

The final ingredient for the efficient solution of the group classification problem 
was the accurate study of properties of the maximal Lie invariance algebras of equations  
from the subclass~\eqref{eq:GenBurgersKdVEqsGaugedSubclass}. 
As a result, we have found quite strong restrictions on appropriate subalgebras 
of the equivalence algebra~$\mathfrak g^\sim$ (more precisely, its projection~$\mathfrak g_\spanindex$) 
of the subclass~\eqref{eq:GenBurgersKdVEqsGaugedSubclass}, 
which has simplified the classification of such subalgebras.  
In particular, we have proved in Lemma~\ref{lem:DimOfLieInvAlgebraBurgersKdV}
that the dimension of maximal Lie invariance algebras 
of equations from the subclass~\eqref{eq:GenBurgersKdVEqsGaugedSubclass}
is less than or equal to five. 
We have furthermore introduced three invariant integer parameters 
$k_1$, $k_2$ and $k_3$
characterizing the dimensions of relevant subspaces of such algebras 
and found essential low-dimensional restrictions for the values 
that the parameters~$k$'s can assume.
The conditions for subspaces related to time transformations 
that are presented in Lemma~\ref{lem:OnOperatorsInvolvingTau} 
are highly common for evolution equations and their systems, 
cf.\ \cite{bihl16a,kuru16a}.
In contrast to this, the restriction for values of the pair of parameters $(k_1,k_2)$
that is given in Lemma~\ref{lem:DimOfLieInvSubalgebraBurgersKdV1} 
is specific for the subclass~\eqref{eq:GenBurgersKdVEqsGaugedSubclass} 
but is still important for effectively classifying appropriate subalgebras of~$\mathfrak g_\spanindex$.
The restrictions for~$k$'s
collected in Section~\ref{sec:PropertiesOfAppropriateSubalgebrasGenBurgersKdVEqs} 
are in total not complete, 
i.e., there are values of~$k$'s among selected ones 
that are associated with no maximal Lie invariance algebras of equations from  
the subclass~\eqref{eq:GenBurgersKdVEqsGaugedSubclass}; 
see Remark~\ref{rem:OnCompletenessOfRestrictionsForAppropriateSubalgebras}. 
At the same time, finding the exhaustive set of restrictions for~$k$'s
in fact needs carrying out the complete group classification 
of the subclass~\eqref{eq:GenBurgersKdVEqsGaugedSubclass}. 
In this sense, the present classification is more similar to 
the group classification of linear Schr\"odinger equations in~\cite{kuru16a} 
than the group classification of linear evolution equations in~\cite{bihl16a}, 
where all selected potential values for analogous invariant integer parameters 
are admitted by some appropriate subalgebras. 

Due to the normalization of the class~\eqref{eq:GenBurgersKdVEqs} in the usual sense
(resp.\ the subclass~\eqref{eq:GenBurgersKdVEqsGaugedSubclass}), 
we were able to classify Lie reductions of equations from the class 
with respect to its usual equivalence group 
using the technique proposed in~\cite{poch16a} and modified in~\cite{bihl16a}. 
We have constructed solutions 
that are at most affine in~$x$ to equations from a quite wide subclass of the class~\eqref{eq:GenBurgersKdVEqs}
as well as have discussed generating new solutions from known ones by equivalence transformations. 

In spite of the comprehensive symmetry analysis of the class~\eqref{eq:GenBurgersKdVEqs} and a number of its subclasses 
in the present paper, 
it is pertinent to carry out group classification 
of other subclasses of the class~\eqref{eq:GenBurgersKdVEqs} even for particular small values of~$n$. 
We plan to solve the group classification problem for 
the class of variable-coefficient Burgers equations of the form $u_t+f(t,x)uu_x+g(t,x)u_{xx}=0$ with $fg\ne0$,
which was considered in~\cite{qu95a} (see a discussion in the introduction of the present paper).
Although this class looks simple, its group classification is quite tricky 
and is supposed to involve partition into subclasses jointly with successively arranging and classifying each part separately.  

Although initially the main purpose of the paper had been 
the exhaustive solution of the group classification problem 
for the class~\eqref{eq:GenBurgersKdVEqs} of general variable-coefficient Burgers--KdV equations of arbitrary fixed order, 
the study of the equivalence groupoids of subclasses of this class led to more important results, which are worth recalling. 
For a long time after the first discussion of the notion of generalized equivalence groups in~\cite{mele94Ay,mele96Ay}, 
no examples of nontrivial generalized equivalence groups were known in the literature, 
except classes for which some of arbitrary elements are constants 
and thus some of components of equivalence transformations associated with system variables 
depend on such arbitrary elements; see, e.g., \cite[Section~6.4]{popo10Ay}, \cite[Section~2]{vane07Ay} and~\cite[Section~3]{vane12Ay}. 
Note that in all these papers, 
effective generalized equivalence groups were given instead of the corresponding generalized equivalence groups.
This is why certain doubts started to circulate in the symmetry community 
whether this notion is valuable at all.   
In the present paper we have happened to construct for the first time 
several examples of nontrivial generalized equivalence groups 
such that equivalence-transformation components corresponding to equation variables 
locally depend on nonconstant arbitrary elements of the corresponding classes. 
All related classes are (reparameterized) subclasses of the class~\eqref{eq:GenBurgersKdVEqs}.
The most significant consequence of the construction of these examples is that 
they make evident the necessity of introducing the notion of effective generalized equivalence group. 
Moreover, they also answer, just by their existence, some theoretical questions, 
which leads to properly posing further questions. 
In particular, the entire generalized equivalence group of a class may be effective itself 
and thus it is a unique effective generalized equivalence group of this class, 
see Remark~\ref{rem:GenEquivGroupOfGenBurgersKdVEqsGaugeAr1A10}. 
Nevertheless, there are classes of differential equations admitting 
multiple effective generalized equivalence groups. 
As discussed in Remark~\ref{rem:UniquenessOfEffectiveGenEquivGroups}, 
this claim is exemplified by classes~$\bar{\mathcal K}_2$ and~$\bar{\mathcal K}_3$, 
for which we have constructed effective generalized equivalence groups 
that are proper but not normal subgroups of the corresponding generalized equivalence groups. 
All known examples of generalized equivalence groups that are related to constant arbitrary elements 
have the same property. 
Then the natural question is whether there exists a class of differential equations 
with effective generalized equivalence group 
being a proper normal subgroup of the corresponding generalized equivalence group. 
By the way, Corollary~\ref{cor:GenBurgersKdVEqsExtGenEquivGroupOfSubclassK3} shows 
that even merely singling out an effective generalized equivalence group 
from the already known generalized equivalence group of a class 
may be a nontrivial problem. 

For the class~$\mathcal K_3$ of general Burgers--KdV equations 
with coefficients depending at most on the time variable, 
which is normalized in the extended generalized sense, 
we have explicitly constructed its extended generalized equivalence group in a rigorous way. 
We have reparameterized the class~$\mathcal K_3$ introducing virtual arbitrary elements 
that are nonlocally related to the native arbitrary elements of this class. 
This is the first example of such a construction for partial differential equations in the literature.  
Similar results were earlier obtained only for classes of linear ordinary differential equations
in the preprint version of~\cite{boyk15a}. 
The reparameterization technique developed in the present paper 
gives hope to us that such construction will be realized soon for many classes of differential equations. 

The equivalence groupoid of the class~$\mathcal F_1$ of general Burgers--KdV equations 
with coefficients depending at most on~$x$ has an elegant structure 
but this structure seems uncommon in comparison with various kinds of normalization. 
Two equations from the class~$\mathcal F_1$ are related by a point transformation~$\varphi$ 
that is not the projection of a transformation from the corresponding usual equivalence group
if and only if these equations belong to a case of Lie symmetry extension for this class. 
At the same time, the transformation~$\varphi$ is in general not decomposed into 
the projection of an equivalence transformation and a Lie symmetry transformation, 
and thus the class~$\mathcal F_1$ is even not semi-normalized.  
Hence no wonder that this example is unique.

\section*{Acknowledgments}

The authors are grateful to Vyacheslav Boyko, Christodoulos Sophocleous and Olena Vaneeva 
for useful discussions and interesting comments. 
This research was undertaken, in part, thanks to funding 
from the Canada Research Chairs program and the NSERC Discovery Grant program.
The research of ROP was supported by the Austrian Science Fund (FWF), project P25064, 
and by ``Project for fostering collaboration in science, research and education'' 
funded by the Moravian-Silesian Region, Czech Republic.

\footnotesize\setlength{\itemsep}{0.3ex}\frenchspacing

\end{document}